\RequirePackage{fix-cm}

\documentclass[smallextended]{svjour3}

\smartqed

\usepackage{graphicx}

\usepackage{amsmath}
\usepackage{amssymb}
\usepackage{braket}
\usepackage{color}

\newtheorem{thm}{Theorem}

\newtheorem{lem}{Lemma}

\begin{document}

\title{Limit distribution of a time-dependent quantum walk on the half line
}
\subtitle{}


\author{Takuya Machida}

\authorrunning{T.~Machida} 

\institute{%
T.~Machida \at
              College of Industrial Technology, Nihon University, Narashino, Chiba 275-8576, Japan\\
              \email{machida.takuya@nihon-u.ac.jp}
}

\date{}

\maketitle

\begin{abstract}
We focus on a 2-period time-dependent quantum walk on the half line in this paper.
The quantum walker launches at the edge of the half line in a localized superposition state and its time evolution is carried out with two unitary operations which are alternately cast to the quantum walk.
As a result, long-time limit finding probabilities of the quantum walk turn to be determined by either one of the two operations, but not both.
More interestingly, the limit finding probabilities are independent from the localized initial state.
We will approach the appreciated features via a quantum walk on the line which is able to reproduce the time-dependent walk on the half line.
\end{abstract}

\section{Introduction}

As one of the quantum counterparts of random walks, quantum walks have been investigated since around 2000 and many features of them have been discovered.
The most appreciated feature is that probability distributions of the quantum walks show ballistic spread as their time evolutions are promoting.
Also the probability distributions are not similar to the Gauss distributions which are known for the approximate distributions of random walks.
The intriguing features have been applied for quantum algorithms and turned out to prove that the algorithms can perform quadratic speed-up~\cite{Venegas-Andraca2012}.

In this paper we focus on a quantum walk on the half line and attempt to get its long-time limit distributions.
This study is motivated on long-time limit distributions because they describe how the quantum walkers behave after long-time, and the limit distributions for time-dependent quantum walks have not been derived.
The walker moves around the locations represented by the set of non-negative integers $\left\{0,1,2,\ldots\right\}$.
Limit distributions were analyzed for several quantum walks on the half line in the past studies~\cite{KonnoSegawa2011,LiuPetulante2013,Machida2016}.
While we take care of a time-dependent quantum walk in this study, the past researches on quantum walks on the half line were all for time-independent walks.
Konno and Segawa~\cite{KonnoSegawa2011} found long-time limit measures of two types of quantum walk on the half line.
Each type had a large mass in distribution and the mass was described as localization.
The presence of localization allowed them to derive limit measures with which where the quantum walker was observed after its unitary evolution ran infinite times.
On the other hand, a limit theorem on a rescaled space by time was proved by Liu and Petulante~\cite{LiuPetulante2013}.
The limit theorem depicted the approximate and global shape of the probability distribution of a quantum walk on the half line.
Machida~\cite{Machida2016} discovered a relation between a quantum walk on the half line and a quantum walk on the line which held the infinite locations represented by the set of integers $\mathbb{Z}=\left\{0, \pm 1, \pm 2,\ldots\right\}$.
From the result, exact representations for the probability distributions of the quantum walk on the half line were revealed, and their limit distributions were computed by Fourier analysis.

As mentioned earlier, we will study a time-dependent quantum walk and aim at its long-time limit distributions.
Starting with a numerical study~\cite{RibeiroMilmanMosseri2004}, time-dependent quantum walks were numerically examined~(e.g.~\cite{BanulsNavarretePerezRoldanSoriano2006,Romanelli2009}) and theoretically analyzed~\cite{MachidaKonno2010,Machida2011,IdeKonnoMachidaSegawa2011,Machida2013b,GrunbaumMachida2015}.
Particularly, a long-time limit distribution of a 2-period time-dependent quantum walk on the line was investigated~\cite{MachidaKonno2010}, and five years after the paper was published, the same kind of limit distribution of a 3-period time-dependent quantum walk on the line came out~\cite{GrunbaumMachida2015}.  
Both quantum walks had specific features and their limit distributions completely reproduced the features. 
With the methods appearing in the papers~\cite{Machida2016,MachidaKonno2010}, we will approach long-time limit distributions of the quantum walk on the half line in this paper.

In the subsequent section, we start off with the definition of the 2-period time-dependent quantum walk on the half line in which the initial state of the walker localizes.
The dependency on time is expressed in the dynamics of the walk by the alternate usage of two unitary operations.
In the same section, we introduce a 2-period time-dependent quantum walk on the line.
If the quantum walk on the line launches with a suitable delocalized initial state, it reproduces all the information of the time-dependent quantum walk on the half line.
Using the limit distributions of the 2-period time-dependent quantum walk on the line, we approach the limit distributions of the time-dependent quantum walk on the half line.
In Sec.~\ref{sec:summary}, this study will be wrapped up along with discussion.
We also see exact representations for probability distributions of a time-independent quantum walk on the half line in Appendix.
Managing a past study~\cite{Konno2002a} and one of the results in this paper, one may lead to the representations.

\section{2-period time-dependent quantum walk on the half line and quantum walk on the line}
\label{sec:HL_QW}

Time-dependent quantum walks are defined as the walks whose unitary operations update in parallel with the evolution of their systems.
In this section we first define the system of quantum walk on the half line and then update it with two unitary operations.
Since the unitary operations are used in temporally alternate shifts, let us call the walk a 2-period time-dependent quantum walk on the half line. 

The quantum walk can be expressed on a linear system and it is given as a tensor of two Hilbert spaces.   
One of the spaces represents the positions of the walker and the other represents the inner states which are interpreted as spin states in terms of physics.
The position space, represented by $\mathcal{H}_p^{HL}$, is spanned by the orthonormal basis $\left\{\ket{x} : x\in\left\{0,1,2,\ldots\right\}\right\}$, and the inner state space, represented by $\mathcal{H}_c$, is spanned by the orthonormal basis $\left\{\ket{0}, \ket{1}\right\}$.
As defined right now, the quantum walker has two inner states $0$ and $1$, expressed as $\ket{0}$ and $\ket{1}$ respectively on the Hilbert space $\mathcal{H}_c$, and they are physically considered as the down-spin state and the up-spin state.
That is, keeping a superposition of two inner states, the walker exists on the half line indicated by the set of non-negative integers $\left\{0,1,2,\ldots\right\}$.  

Now that the system of quantum walk has been defined, let us describe the system at time $t\,(=0,1,2,\ldots)$ by $\ket{\Psi_t}\in\mathcal{H}_p^{HL}\otimes\mathcal{H}_c$ and give an evolution to it.
The evolution is determined dependently on whether the time of the system is even or odd,
\begin{equation}
 \ket{\Psi_{t+1}}=\left\{\begin{array}{ll}
	       \tilde{S}^{HL}\tilde{C_1}^{HL}\ket{\Psi_t} & (t=0,2,4,\ldots),\\
			  \tilde{S}^{HL}\tilde{C_2}^{HL}\ket{\Psi_t} & (t=1,3,5,\ldots),\\
		     \end{array}\right.\label{eq:HL_time_evolution}\
\end{equation}
where
\begin{align}
 \tilde{C_1}^{HL}=&\sum_{x=0}^\infty\ket{x}\bra{x}\otimes C_1,\\
 \tilde{C_2}^{HL}=&\sum_{x=0}^\infty\ket{x}\bra{x}\otimes C_2,\\
 \tilde{S}^{HL}=&\ket{0}\bra{0}\otimes\ket{1}\bra{0}+\sum_{x=1}^\infty\ket{x-1}\bra{x}\otimes\ket{0}\bra{0}+\sum_{x=0}^\infty\ket{x+1}\bra{x}\otimes\ket{1}\bra{1},\label{eq:HL_shift_operator}
\end{align}
with
\begin{align}
 C_1=&\cos\theta_1\ket{0}\bra{0}+\sin\theta_1\ket{0}\bra{1}+\sin\theta_1\ket{1}\bra{0}-\cos\theta_1\ket{1}\bra{1}\nonumber\\
 =&c_1\ket{0}\bra{0}+s_1\ket{0}\bra{1}+s_1\ket{1}\bra{0}-c_1\ket{1}\bra{1},\label{eq:coin-flip_operator_1}\\[2mm]
 C_2=&\cos\theta_2\ket{0}\bra{0}+\sin\theta_2\ket{0}\bra{1}+\sin\theta_2\ket{1}\bra{0}-\cos\theta_2\ket{1}\bra{1}\nonumber\\
 =&c_2\ket{0}\bra{0}+s_2\ket{0}\bra{1}+s_2\ket{1}\bra{0}-c_2\ket{1}\bra{1}.\label{eq:coin-flip_operator_2}
\end{align}
in which $\cos\theta_1, \sin\theta_1, \cos\theta_2$, and $\sin\theta_2$ have been shortly written as $c_1, s_1, c_2$, and $s_2$ respectively.
The values of parameters $\theta_1$ and $\theta_2$ are supposed to stay in the interval $[0,2\pi)$.
In this study, the quantum walker is set in a localized state at time $0$,
\begin{equation}
 \ket{\Psi_0}=\ket{0}\otimes \left(\alpha\ket{0}+\beta\ket{1}\right),\label{eq:HL_initial_state}
\end{equation}
with $|\alpha|^2+|\beta|^2=1\,(\alpha,\beta\in\mathbb{C})$, and the system iterates Eq.~\eqref{eq:HL_time_evolution}.
The letter $\mathbb{C}$ denotes the set of complex numbers.

The study of quantum walks normally aims at finding where the walkers exist after their updates, and their positions are observed with probability laws.
The finding probability of the walker with inner state $j\in\left\{0,1\right\}$ at position $x\in\left\{0,1,2,\ldots\right\}$ is given by
\begin{equation}
 \mathbb{P}(X_t^{HL}=x;j)=\bra{\Psi_t}\Bigl(\ket{x}\bra{x}\otimes\ket{j}\bra{j}\Bigr)\ket{\Psi_t},\label{eq:HL_probability_inner_state}
\end{equation}
where $X_t^{HL}$ indicates the position of the walker on the half line.
If we observe just the position without considering the inner states, the finding probability should be defined as the sum of the finding probabilities based on Eq.~\eqref{eq:HL_probability_inner_state},
\begin{equation}
 \mathbb{P}(X_t^{HL}=x)=\sum_{j=0}^1 \mathbb{P}(X_t^{HL}=x;j)=\bra{\Psi_t}\biggl(\ket{x}\bra{x}\otimes\sum_{j=0}^1\ket{j}\bra{j}\biggr)\ket{\Psi_t}.\label{eq:HL_probability}
\end{equation}
These finding probabilities output Figs.~\ref{fig:160720_01}, \ref{fig:160718_04}, \ref{fig:160718_07}, and \ref{fig:160718_10}.
Viewing Figs.~\ref{fig:160720_01} and \ref{fig:160718_04}, we can get the representative features of quantum walks, that is, the distributions hold a sharp peak and spread out in proportion to time $t$.
Figures~\ref{fig:160718_07} and \ref{fig:160718_10} should be appreciated because there seems to be a region where the distributions are determined by either $\theta_1$ or $\theta_2$, but not both.
For instance, seeing Fig.~\ref{fig:160718_07}-(a), we realize that the distribution seems to be independent from $\theta_1$ as long as the parameter $\theta_1$ picks a value from the region $[0,\pi/4]\cup [3\pi/4,5\pi/4]\cup [7\pi/4,2\pi)$.
That fact will indeed make sense when the limit distributions show up in Theorem~\ref{th:limit}.

\begin{figure}[h]
\begin{center}
 \begin{minipage}{35mm}
  \begin{center}
   \includegraphics[scale=0.3]{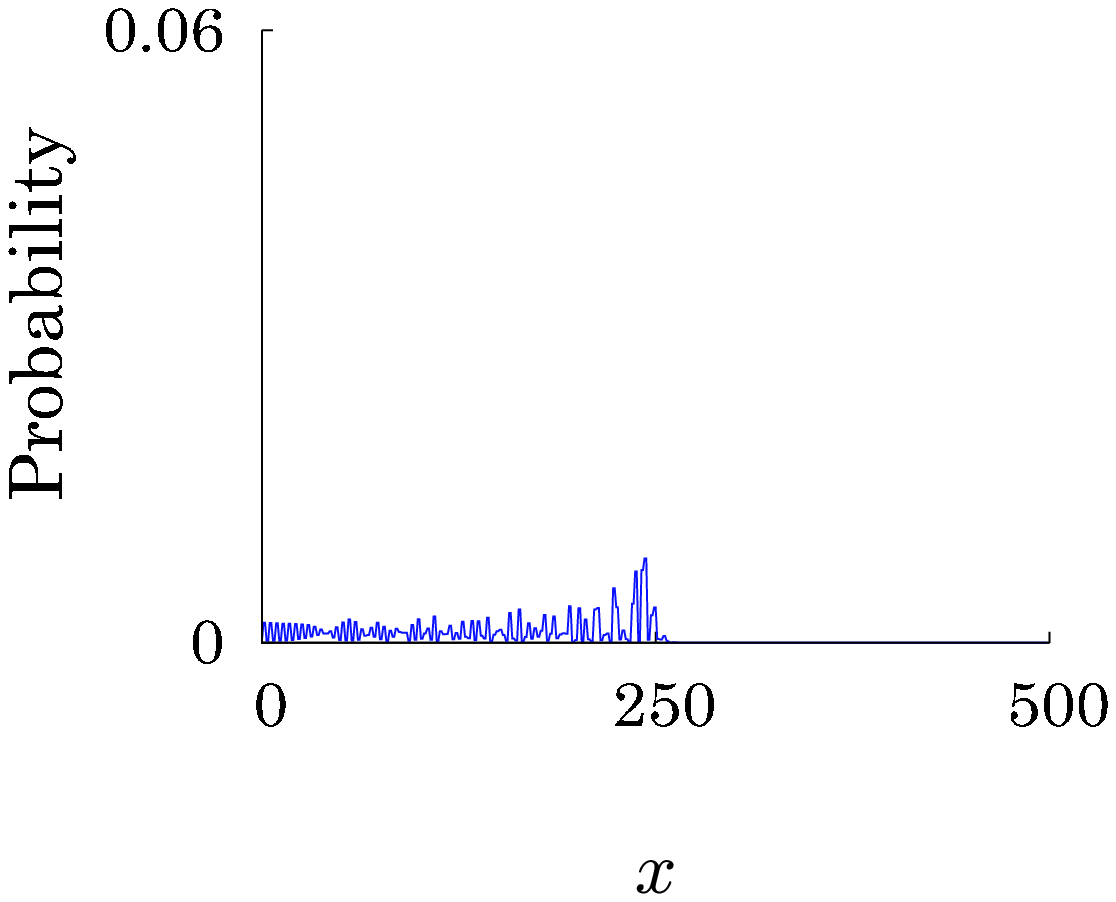}\\[2mm]
  (a) $\mathbb{P}(X_{500}^{HL}=x;0)$
  \end{center}
 \end{minipage}
 \begin{minipage}{35mm}
  \begin{center}
   \includegraphics[scale=0.3]{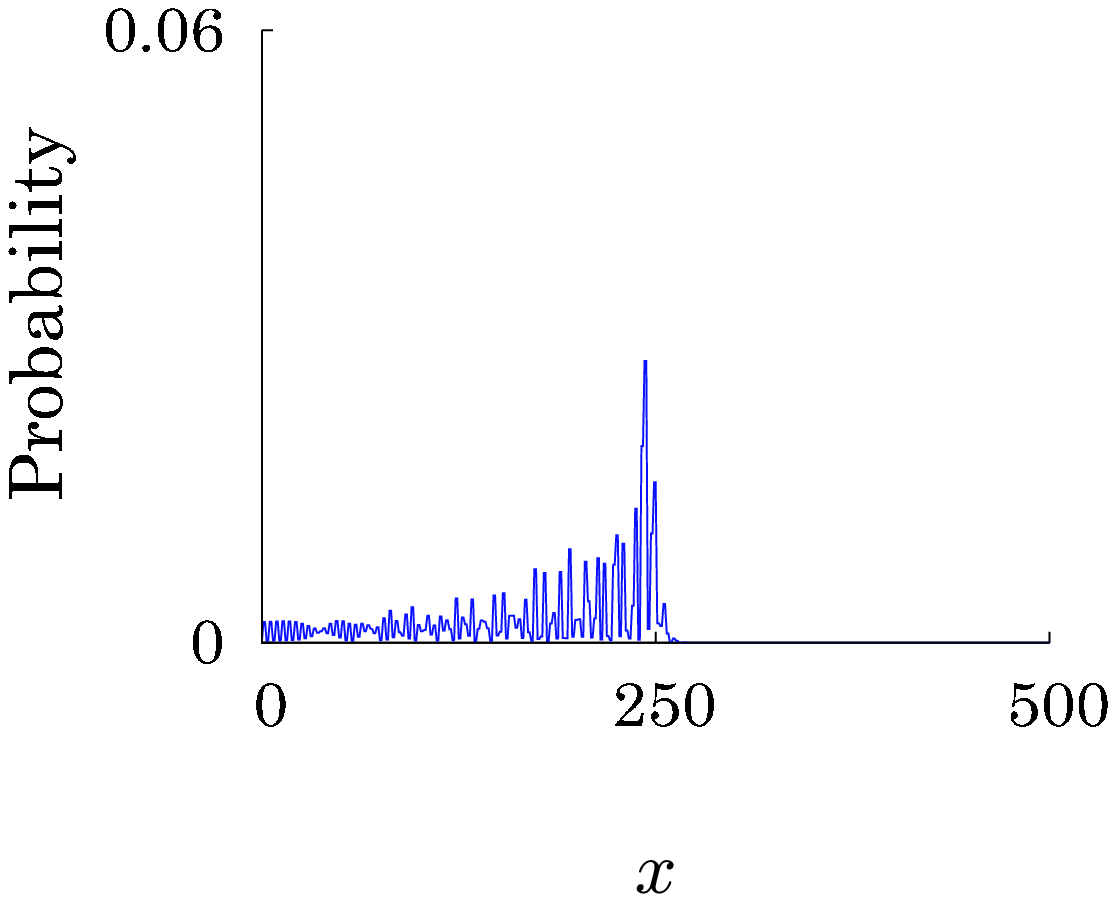}\\[2mm]
  (b) $\mathbb{P}(X_{500}^{HL}=x;1)$
  \end{center}
 \end{minipage}
 \begin{minipage}{35mm}
  \begin{center}
   \includegraphics[scale=0.3]{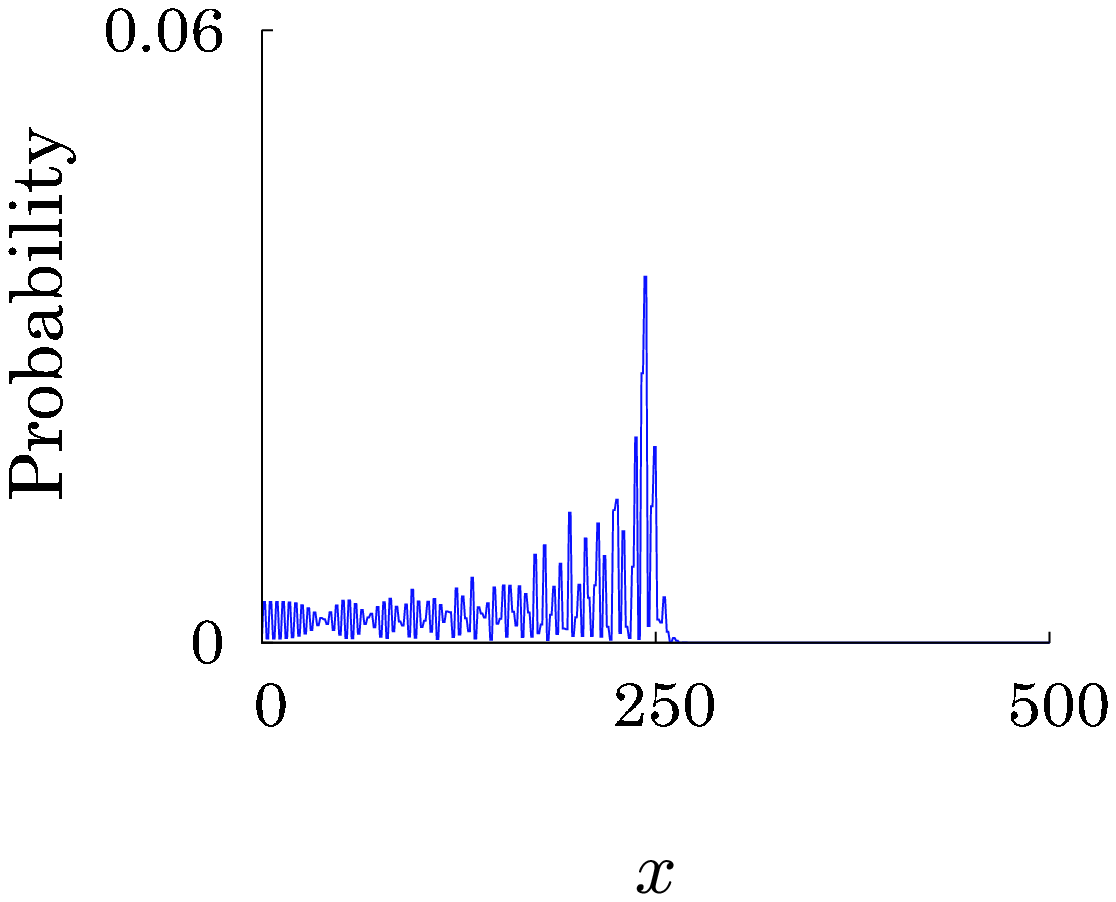}\\[2mm]
  (c) $\mathbb{P}(X_{500}^{HL}=x)$
  \end{center}
 \end{minipage}
\vspace{5mm}
\caption{$\theta_1=\pi/3,\,\theta_2=\pi/4$ : After the quantum walker has iterated its update $500$ times, we get the finding probabilities shown in these pictures. ($\alpha=1/\sqrt{2},\,\beta=i/\sqrt{2}$)}
\label{fig:160720_01}
\end{center}
\end{figure}

\begin{figure}[h]
\begin{center}
 \begin{minipage}{35mm}
  \begin{center}
   \includegraphics[scale=0.2]{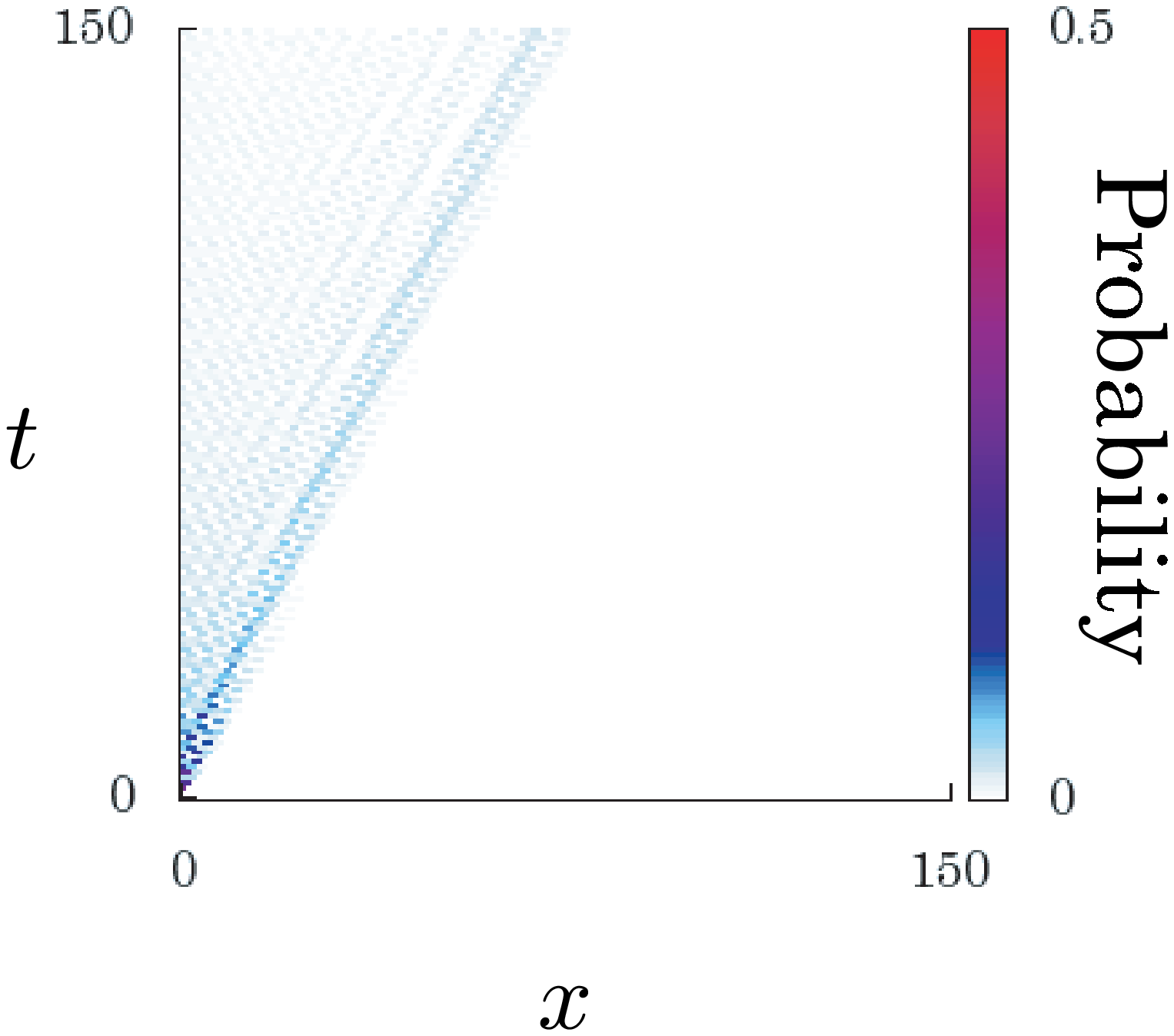}\\[2mm]
  (a) $\mathbb{P}(X_t^{HL}=x;0)$
  \end{center}
 \end{minipage}
 \begin{minipage}{35mm}
  \begin{center}
   \includegraphics[scale=0.2]{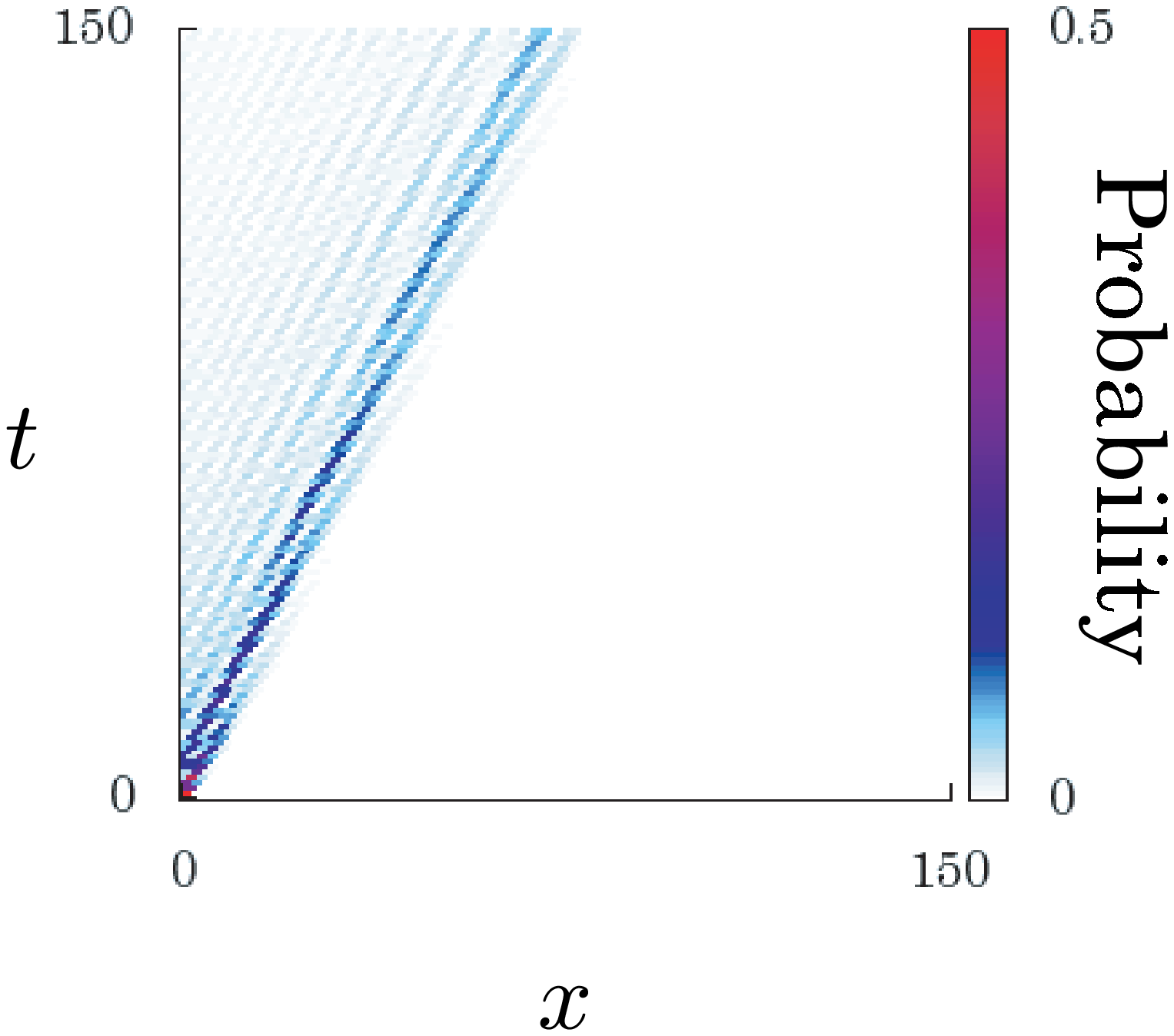}\\[2mm]
  (b) $\mathbb{P}(X_t^{HL}=x;1)$
  \end{center}
 \end{minipage}
 \begin{minipage}{35mm}
  \begin{center}
   \includegraphics[scale=0.2]{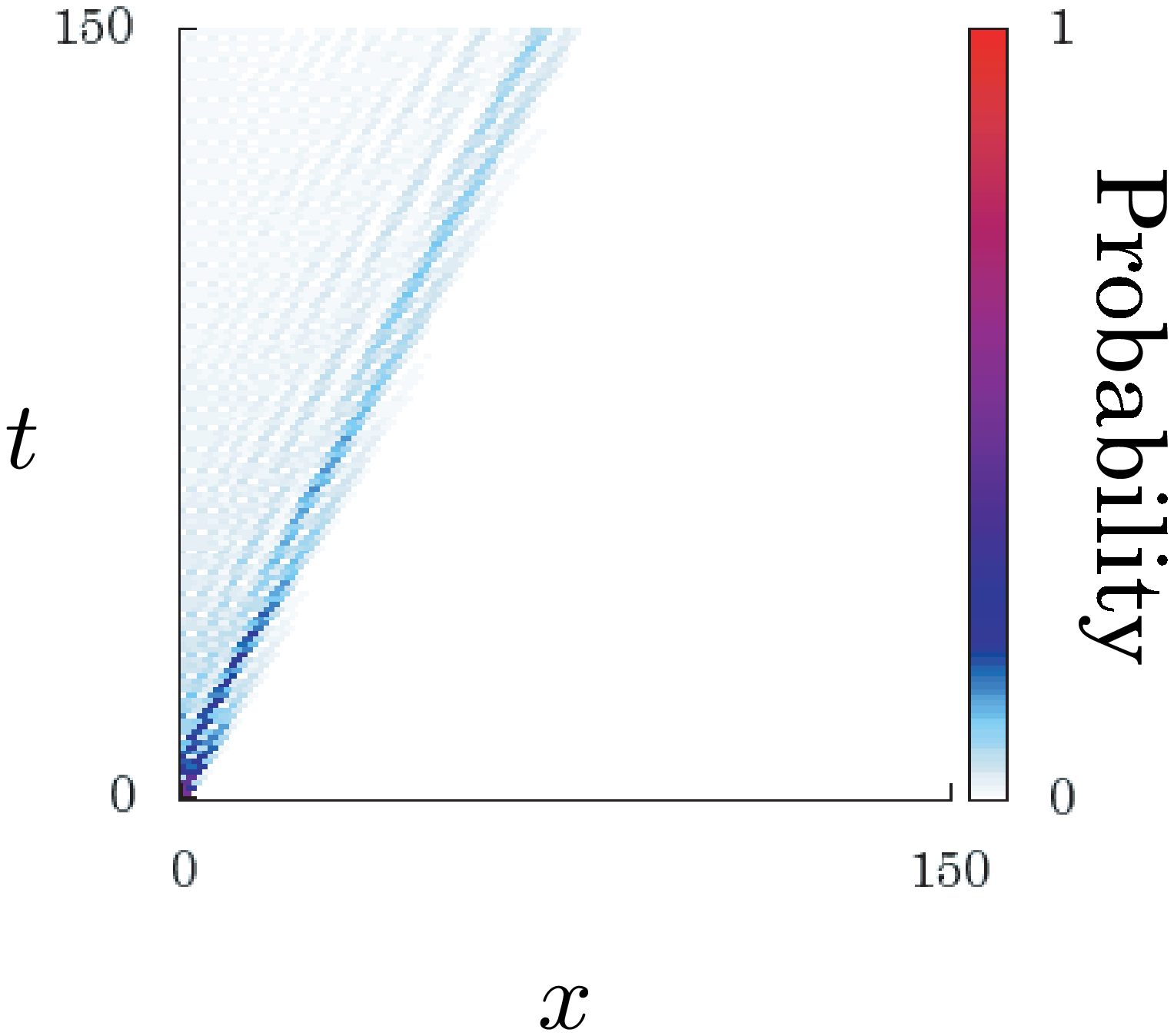}\\[2mm]
  (c) $\mathbb{P}(X_t^{HL}=x)$
  \end{center}
 \end{minipage}
\vspace{5mm}
\caption{$\theta_1=\pi/3,\,\theta_2=\pi/4$ : The distributions spread out in proportion to time $t$ as the quantum walk is getting updated. ($\alpha=1/\sqrt{2},\,\beta=i/\sqrt{2}$)}
\label{fig:160718_04}
\end{center}
\end{figure}

\begin{figure}[h]
\begin{center}
 \begin{minipage}{35mm}
  \begin{center}
   \includegraphics[scale=0.2]{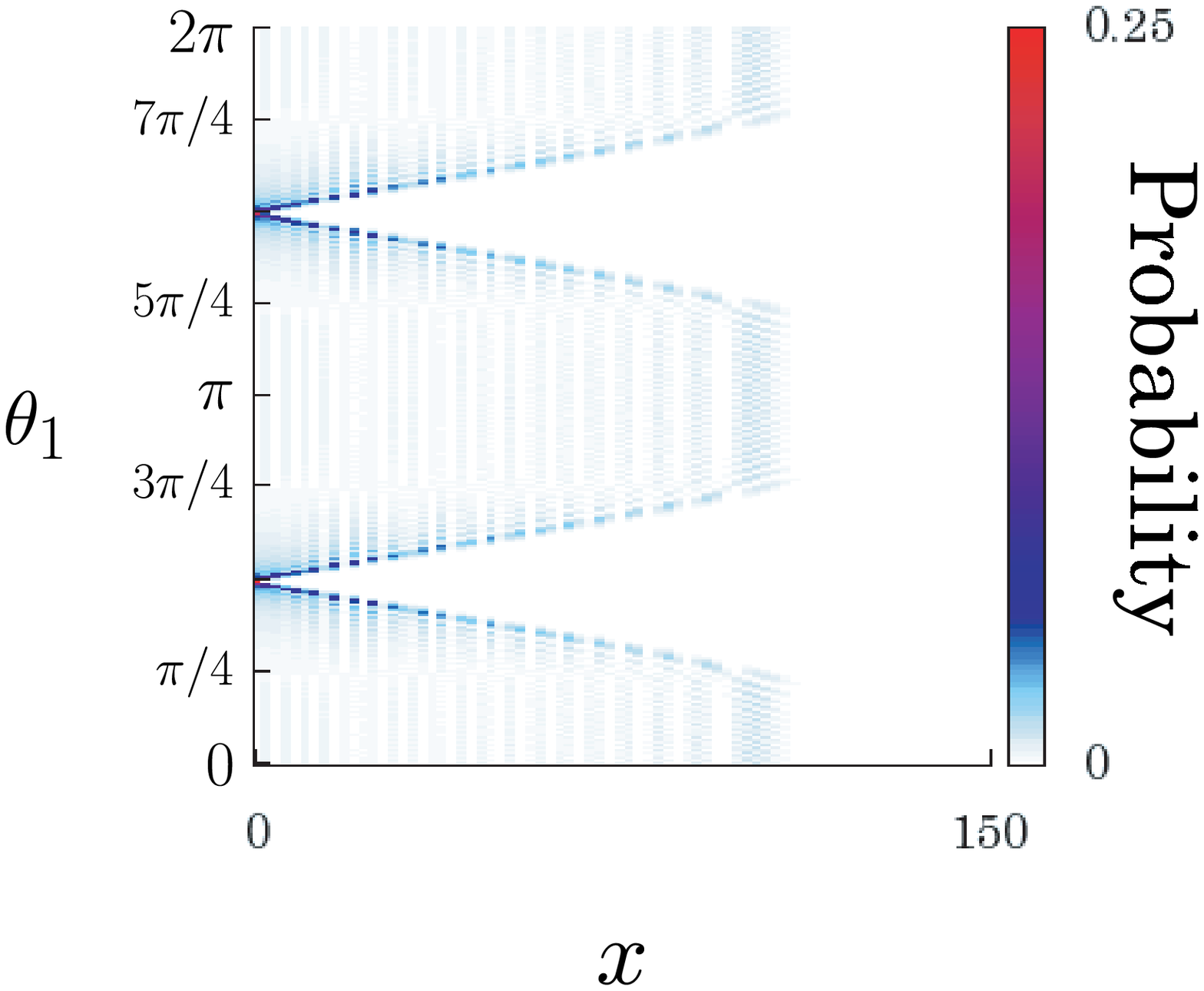}\\[2mm]
  (a) $\mathbb{P}(X_{150}^{HL}=x;0)$
  \end{center}
 \end{minipage}
 \begin{minipage}{35mm}
  \begin{center}
   \includegraphics[scale=0.2]{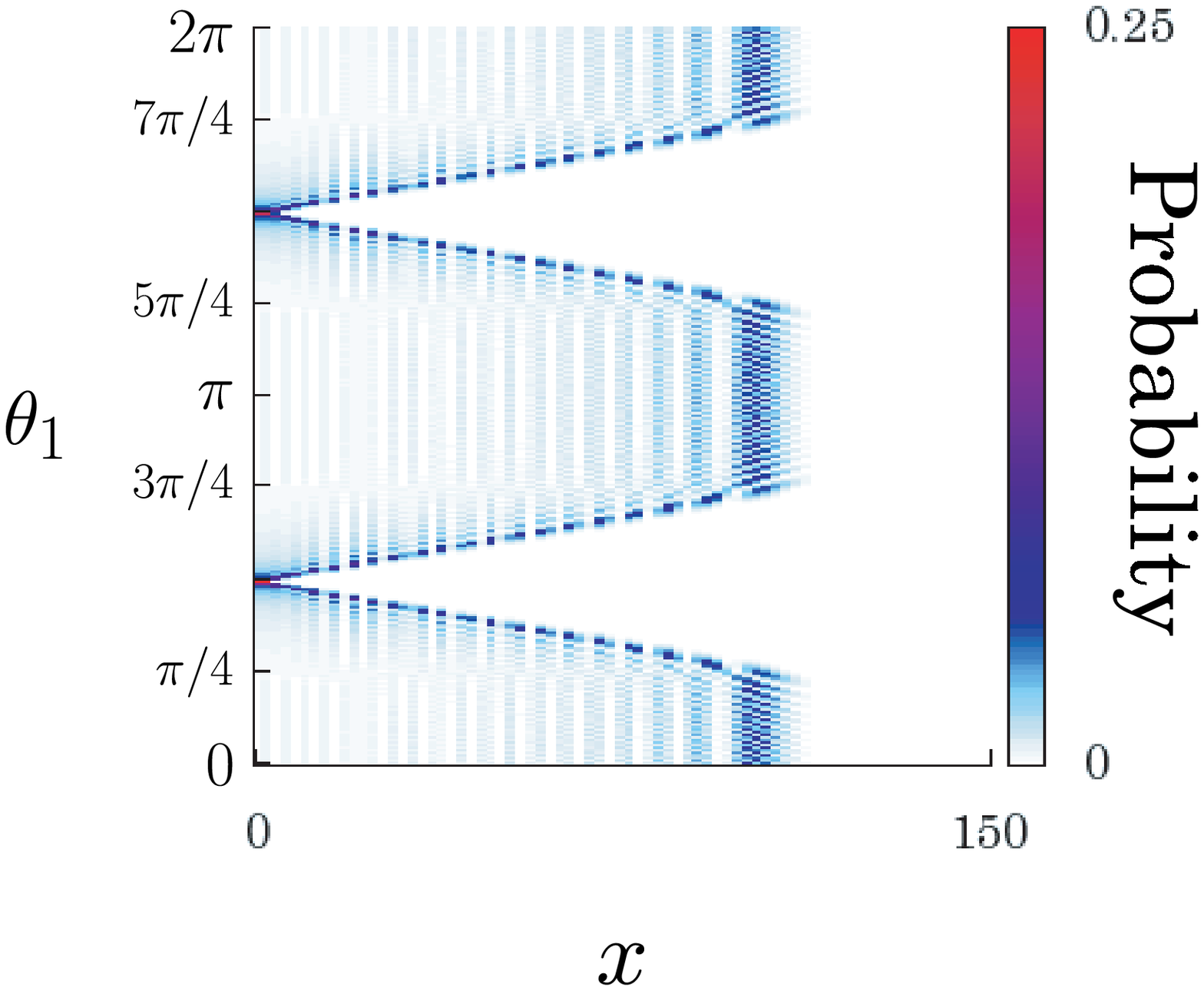}\\[2mm]
  (b) $\mathbb{P}(X_{150}^{HL}=x;1)$
  \end{center}
 \end{minipage}
 \begin{minipage}{35mm}
  \begin{center}
   \includegraphics[scale=0.2]{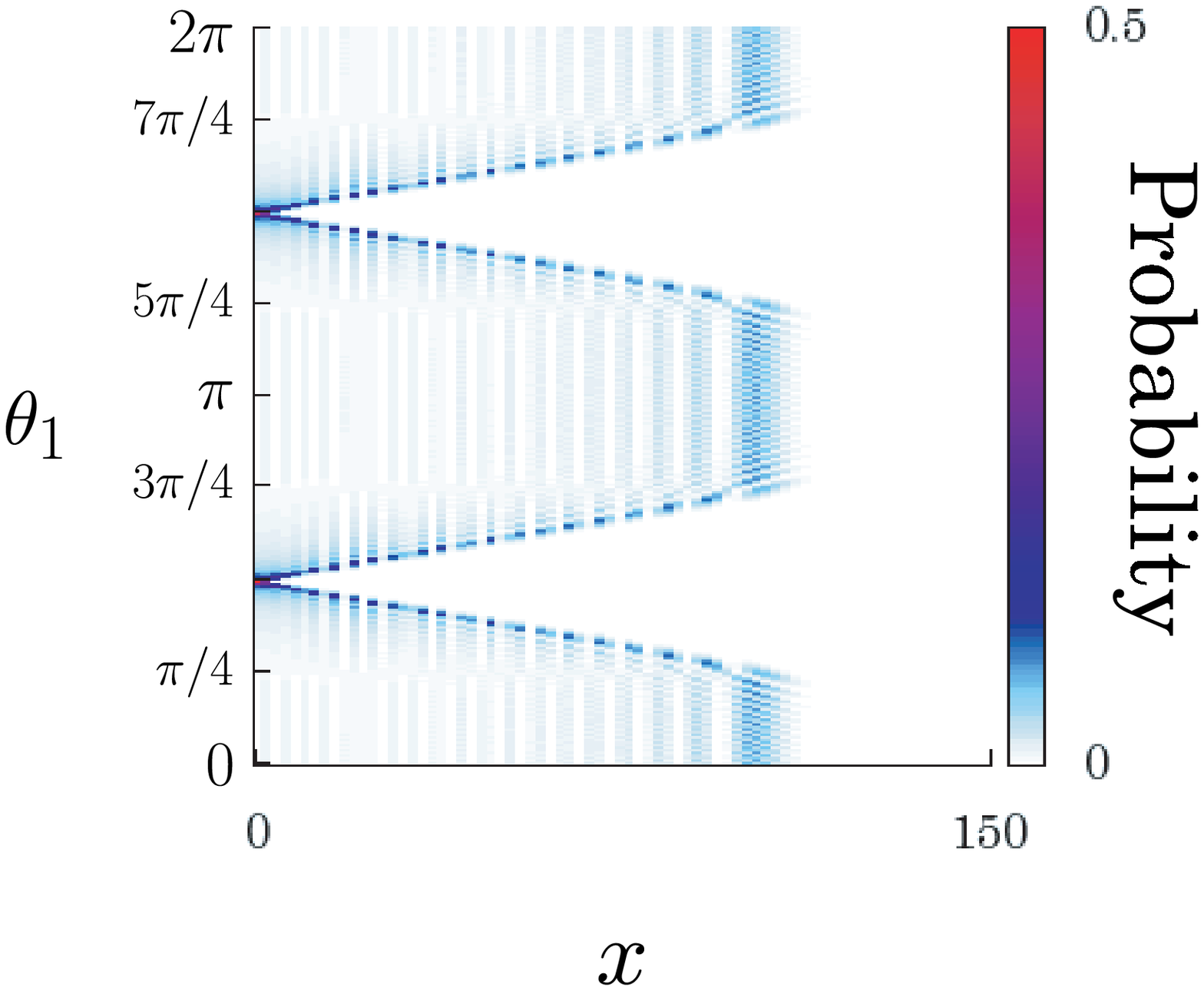}\\[2mm]
  (c) $\mathbb{P}(X_{150}^{HL}=x)$
  \end{center}
 \end{minipage}
\vspace{5mm}
\caption{$\theta_2=\pi/4$ : These pictures show how far the distributions spread out at time $150$, dependently on the parameter $\theta_1$. ($\alpha=1/\sqrt{2},\,\beta=i/\sqrt{2}$)}
\label{fig:160718_07}
\end{center}
\end{figure}

\begin{figure}[h]
\begin{center}
 \begin{minipage}{35mm}
  \begin{center}
   \includegraphics[scale=0.2]{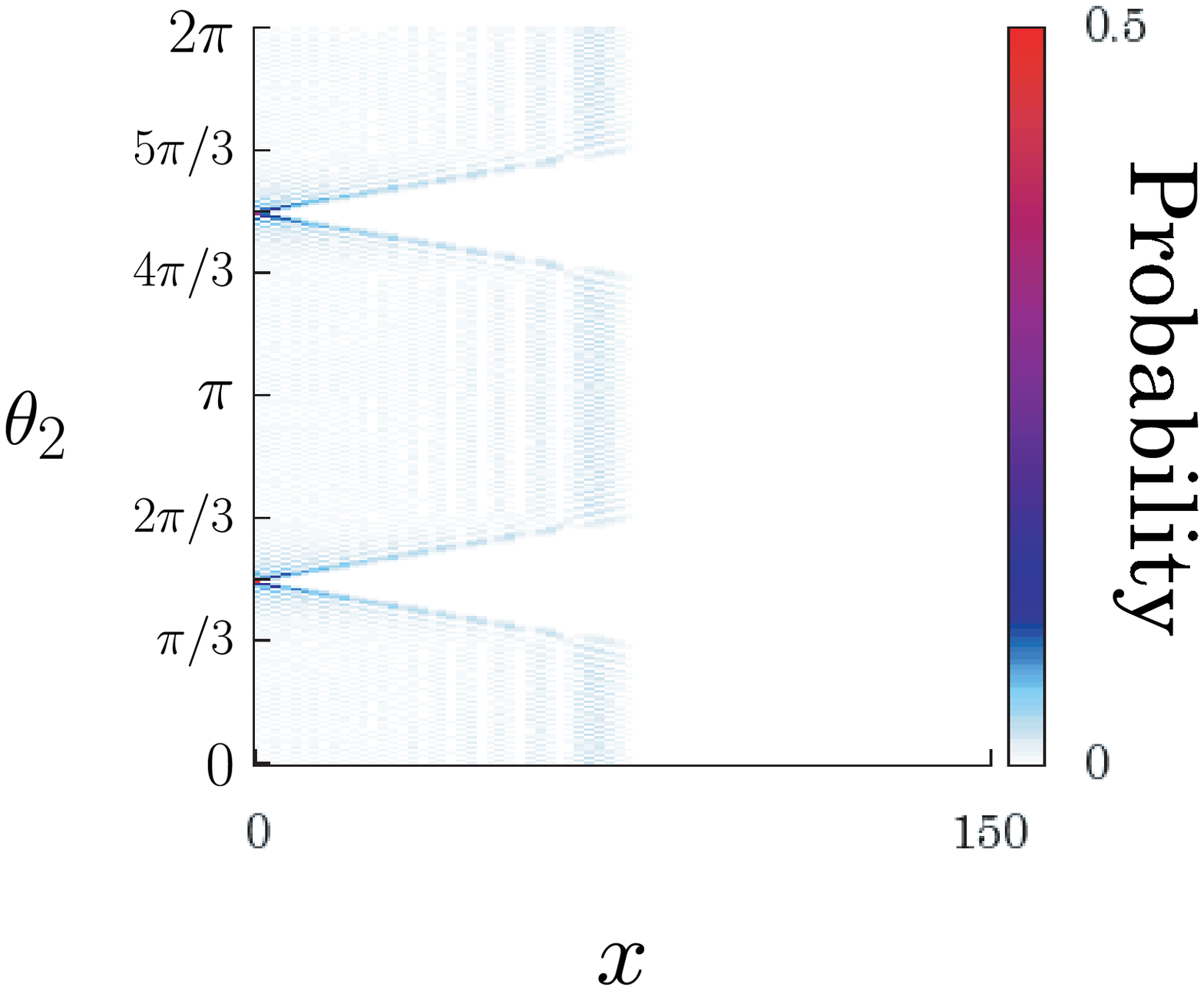}\\[2mm]
  (a) $\mathbb{P}(X_{150}^{HL}=x;0)$
  \end{center}
 \end{minipage}
 \begin{minipage}{35mm}
  \begin{center}
   \includegraphics[scale=0.2]{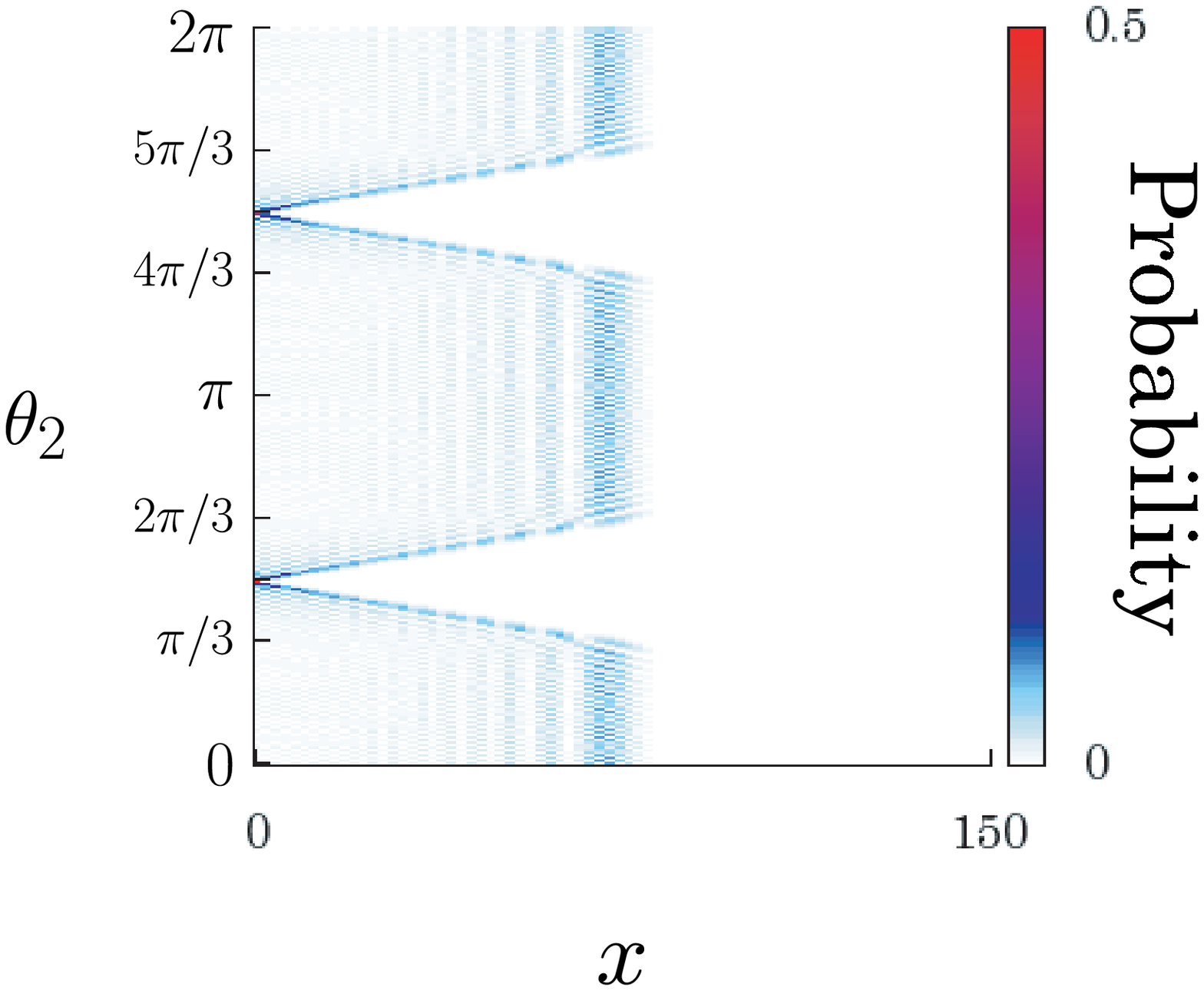}\\[2mm]
  (b) $\mathbb{P}(X_{150}^{HL}=x;1)$
  \end{center}
 \end{minipage}
 \begin{minipage}{35mm}
  \begin{center}
   \includegraphics[scale=0.2]{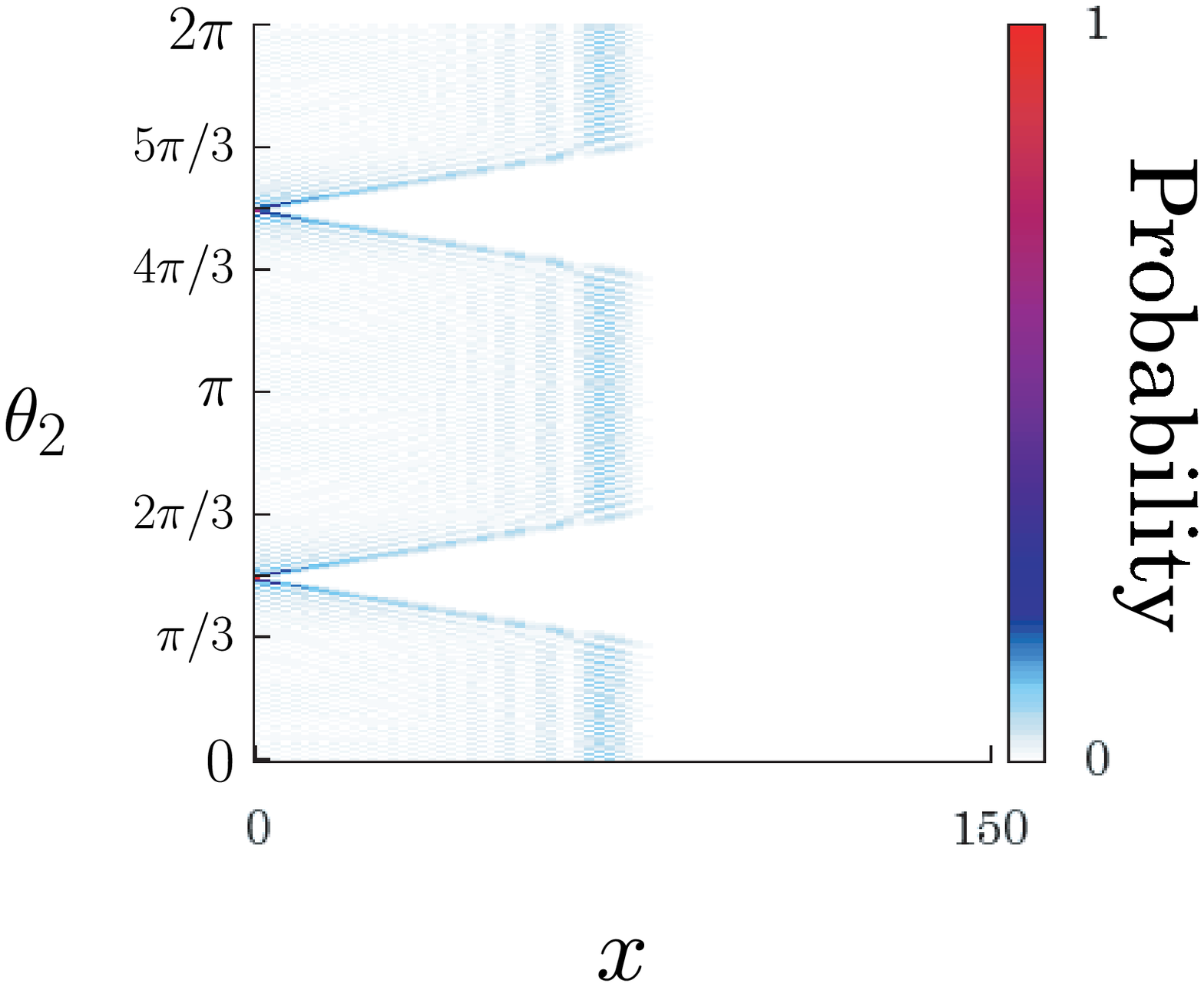}\\[2mm]
  (c) $\mathbb{P}(X_{150}^{HL}=x)$
  \end{center}
 \end{minipage}
\vspace{5mm}
\caption{$\theta_1=\pi/3$ : These pictures show how far the distributions spread out at time $150$, dependently on the parameter $\theta_2$. ($\alpha=1/\sqrt{2},\,\beta=i/\sqrt{2}$)}
\label{fig:160718_10}
\end{center}
\end{figure}

It was proved in Machida~\cite{Machida2016} that a time-independent quantum walk on the half line, that happened when the parameters $\theta_1$ and $\theta_2$ took the same value, was reproduced by a time-independent quantum walk on the line.
In this section we will see the system of time-dependent quantum walk on the half line can be also copied by a time-dependent quantum walk on the line.
The position Hilbert space of the quantum walk on the line, represented by $\mathcal{H}_p^L$, is spanned by the orthonormal basis $\left\{\ket{x} : x\in\mathbb{Z}\right\}$.
On the other hand, the inner state space is described by the same thing as the one for the quantum walk on the half line, the Hilbert space $\mathcal{H}_c$.
Then, the system of quantum walk on the line at time $t\in\left\{0,1,2,\ldots\right\}$, represented by $\ket{\Phi_t}\in\mathcal{H}_p^{L}\otimes\mathcal{H}_c$, gets a 2-periodic unitary evolution,
\begin{equation}
 \ket{\Phi_{t+1}}=\left\{\begin{array}{ll}
	       \tilde{S}^L\tilde{C_1}^L\ket{\Phi_t} & (t=0,2,4,\ldots),\\
			  \tilde{S}^L\tilde{C_2}^L\ket{\Phi_t} & (t=1,3,5,\ldots),\\
		     \end{array}\right.
\end{equation}
where
\begin{align}
 \tilde{C_1}^L=&\sum_{x\in\mathbb{Z}}\ket{x}\bra{x}\otimes C_1,\\
 \tilde{C_2}^L=&\sum_{x\in\mathbb{Z}}\ket{x}\bra{x}\otimes C_2,\\
 \tilde{S}^L=&\sum_{x\in\mathbb{Z}}\ket{x-1}\bra{x}\otimes\ket{0}\bra{0}+\ket{x+1}\bra{x}\otimes\ket{1}\bra{1},
\end{align}
in which the unitary operations $C_1$ and $C_2$ are the ones given in Eqs.~\eqref{eq:coin-flip_operator_1} and \eqref{eq:coin-flip_operator_2}.
Setting up a possibly delocalized initial state, also simply called a delocalized initial state, on the quantum walk on the line,
\begin{equation}
 \ket{\Phi_0}=\ket{-1}\otimes\left(\alpha_{-1}\ket{0}+\beta_{-1}\ket{1}\right)+\ket{0}\otimes\left(\alpha_0\ket{0}+\beta_0\ket{1}\right),\label{eq:L_initial_state}
\end{equation}
with $|\alpha_{-1}|^2+|\beta_{-1}|^2+|\alpha_0|^2+|\beta_0|^2=1\, (\alpha_{-1},\beta_{-1},\alpha_0,\beta_0\in\mathbb{C})$, we will realize a relation between the quantum walk on the half line and the quantum walk on the line, as shown in Lemma~\ref{lem:160612_01} later.
    
Similarly to Eq.~\eqref{eq:HL_probability_inner_state}, the walker in inner state $j\in\left\{0,1\right\}$ is observed at position $x\in\mathbb{Z}$ on the line with probability
\begin{equation}
 \mathbb{P}(Y_t^L=x;j)=\bra{\Phi_t}\Bigl(\ket{x}\bra{x}\otimes\ket{j}\bra{j}\Bigr)\ket{\Phi_t},
\end{equation}
and the finding probability regardless of the inner states should be defined as the sum over the two inner states, 
\begin{equation}
 \mathbb{P}(Y_t^L=x)=\sum_{j=0}^1\mathbb{P}(Y_t^L=x;j)=\bra{\Phi_t}\biggl(\ket{x}\bra{x}\otimes\sum_{j=0}^1\ket{j}\bra{j}\biggr)\ket{\Phi_t},
\end{equation}
where $Y_t^L$ represents the position of the quantum walker on the line launching with the delocalized initial state given in Eq.~\eqref{eq:L_initial_state}.
As shown in Figs.~\ref{fig:160718_13} and \ref{fig:160718_16}, each finding probability at a moment contains two sharp peaks and linearly spreads out as time $t$ goes up.
Figures~\ref{fig:160718_19} and \ref{fig:160718_22} tell us that the dependency on the parameters $\theta_1$ and $\theta_2$ has the same feature that we already viewed in Figs.~\ref{fig:160718_07} and \ref{fig:160718_10}.

\begin{figure}[h]
\begin{center}
 \begin{minipage}{35mm}
  \begin{center}
   \includegraphics[scale=0.3]{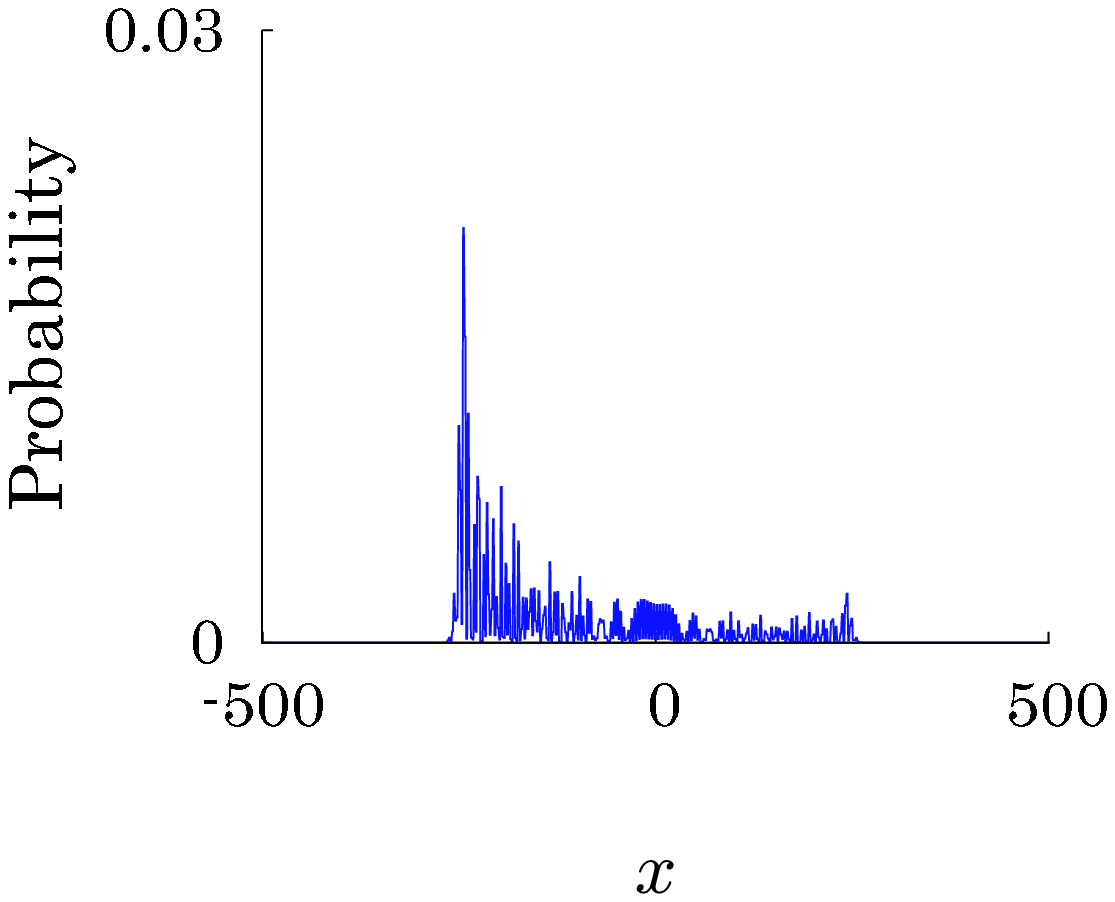}\\[2mm]
  (a) $\mathbb{P}(Y_{500}^L=x;0)$
  \end{center}
 \end{minipage}
 \begin{minipage}{35mm}
  \begin{center}
   \includegraphics[scale=0.3]{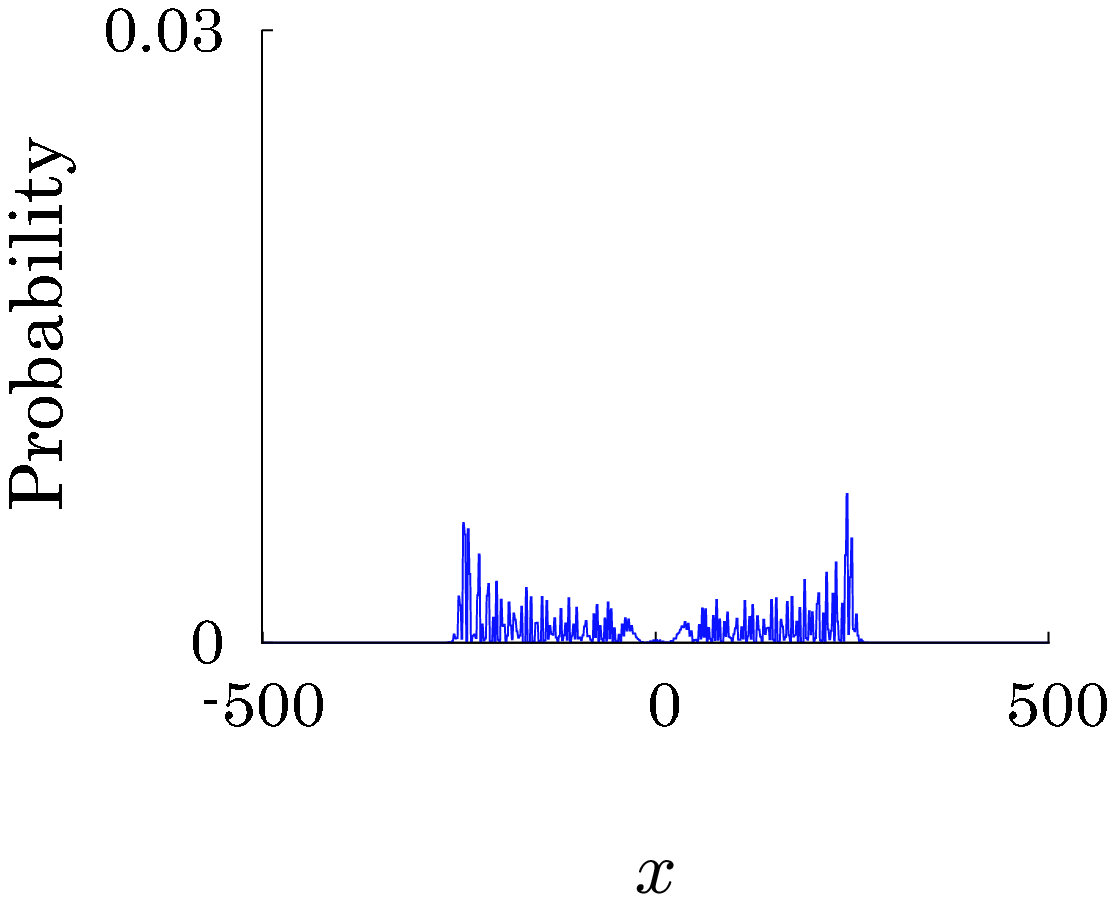}\\[2mm]
  (b) $\mathbb{P}(Y_{500}^L=x;1)$
  \end{center}
 \end{minipage}
 \begin{minipage}{35mm}
  \begin{center}
   \includegraphics[scale=0.3]{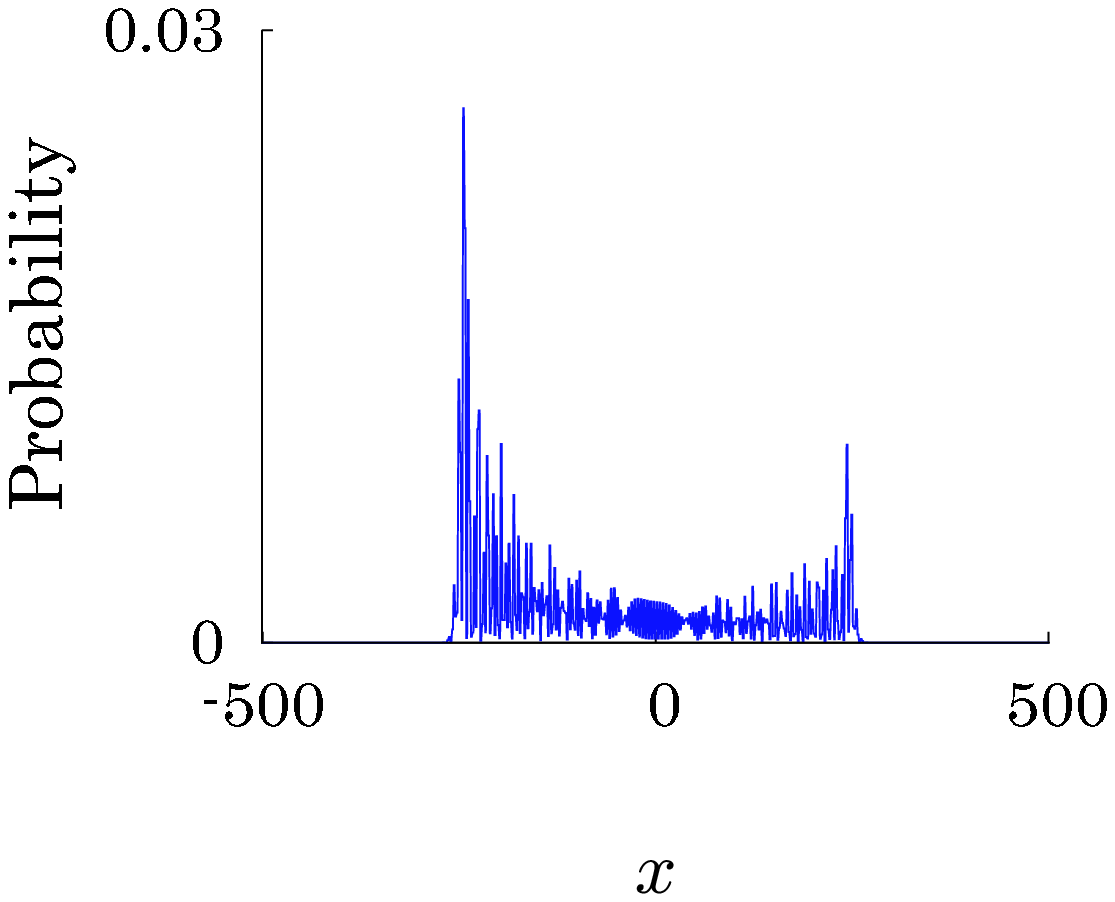}\\[2mm]
  (c) $\mathbb{P}(Y_{500}^L=x)$
  \end{center}
 \end{minipage}
\vspace{5mm}
\caption{$\theta_1=\pi/3,\,\theta_2=\pi/4$ : Finding probabilities of the quantum walker on the line at time $500$ when the walker launches at time $0$ with the delocalized initial state in Eq.~\eqref{eq:L_initial_state}. ($\alpha_{-1}=1/\sqrt{2},\, \beta_{-1}=0,\, \alpha_0=1/\sqrt{2},\,\beta_0=0$)}
\label{fig:160718_13}
\end{center}
\end{figure}

\begin{figure}[h]
\begin{center}
 \begin{minipage}{35mm}
  \begin{center}
   \includegraphics[scale=0.2]{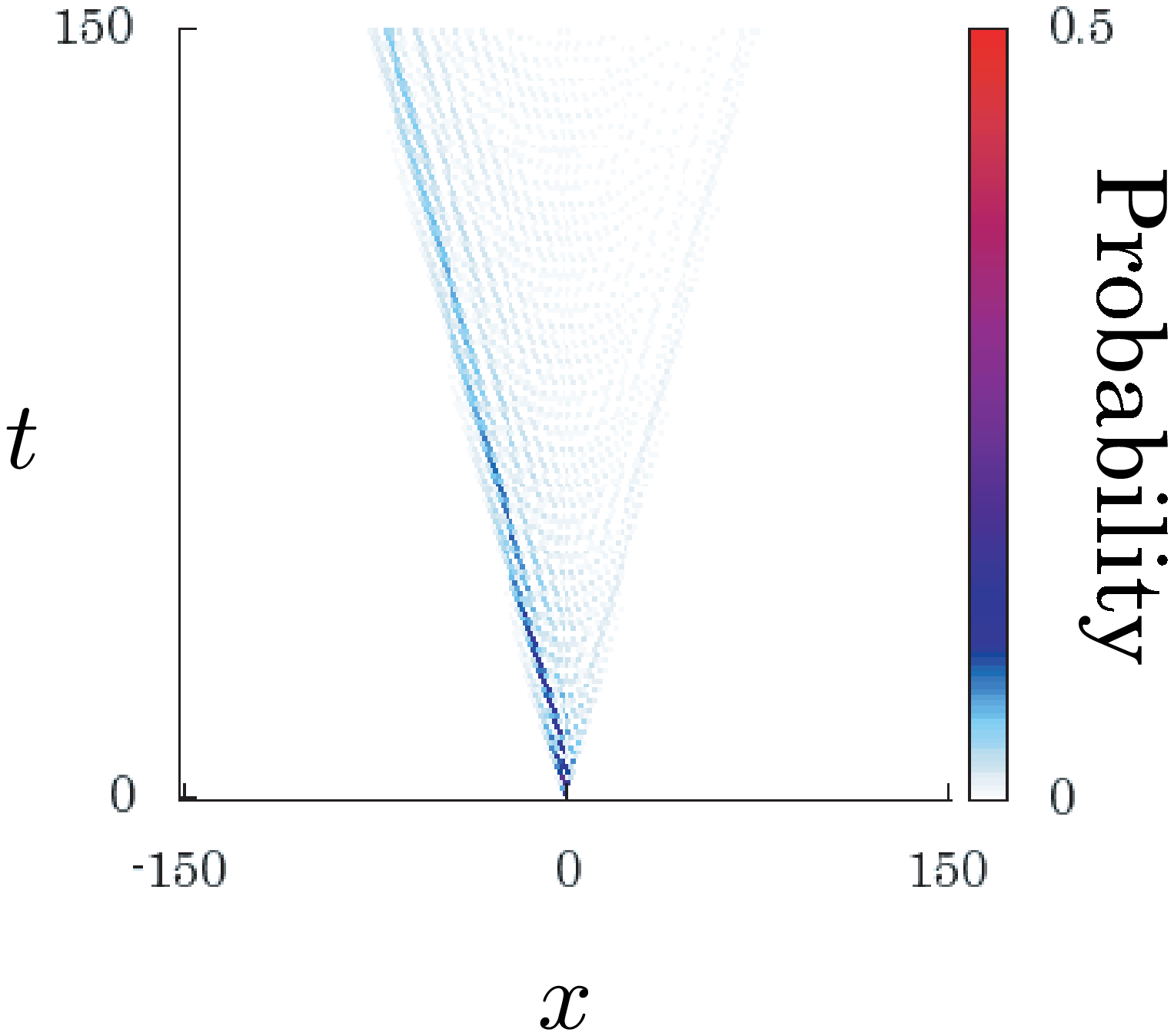}\\[2mm]
  (a) $\mathbb{P}(Y_t^L=x;0)$
  \end{center}
 \end{minipage}
 \begin{minipage}{35mm}
  \begin{center}
   \includegraphics[scale=0.2]{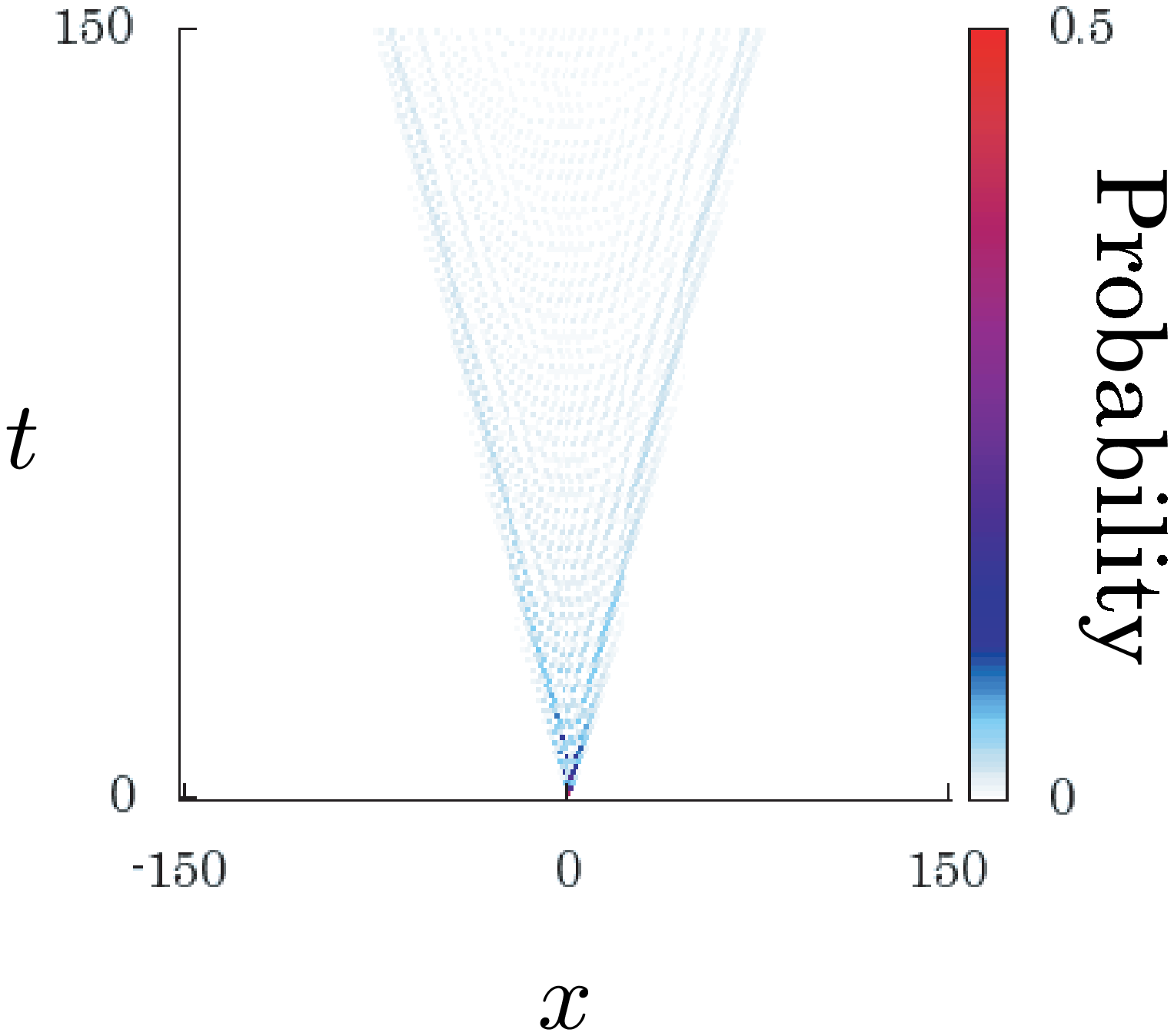}\\[2mm]
  (b) $\mathbb{P}(Y_t^L=x;1)$
  \end{center}
 \end{minipage}
 \begin{minipage}{35mm}
  \begin{center}
   \includegraphics[scale=0.2]{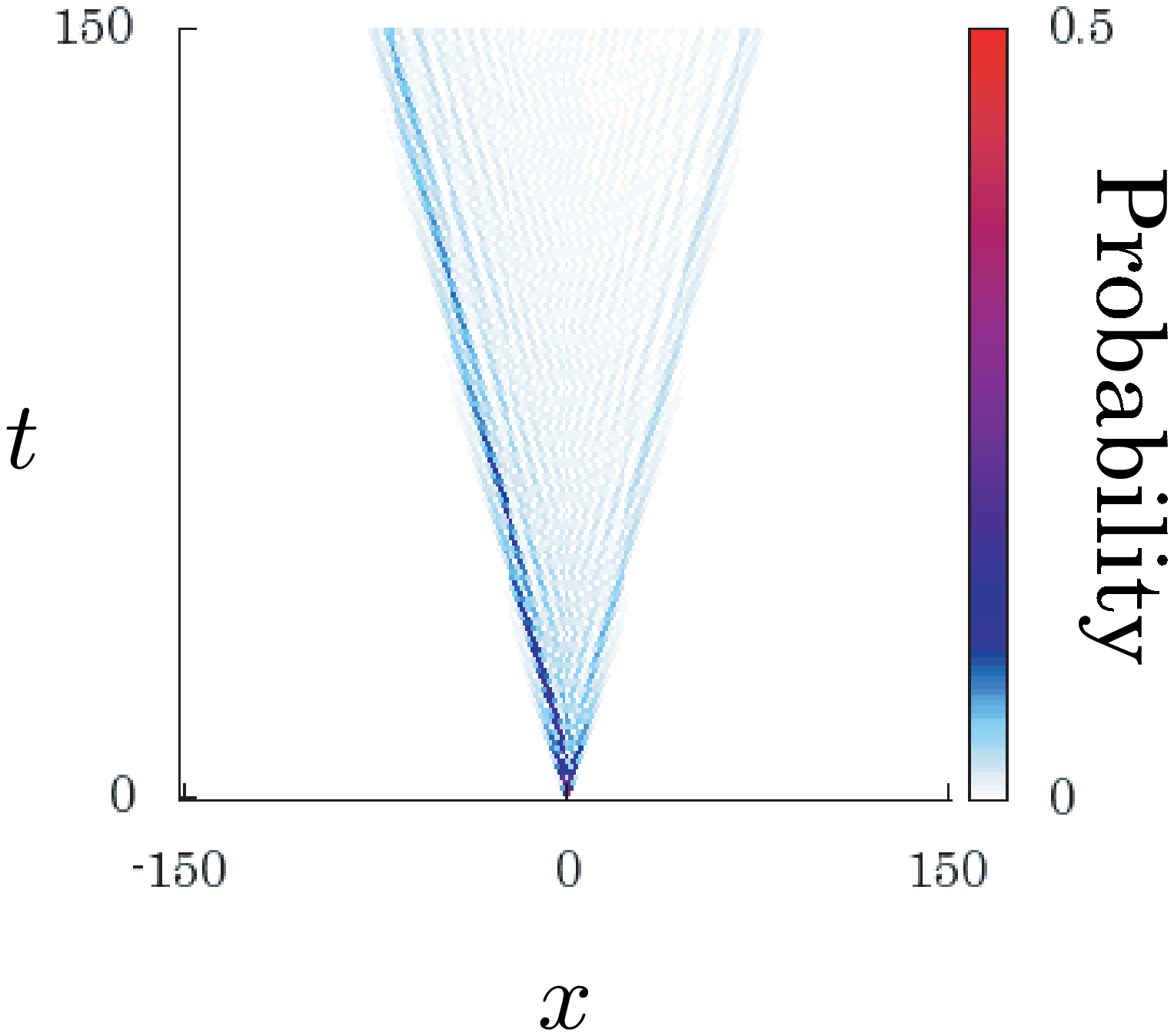}\\[2mm]
  (c) $\mathbb{P}(Y_t^L=x)$
  \end{center}
 \end{minipage}
\vspace{5mm}
\caption{$\theta_1=\pi/3,\,\theta_2=\pi/4$ : Each distribution is spreading out as time $t$ goes up, and its behavior is ballistic. ($\alpha_{-1}=1/\sqrt{2},\, \beta_{-1}=0,\, \alpha_0=1/\sqrt{2},\,\beta_0=0$)}
\label{fig:160718_16}
\end{center}
\end{figure}

\begin{figure}[h]
\begin{center}
 \begin{minipage}{35mm}
  \begin{center}
   \includegraphics[scale=0.2]{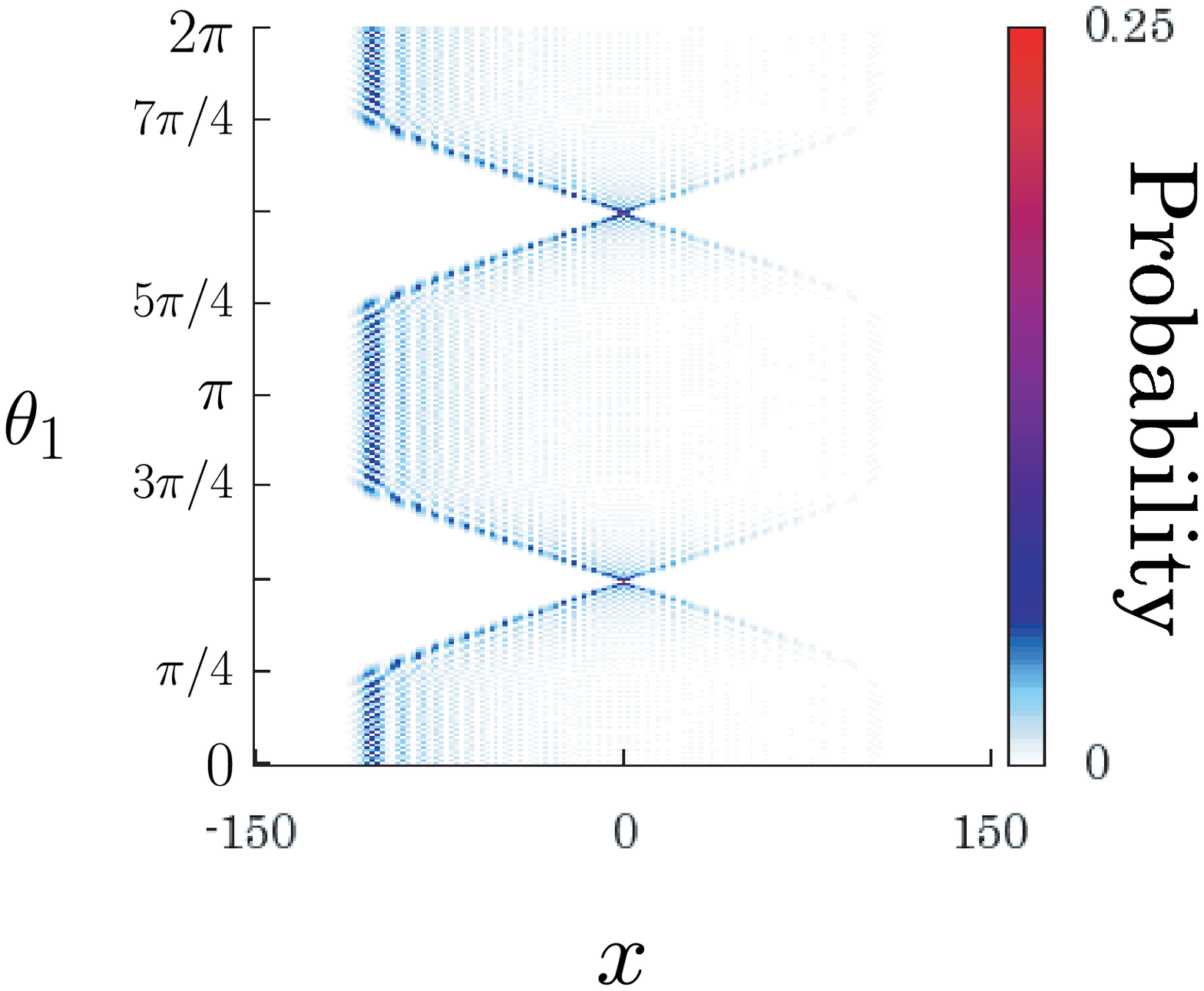}\\[2mm]
  (a) $\mathbb{P}(Y_{150}^L=x;0)$
  \end{center}
 \end{minipage}
 \begin{minipage}{35mm}
  \begin{center}
   \includegraphics[scale=0.2]{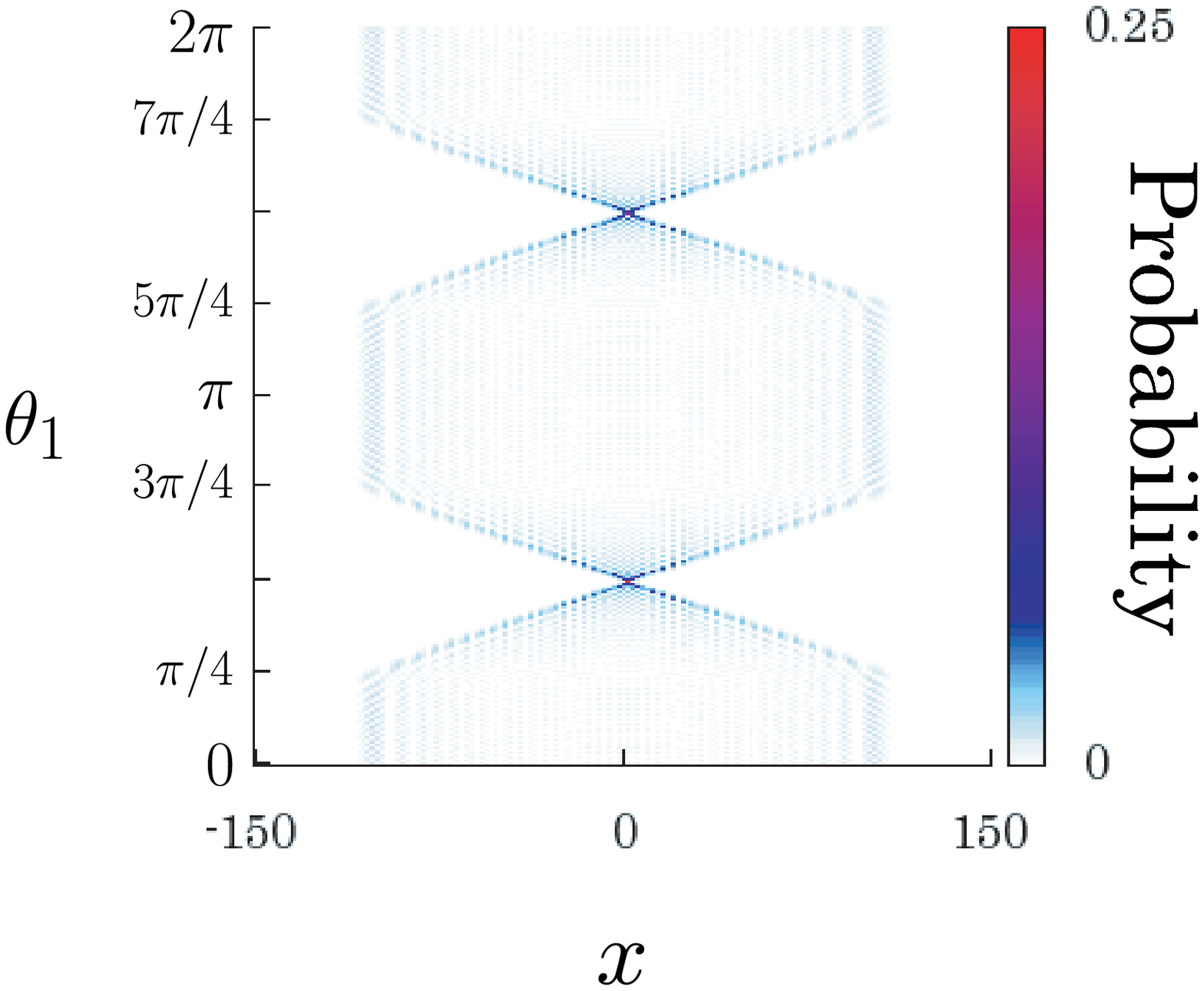}\\[2mm]
  (b) $\mathbb{P}(Y_{150}^L=x;1)$
  \end{center}
 \end{minipage}
 \begin{minipage}{35mm}
  \begin{center}
   \includegraphics[scale=0.2]{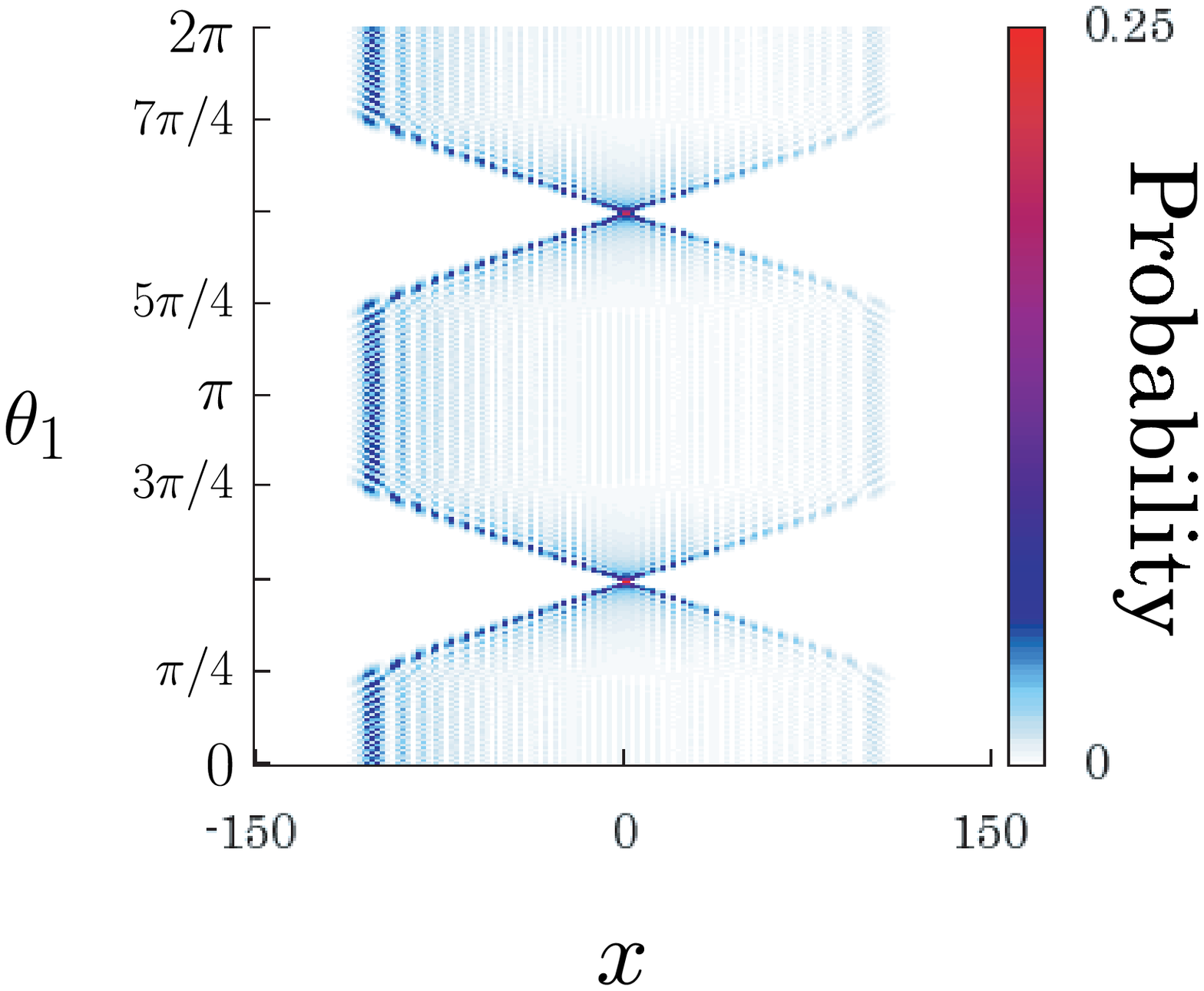}\\[2mm]
  (c) $\mathbb{P}(Y_{150}^L=x)$
  \end{center}
 \end{minipage}
\vspace{5mm}
\caption{$\theta_2=\pi/4$ : These pictures show how the distributions depend on the parameter $\theta_1$. ($\alpha_{-1}=1/\sqrt{2},\, \beta_{-1}=0,\, \alpha_0=1/\sqrt{2},\,\beta_0=0$)}
\label{fig:160718_19}
\end{center}
\end{figure}

\begin{figure}[h]
\begin{center}
 \begin{minipage}{35mm}
  \begin{center}
   \includegraphics[scale=0.2]{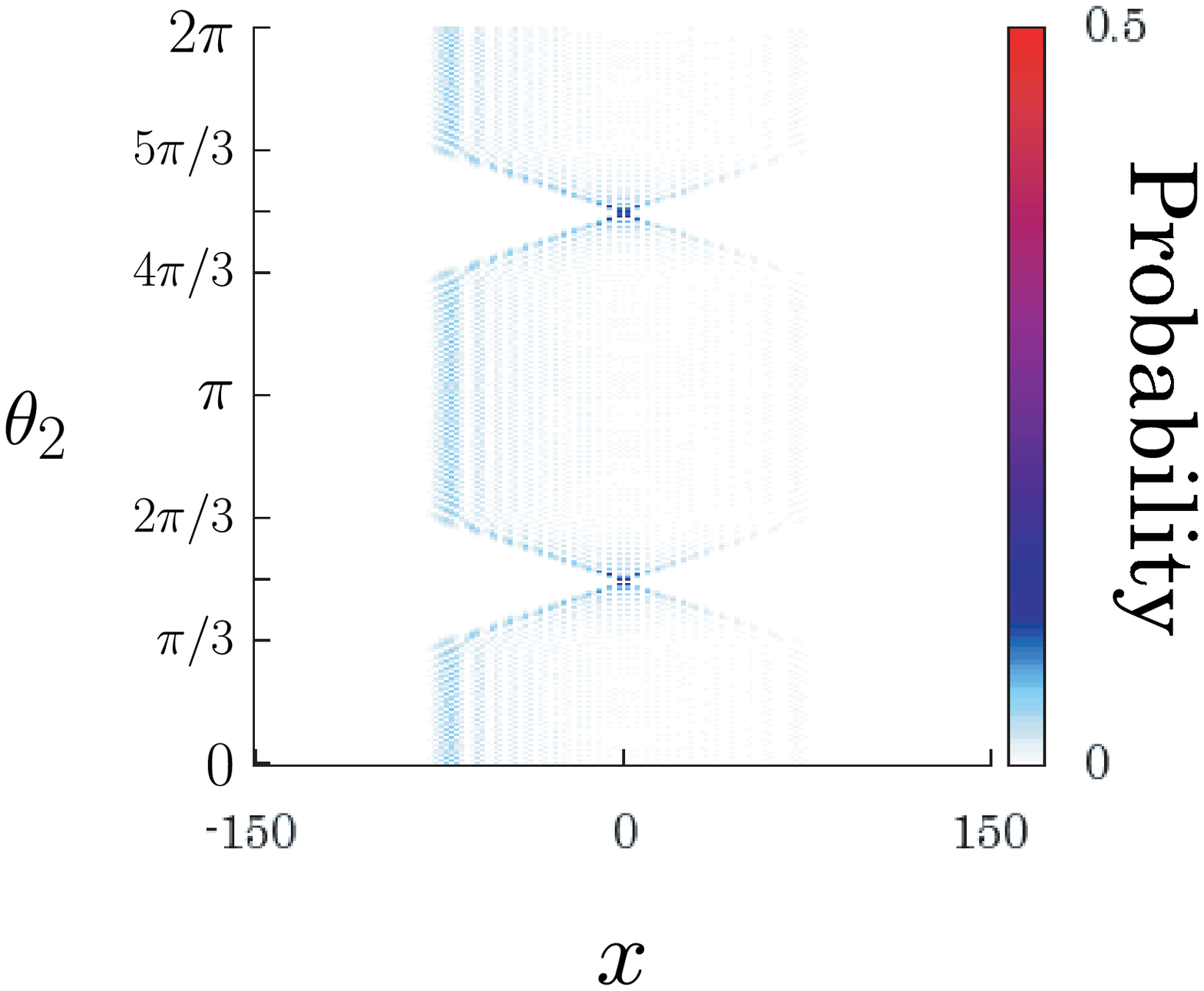}\\[2mm]
  (a) $\mathbb{P}(Y_{150}^L=x;0)$
  \end{center}
 \end{minipage}
 \begin{minipage}{35mm}
  \begin{center}
   \includegraphics[scale=0.2]{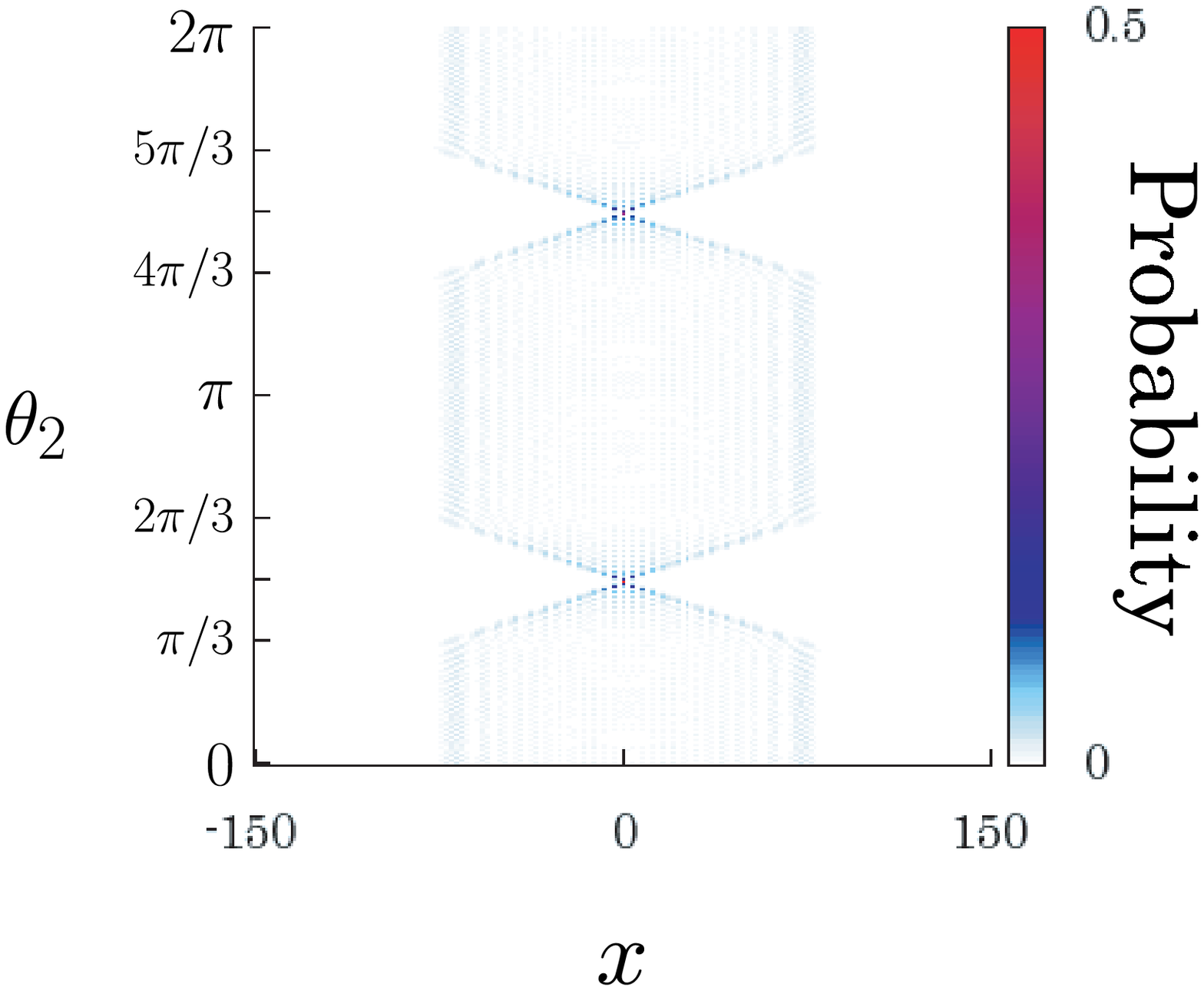}\\[2mm]
  (b) $\mathbb{P}(Y_{150}^L=x;1)$
  \end{center}
 \end{minipage}
 \begin{minipage}{35mm}
  \begin{center}
   \includegraphics[scale=0.2]{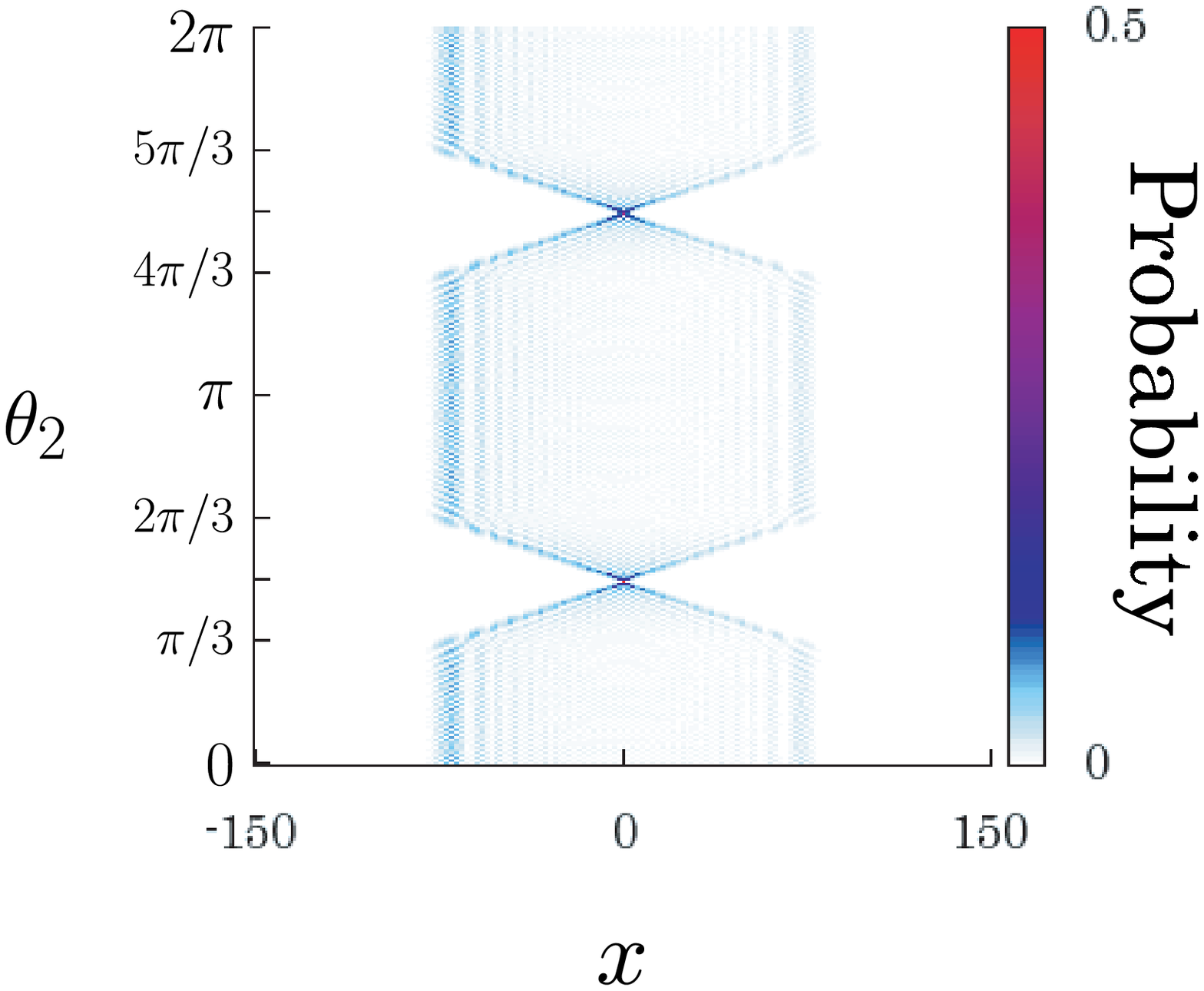}\\[2mm]
  (c) $\mathbb{P}(Y_{150}^L=x)$
  \end{center}
 \end{minipage}
\vspace{5mm}
\caption{$\theta_1=\pi/3$ : These pictures show how the distributions depend on the parameter $\theta_2$. ($\alpha_{-1}=1/\sqrt{2},\, \beta_{-1}=0,\, \alpha_0=1/\sqrt{2},\,\beta_0=0$)}
\label{fig:160718_22}
\end{center}
\end{figure}

Let $\xi\in\left\{1,2\right\}$ be the subscription such that $|c_\xi|=\min\left\{|c_1|, |c_2|\right\}$.
Then we get limit distributions of the quantum walk on the line. 
For $\theta_1,\theta_2\neq 0,\pi/2,\pi,3\pi/2$, the distributions of the random variable $Y_t^L/t$ converge to integral representations as $t\to\infty$, 
\begin{align}
 \lim_{t\to\infty}\mathbb{P}\left(\frac{Y_t^L}{t}\leq x;0\right)&=\int_{-\infty}^x \frac{|s_\xi|}{2\pi(1+y)\sqrt{c_\xi^2-y^2}}\,\eta(y)I_{(-|c_\xi|,|c_\xi|)}(y)\,dy,\label{eq:limit_QWL_0}\\
 \lim_{t\to\infty}\mathbb{P}\left(\frac{Y_t^L}{t}\leq x;1\right)&=\int_{-\infty}^x \frac{|s_\xi|}{2\pi(1-y)\sqrt{c_\xi^2-y^2}}\,\eta(y)I_{(-|c_\xi|,|c_\xi|)}(y)\,dy,\label{eq:limit_QWL_1}\\
 \lim_{t\to\infty}\mathbb{P}\left(\frac{Y_t^L}{t}\leq x\right)&=\int_{-\infty}^x \frac{|s_\xi|}{\pi(1-y^2)\sqrt{c_\xi^2-y^2}}\,\eta(y)I_{(-|c_\xi|,|c_\xi|)}(y)\,dy,\label{eq:limit_QWL_prob}
\end{align}
where
\begin{equation}
 \eta(x)=1-\left(|\alpha_{-1}|^2+|\alpha_0|^2-|\beta_{-1}|^2-|\beta_0|^2+\frac{2s_1\Re(\alpha_{-1}\overline{\beta_{-1}}+\alpha_0\overline{\beta_0})}{c_1}\right)x,\label{eq:160810_eta}
\end{equation}
in which $\Re(z)$ means the real part of a complex number $z$.
These limit distributions can be computed by Fourier analysis which was applied for a time-dependent quantum walk with a localized initial state~\cite{MachidaKonno2010}.
However, since the computation for Eqs.~\eqref{eq:limit_QWL_0}, \eqref{eq:limit_QWL_1}, and \eqref{eq:limit_QWL_prob} is almost same as that for the limit distribution shown in the past study~\cite{MachidaKonno2010}, the proof is omitted in this paper. 

Here, we will understand that the quantum walk on the line with a delocalized initial state can hold all the information, leveled at amplitude, about the quantum walk on the line with a localized initial state.
With the representations  
\begin{align}
 \ket{\Psi_t}=&\sum_{x=0}^\infty \ket{x}\otimes\Bigl(\alpha_t(x)\ket{0}+\beta_t(x)\ket{1}\Bigr),\\
 \ket{\Phi_t}=&\sum_{x\in\mathbb{Z}} \ket{x}\otimes\Bigl(\gamma_t(x)\ket{0}+\delta_t(x)\ket{1}\Bigr),
\end{align}
where $\alpha_t(x), \beta_t(x), \gamma_t(x)$, and $\delta_t(x)$ have to be complex numbers, a desired connection comes up as the following lemma.
\begin{lem}
\label{lem:160612_01}
 Let $\alpha_{-1}, \beta_{-1}, \alpha_0$, and $\beta_0$ be the real numbers such that
 \begin{equation}
  \alpha_{-1}=\Im(\beta),\quad \beta_{-1}=-\Im(\alpha),\quad \alpha_0=\Re(\alpha),\quad \beta_0=\Re(\beta),
 \end{equation}
 from which the delocalized initial state becomes of the form
 \begin{equation}
  \ket{\Phi_0}=\ket{-1}\otimes\Bigl(\Im(\beta)\ket{0}-\Im(\alpha)\ket{1}\Bigr)+\ket{0}\otimes\Bigl(\Re(\alpha)\ket{0}+\Re(\beta)\ket{1}\Bigr).\label{eq:160720_09}
 \end{equation}
 The letters $\Re(z)$ and $\Im(z)$ indicate the real part and the imaginary part of a complex number $z$ respectively.
 Then, for $x=0,1,2,\ldots$, we have
 \begin{align}
  \alpha_{2t}(2x)=&\gamma_{2t}(2x)-i\delta_{2t}(-2x-1),\label{eq:160720_01}\\
  \beta_{2t}(2x)=&\delta_{2t}(2x)+i\gamma_{2t}(-2x-1),\label{eq:160720_02}\\
  \alpha_{2t}(2x+1)=&\delta_{2t}(-2x-2)+i\gamma_{2t}(2x+1),\label{eq:160720_03}\\
  \beta_{2t}(2x+1)=&-\gamma_{2t}(-2x-2)+i\delta_{2t}(2x+1),\label{eq:160720_04}\\[2mm]
  \alpha_{2t+1}(2x)=&-\delta_{2t+1}(-2x-1)+i\gamma_{2t+1}(2x),\label{eq:160720_05}\\
  \beta_{2t+1}(2x)=&\gamma_{2t+1}(-2x-1)+i\delta_{2t+1}(2x),\label{eq:160720_06}\\
  \alpha_{2t+1}(2x+1)=&\gamma_{2t+1}(2x+1)+i\delta_{2t+1}(-2x-2),\label{eq:160720_07}\\
  \beta_{2t+1}(2x+1)=&\delta_{2t+1}(2x+1)-i\gamma_{2t+1}(-2x-2).\label{eq:160720_08}
 \end{align}
\end{lem}

\begin{proof}{%
Since the initial states imply
\begin{align}
 &\alpha_0(0)=\alpha,\,\beta_0(0)=\beta,\\
 &\alpha_0(x)=\beta_0(x)=0\quad(x=1,2,\ldots),\\[3mm]
 &\gamma_0(-1)=\Im(\beta),\,\delta_0(-1)=-\Im(\alpha),\,\gamma_0(0)=\Re(\alpha),\,\delta_0(0)=\Re(\beta),\\
 &\gamma_0(x)=\delta_0(x)=0\quad (x= 1,\pm 2,\pm 3,\ldots),
\end{align}
we see Eqs.~\eqref{eq:160720_01}--\eqref{eq:160720_04} for $t=0$,
\begin{align}
 &\left\{\begin{array}{l}
   \alpha_0(0)=\alpha,\\
	  \gamma_0(0)-i\delta_0(-1)=\Re(\alpha)-i(-\Im(\alpha))=\Re(\alpha)+i\Im(\alpha)=\alpha,\\
	  \beta_0(0)=\beta,\\
	  \delta_0(0)+i\gamma_0(-1)=\Re(\beta)+i\Im(\beta)=\beta,
	 \end{array}\right.\label{eq:160721_01}\\
 &\left\{\begin{array}{l}
   \alpha_0(2x)=0,\\
	  \gamma_0(2x)-i\delta_0(-2x-1)=0-i\cdot 0=0,\\
	  \beta_0(2x)=0,\\
	  \delta_0(2x)+i\gamma_0(-2x-1)=0+i\cdot 0=0,
	 \end{array}\right.\quad (x=1,2,\ldots),\label{eq:160721_02}\\
 &\left\{\begin{array}{l}
   \alpha_0(2x+1)=0,\\
	  \delta_0(-2x-2)+i\gamma_0(2x+1)=0+i\cdot 0=0,\\
	  \beta_0(2x+1)=0,\\
	  -\gamma_0(-2x-2)+i\delta_0(2x+1)=-0+i\cdot 0=0,
	 \end{array}\right.\quad (x=0,1,2,\ldots).\label{eq:160721_03}
\end{align}
On the other hand, assuming Eqs.~\eqref{eq:160720_01}--\eqref{eq:160720_04} are true, Eqs.~\eqref{eq:160720_05}--\eqref{eq:160720_08} are derived by the usage of the assumption,
\begin{align}
 \alpha_{2t+1}(0)=&c_1\alpha_{2t}(1)+s_1\beta_{2t}(1)\nonumber\\
 =&c_1\left\{\delta_{2t}(-2)+i\gamma_{2t}(1)\right\}+s_1\left\{-\gamma_{2t}(-2)+i\delta_{2t}(1)\right\}\nonumber\\
 =&-\left\{s_1\gamma_{2t}(-2)-c_1\delta_{2t}(-2)\right\}+i\left\{c_1\gamma_{2t}(1)+s_1\delta_{2t}(1)\right\}\nonumber\\
 =&-\delta_{2t+1}(-1)+i\gamma_{2t+1}(0),\label{eq:160721_04}\\[2mm]
 \beta_{2t+1}(0)=&c_1\alpha_{2t}(0)+s_1\beta_{2t}(0)\nonumber\\
 =&c_1\left\{\gamma_{2t}(0)-i\delta_{2t}(-1)\right\}+s_1\left\{\delta_{2t}(0)+i\gamma_{2t}(-1)\right\}\nonumber\\
 =&c_1\gamma_{2t}(0)+s_1\delta_{2t}(0)+i\left\{s_1\gamma_{2t}(-1)-c_1\delta_{2t}(-1)\right\}\nonumber\\
 =&\gamma_{2t+1}(-1)+i\delta_{2t+1}(0),\label{eq:160721_05}
\end{align}
\begin{align}
 \alpha_{2t+1}(2x)=&c_1\alpha_{2t}(2x+1)+s_1\beta_{2t}(2x+1)\nonumber\\
 =&c_1\left\{\delta_{2t}(-2x-2)+i\gamma_{2t}(2x+1)\right\}+s_1\left\{-\gamma_{2t}(-2x-2)+i\delta_{2t}(2x+1)\right\}\nonumber\\
 =&-\left\{s_1\gamma_{2t}(-2x-2)-c_1\delta_{2t}(-2x-2)\right\}+i\left\{c_1\gamma_{2t}(2x+1)+s_1\delta_{2t}(2x+1)\right\}\nonumber\\
 =&-\delta_{2t+1}(-2x-1)+i\gamma_{2t+1}(2x)\qquad (x=1,2,\ldots),\label{eq:160721_06}\\[2mm]
 \beta_{2t+1}(2x)=&s_1\alpha_{2t}(2x-1)-c_1\beta_{2t}(2x-1)\nonumber\\
 =&s_1\left\{\delta_{2t}(-2x)+i\gamma_{2t}(2x-1)\right\}-c_1\left\{-\gamma_{2t}(-2x)+i\delta_{2t}(2x-1)\right\}\nonumber\\
 =&c_1\gamma_{2t}(-2x)+s_1\delta_{2t}(-2x)+i\left\{s_1\gamma_{2t}(2x-1)-c_1\delta_{2t}(2x-1)\right\}\nonumber\\
 =&\gamma_{2t+1}(-2x-1)+i\delta_{2t+1}(2x)\qquad (x=1,2,\ldots),\label{eq:160721_07}
\end{align}
\begin{align}
 \alpha_{2t+1}(2x+1)=&c_1\alpha_{2t}(2x+2)+s_1\beta_{2t}(2x+2)\nonumber\\
 =&c_1\left\{\gamma_{2t}(2x+2)-i\delta_{2t}(-2x-3)\right\}+s_1\left\{\delta_{2t}(2x+2)+i\gamma_{2t}(-2x-3)\right\}\nonumber\\
 =&c_1\gamma_{2t}(2x+2)+s_1\delta_{2t}(2x+2)+i\left\{s_1\gamma_{2t}(-2x-3)-c_1\delta_{2t}(-2x-3)\right\}\nonumber\\
 =&\gamma_{2t+1}(2x+1)+i\delta_{2t+1}(-2x-2)\qquad (x=0,1,2,\ldots),\label{eq:160721_08}\\[2mm]
 \beta_{2t+1}(2x+1)=&s_1\alpha_{2t}(2x)-c_1\beta_{2t}(2x)\nonumber\\
 =&s_1\left\{\gamma_{2t}(2x)-i\delta_{2t}(-2x-1)\right\}-c_1\left\{\delta_{2t}(2x)+i\gamma_{2t}(-2x-1)\right\}\nonumber\\
 =&s_1\gamma_{2t}(2x)-c_1\delta_{2t}(2x)-i\left\{c_1\gamma_{2t}(-2x-1)+s_1\delta_{2t}(-2x-1)\right\}\nonumber\\
 =&\delta_{2t+1}(2x+1)-i\gamma_{2t+1}(-2x-2)\qquad (x=0,1,2,\ldots).\label{eq:160721_09}
\end{align}
In a similar way, assuming Eqs.~\eqref{eq:160720_05}--\eqref{eq:160720_08} allows us to hold Eqs.~\eqref{eq:160720_01}--\eqref{eq:160720_04} in which $t$ is replaced with $t+1$,
\begin{align}
 \alpha_{2t+2}(0)=&c_2\alpha_{2t+1}(1)+s_2\beta_{2t+1}(1)\nonumber\\
 =&c_2\left\{\gamma_{2t+1}(1)+i\delta_{2t+1}(-2)\right\}+s_2\left\{\delta_{2t+1}(1)-i\gamma_{2t+1}(-2)\right\}\nonumber\\
 =&c_2\gamma_{2t+1}(1)+s_2\delta_{2t+1}(1)-i\left\{s_2\gamma_{2t+1}(-2)-c_2\delta_{2t+1}(-2)\right\}\nonumber\\
 =&\gamma_{2t+2}(0)-i\delta_{2t+2}(-1),\label{eq:160721_10}\\[2mm]
 \beta_{2t+2}(0)=&c_2\alpha_{2t+1}(0)+s_2\beta_{2t+1}(0)\nonumber\\
 =&c_2\left\{-\delta_{2t+1}(-1)+i\gamma_{2t+1}(0)\right\}+s_2\left\{\gamma_{2t+1}(-1)+i\delta_{2t+1}(0)\right\}\nonumber\\
 =&s_2\gamma_{2t+1}(-1)-c_2\delta_{2t+1}(-1)+i\left\{c_2\gamma_{2t+1}(0)+s_2\delta_{2t+1}(0)\right\}\nonumber\\
 =&\delta_{2t+2}(0)+i\gamma_{2t+2}(-1),\label{eq:160721_11}
\end{align}
\begin{align}
 \alpha_{2t+2}(2x)=&c_2\alpha_{2t+1}(2x+1)+s_2\beta_{2t+1}(2x+1)\nonumber\\
 =&c_2\left\{\gamma_{2t+1}(2x+1)+i\delta_{2t+1}(-2x-2)\right\}+s_2\left\{\delta_{2t+1}(2x+1)-i\gamma_{2t+1}(-2x-2)\right\}\nonumber\\
 =&c_2\gamma_{2t+1}(2x+1)+s_2\delta_{2t+1}(2x+1)-i\left\{s_2\gamma_{2t+1}(-2x-2)-c_2\delta_{2t+1}(-2x-2)\right\}\nonumber\\
 =&\gamma_{2t+2}(2x)-i\delta_{2t+2}(-2x-1)\qquad (x=1,2,\ldots),\label{eq:160721_12}\\[2mm]
 \beta_{2t+2}(2x)=&s_2\alpha_{2t+1}(2x-1)-c_2\beta_{2t+1}(2x-1)\nonumber\\
 =&s_2\left\{\gamma_{2t+1}(2x-1)+i\delta_{2t+1}(-2x)\right\}-c_2\left\{\delta_{2t+1}(2x-1)-i\gamma_{2t+1}(-2x)\right\}\nonumber\\
 =&s_2\gamma_{2t+1}(2x-1)-c_2\delta_{2t+1}(2x-1)+i\left\{c_2\gamma_{2t+1}(-2x)+s_2\delta_{2t+1}(-2x)\right\}\nonumber\\
 =&\delta_{2t+2}(2x)+i\gamma_{2t+2}(-2x-1)\qquad (x=1,2,\ldots),\label{eq:160721_13}
\end{align}
\begin{align}
 \alpha_{2t+2}(2x+1)=&c_2\alpha_{2t+1}(2x+2)+s_2\beta_{2t+1}(2x+2)\nonumber\\
 =&c_2\left\{-\delta_{2t+1}(-2x-3)+i\gamma_{2t+1}(2x+2)\right\}+s_2\left\{\gamma_{2t+1}(-2x-3)+i\delta_{2t+1}(2x+2)\right\}\nonumber\\
 =&s_2\gamma_{2t+1}(-2x-3)-c_2\delta_{2t+1}(-2x-3)+i\left\{c_2\gamma_{2t+1}(2x+2)+s_2\delta_{2t+1}(2x+2)\right\}\nonumber\\
 =&\delta_{2t+2}(-2x-2)+i\gamma_{2t+2}(2x+1)\qquad (x=0,1,2,\ldots),\label{eq:160721_14}\\[2mm]
 \beta_{2t+2}(2x+1)=&s_2\alpha_{2t+1}(2x)-c_2\beta_{2t+1}(2x)\nonumber\\
 =&s_2\left\{-\delta_{2t+1}(-2x-1)+i\gamma_{2t+1}(2x)\right\}-c_2\left\{\gamma_{2t+1}(-2x-1)+i\delta_{2t+1}(2x)\right\}\nonumber\\
 =&-\left\{c_2\gamma_{2t+1}(-2x-1)+s_2\delta_{2t+1}(-2x-1)\right\}+i\left\{s_2\gamma_{2t+1}(2x)-c_2\delta_{2t+1}(2x)\right\}\nonumber\\
 =&-\gamma_{2t+2}(-2x-2)+i\delta_{2t+2}(2x+1)\qquad (x=0,1,2,\ldots).\label{eq:160721_15}
\end{align}
Combining Eqs.~\eqref{eq:160721_01}--\eqref{eq:160721_15}, one can tell the statement of Lemma \ref{lem:160612_01} by mathematical induction.
}
\end{proof}
\bigskip

\begin{lem}
\label{lem:160612_03}
 If the walker on the line starts off with the initial state
 \begin{equation}
  \ket{\Phi_0}=\ket{-1}\otimes\Bigl(\Im(\beta)\ket{0}-\Im(\alpha)\ket{1}\Bigr)+\ket{0}\otimes\Bigl(\Re(\alpha)\ket{0}+\Re(\beta)\ket{1}\Bigr),\label{eq:ass_lem_160612_03}
 \end{equation}
 its probability distributions reproduce those of the walker on the half line,
 \begin{align}
  \mathbb{P}(X_t^{HL}=x;0)=&\mathbb{P}(Y_t^L=-x-1;1)+\mathbb{P}(Y_t^L=x;0),\label{eq:lem_160612_03_01}\\
  \mathbb{P}(X_t^{HL}=x;1)=&\mathbb{P}(Y_t^L=-x-1;0)+\mathbb{P}(Y_t^L=x;1),\label{eq:lem_160612_03_02}\\
  \mathbb{P}(X_t^{HL}=x)=&\mathbb{P}(Y_t^L=-x-1)+\mathbb{P}(Y_t^L=x),\label{eq:lem_160612_03_03}
 \end{align}
 which hold for $x=0,1,2,\ldots$.
\end{lem}

\begin{proof}{%
Noting that the complex numbers $\gamma_t(x)$ and $\delta_t(x)$ stay in the set of real numbers because of the initial state in Eq.~\eqref{eq:ass_lem_160612_03}, Lemma \ref{lem:160612_01} finds
\begin{align}
 |\alpha_t(x)|^2=&|\gamma_t(x)|^2+|\delta_t(-x-1)|^2,\\
 |\beta_t(x)|^2=&|\gamma_t(-x-1)|^2+|\delta_t(x)|^2.
\end{align}
Expressing these equations with the probability distributions, we realize the statement of Lemma \ref{lem:160612_03}.
}
\end{proof}
\bigskip

The numerical experiments carried out in Fig. \ref{fig:160719_01} support the validity of Lemma \ref{lem:160612_03}.
The bars depict the distributions of the quantum walk on the half line, and the circles are estimated by the right hand sides of Eqs.~\eqref{eq:lem_160612_03_01}, \eqref{eq:lem_160612_03_02}, and \eqref{eq:lem_160612_03_03}.
\begin{figure}[h]
\begin{center}
 \begin{minipage}{35mm}
  \begin{center}
   \includegraphics[scale=0.25]{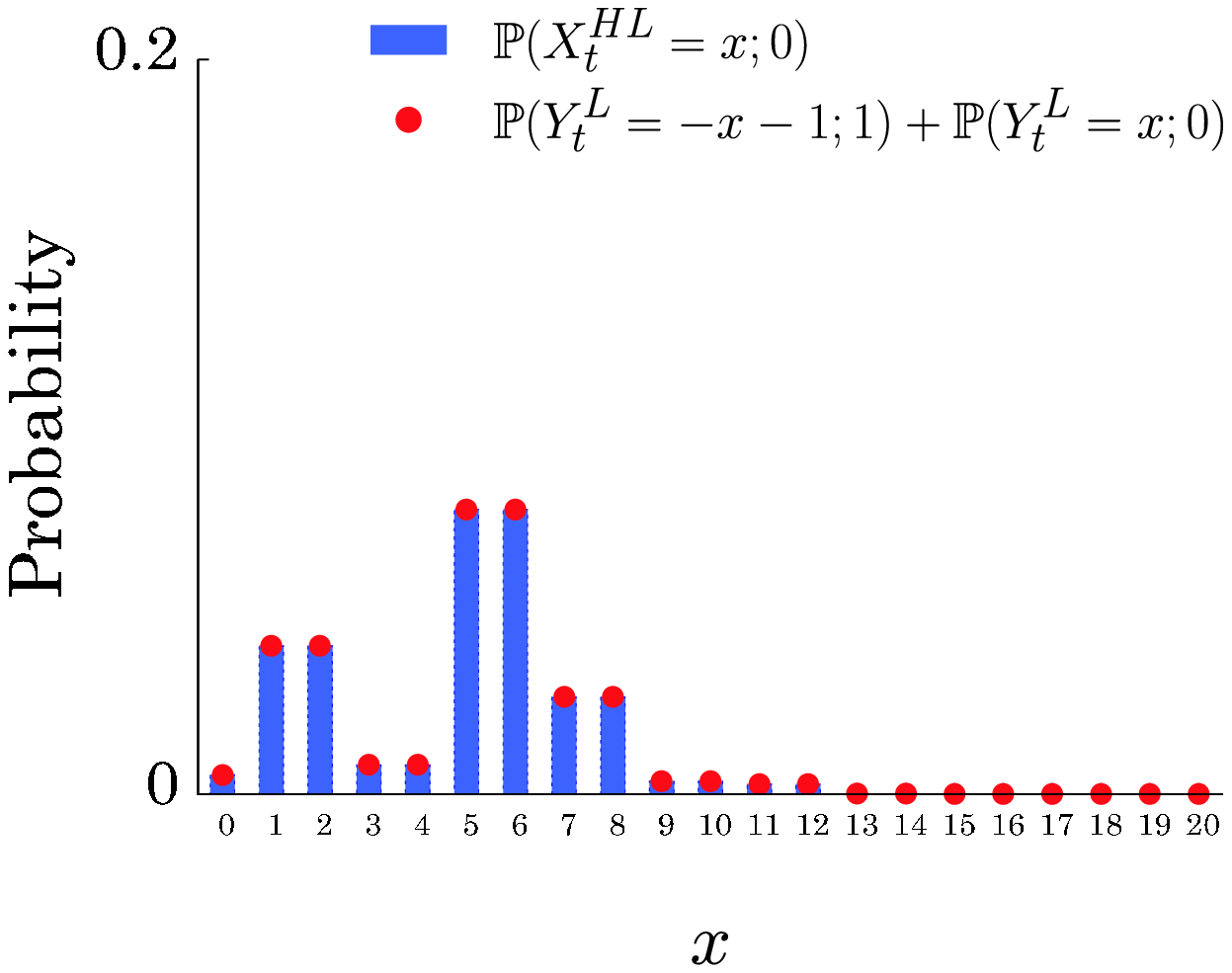}\\[2mm]
  (a) $\mathbb{P}(X_{20}^{HL}=x;0)$
  \end{center}
 \end{minipage}
 \begin{minipage}{35mm}
  \begin{center}
   \includegraphics[scale=0.25]{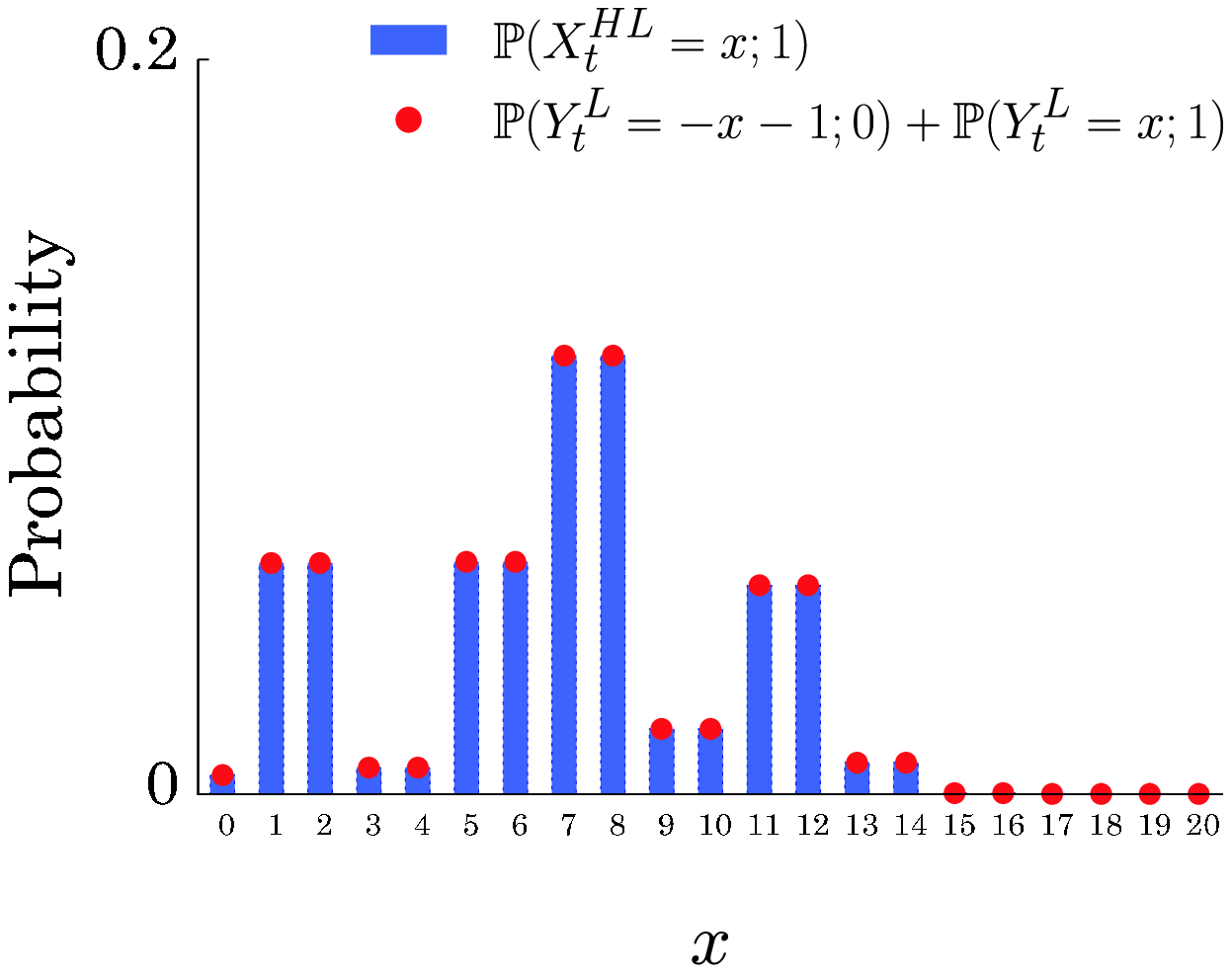}\\[2mm]
  (b) $\mathbb{P}(X_{20}^{HL}=x;1)$
  \end{center}
 \end{minipage}
 \begin{minipage}{35mm}
  \begin{center}
   \includegraphics[scale=0.25]{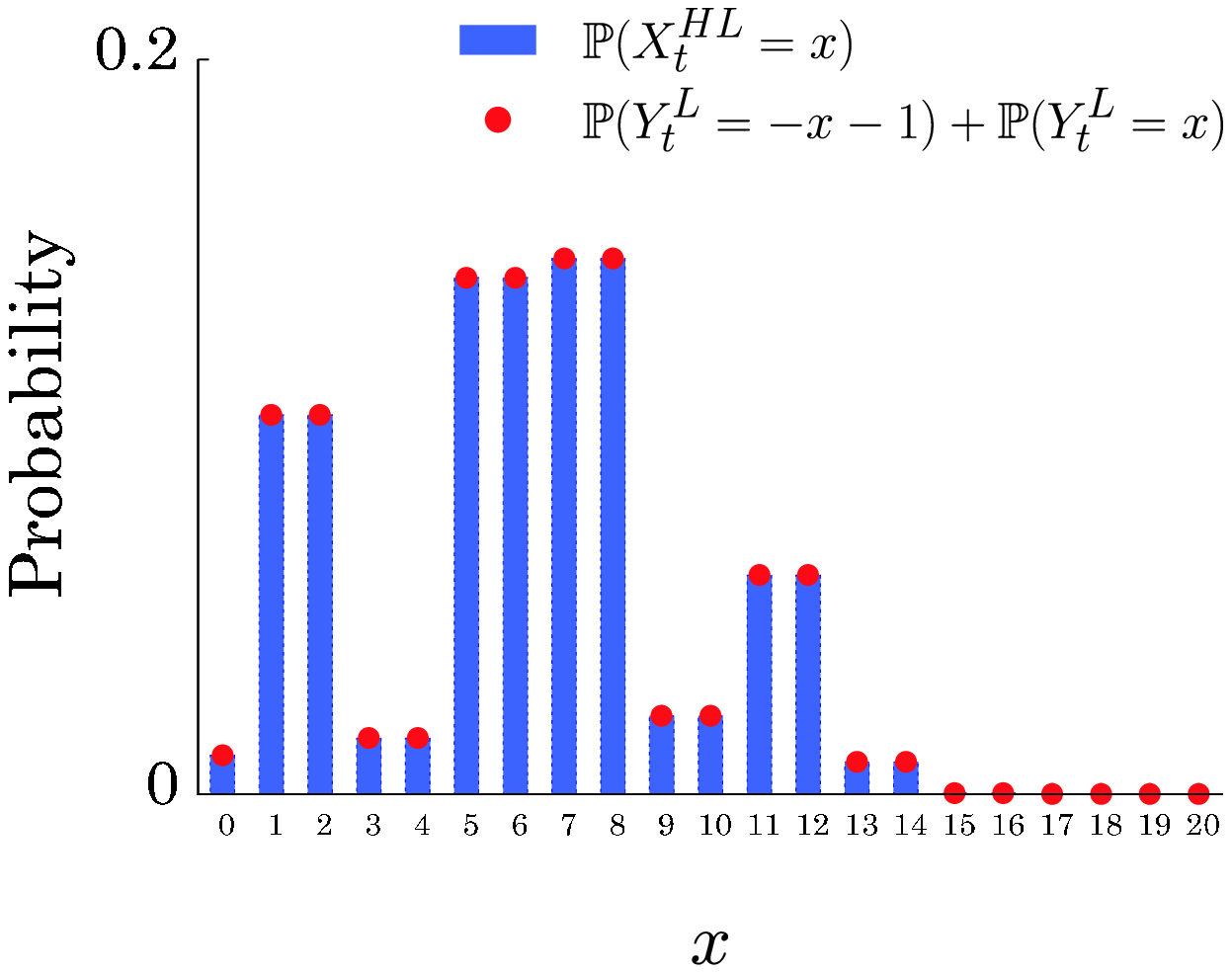}\\[2mm]
  (c) $\mathbb{P}(X_{20}^{HL}=x)$
  \end{center}
 \end{minipage}
\vspace{5mm}
\caption{$\theta_1=\pi/3,\,\theta_2=\pi/4$ : The blue bars show the distributions of the quantum walk on the half line at $t=20$. The red circles are obtained from the right hand sides of Eqs.~\eqref{eq:lem_160612_03_01}, \eqref{eq:lem_160612_03_02}, and \eqref{eq:lem_160612_03_03} as $t=20$.}
\label{fig:160719_01}
\end{center}
\end{figure}

Now, recalling the subscription $\xi\in\left\{1,2\right\}$ such that $|c_\xi|=\min\left\{|c_1|,|c_2|\right\}$, we reach our target, that is, limit distributions of the time-dependent quantum walk on the half line.
\begin{thm}
\label{th:limit}
 Assume that $\theta_1, \theta_2\neq 0,\pi/2,\pi,3\pi/2$.
 For a real number $x$, we have
 \begin{align}
  \lim_{t\to\infty}\mathbb{P}\left(\frac{X_t^{HL}}{t}\leq x;0\right)=&\int_{-\infty}^x \frac{|s_\xi|}{\pi(1+y)\sqrt{c_\xi^2-y^2}}I_{[0,|c_\xi|)}(y)\,dy,\label{eq:th_limit_01}\\
  \lim_{t\to\infty}\mathbb{P}\left(\frac{X_t^{HL}}{t}\leq x;1\right)=&\int_{-\infty}^x \frac{|s_\xi|}{\pi(1-y)\sqrt{c_\xi^2-y^2}}I_{[0,|c_\xi|)}(y)\,dy,\label{eq:th_limit_02}\\
  \lim_{t\to\infty}\mathbb{P}\left(\frac{X_t^{HL}}{t}\leq x\right)=&\int_{-\infty}^x \frac{2|s_\xi|}{\pi(1-y^2)\sqrt{c_\xi^2-y^2}}I_{[0,|c_\xi|)}(y)\,dy.\label{eq:th_limit_03}
 \end{align}
\end{thm}

\begin{proof}{%
With Lemma \ref{lem:160612_03}, we can derive these limit distributions.
When the quantum walk on the line starts with the initial state
\begin{equation}
  \ket{\Phi_0}=\ket{-1}\otimes\Bigl(\Im(\beta)\ket{0}-\Im(\alpha)\ket{1}\Bigr)+\ket{0}\otimes\Bigl(\Re(\alpha)\ket{0}+\Re(\beta)\ket{1}\Bigr),
 \end{equation}
the function $\eta(x)$, which is a part of the limit density functions of the walk in Eq.~\eqref{eq:160810_eta}, has the representation
\begin{align}
 \eta(x)&=1-\left\{\Re(\alpha)^2-\Im(\alpha)^2-\Bigl(\Re(\beta)^2-\Im(\beta)^2\Bigr)+\frac{2s_1\Bigl(\Re(\alpha)\Re(\beta)-\Im(\alpha)\Im(\beta)\Bigr)}{c_1}\right\}x\nonumber\\
 &=1-\left(\Re(\alpha^2)-\Re(\beta^2)+\frac{2s_1\Re(\alpha\beta)}{c_1}\right)x
 =1-\Re\left(\alpha^2-\beta^2+\frac{2s_1}{c_1}\alpha\beta\right)x.
\end{align}
Keeping in mind the relation shown in Eq.~\eqref{eq:lem_160612_03_01}, we make a computation of the long-time limit probability law that the quantum walker on the half line is observed in inner state $0$.
For a non-negative real number $x$, the finding probability as $t\to\infty$ is given by the limit distributions of the quantum walk on the line, 
\begin{align}
 &\lim_{t\to\infty}\mathbb{P}\left(\frac{X_t^{HL}}{t}\leq x;0\right)
 =\lim_{t\to\infty}\mathbb{P}\left(-x\leq\frac{Y_t^L}{t}<0;1\right)+\lim_{t\to\infty}\mathbb{P}\left(0\leq\frac{Y_t^L}{t}\leq x;0\right)\nonumber\\
 =&\int_{-x}^0\frac{|s_\xi|}{2\pi (1-y)\sqrt{c_\xi^2-y^2}}\left\{1-\Re\left(\alpha^2-\beta^2+\frac{2s_1}{c_1}\alpha\beta\right)y\right\}I_{(-|c_\xi|,\,|c_\xi|)}(y)\,dy\nonumber\\ 
 &+\int_0^x\frac{|s_\xi|}{2\pi (1+y)\sqrt{c_\xi^2-y^2}}\left\{1-\Re\left(\alpha^2-\beta^2+\frac{2s_1}{c_1}\alpha\beta\right)y\right\}I_{(-|c_\xi|,\,|c_\xi|)}(y)\,dy\nonumber\\
 =&\int_0^x\frac{|s_\xi|}{2\pi (1+y)\sqrt{c_\xi^2-y^2}}\left\{1+\Re\left(\alpha^2-\beta^2+\frac{2s_1}{c_1}\alpha\beta\right)y\right\}I_{(-|c_\xi|,\,|c_\xi|)}(y)\,dy\nonumber\\
 &+\int_0^x\frac{|s_\xi|}{2\pi (1+y)\sqrt{c_\xi^2-y^2}}\left\{1-\Re\left(\alpha^2-\beta^2+\frac{2s_1}{c_1}\alpha\beta\right)y\right\}I_{(-|c_\xi|,\,|c_\xi|)}(y)\,dy\nonumber\\ 
 =&\int_0^x\frac{|s_\xi|}{\pi (1+y)\sqrt{c_\xi^2-y^2}}I_{(-|c_\xi|,\,|c_\xi|)}(y)\,dy
 =\int_{-\infty}^x\frac{|s_\xi|}{\pi (1+y)\sqrt{c_\xi^2-y^2}}I_{[0,\,|c_\xi|)}(y)\,dy.
\end{align}
The other finding probabilities come from Eqs.~\eqref{eq:lem_160612_03_02} and \eqref{eq:lem_160612_03_03} as well,
\begin{align}
 \lim_{t\to\infty}\mathbb{P}\left(\frac{X_t^{HL}}{t}\leq x;1\right)
 =&\lim_{t\to\infty}\mathbb{P}\left(-x\leq\frac{Y_t^L}{t}< 0;0\right)+\lim_{t\to\infty}\mathbb{P}\left(0\leq\frac{Y_t^L}{t}\leq x;1\right)\nonumber\\
 =&\int_{-\infty}^x\frac{|s_\xi|}{\pi (1-y)\sqrt{c_\xi^2-y^2}}I_{[0,\,|c_\xi|)}(y)\,dy,\\[3mm]
 \lim_{t\to\infty}\mathbb{P}\left(\frac{X_t^{HL}}{t}\leq x\right)
 =&\lim_{t\to\infty}\mathbb{P}\left(-x\leq\frac{Y_t^L}{t}< 0\right)+\lim_{t\to\infty}\mathbb{P}\left(0\leq\frac{Y_t^L}{t}\leq x\right)\nonumber\\
 =&\int_{-\infty}^x\frac{2|s_\xi|}{\pi (1-y^2)\sqrt{c_\xi^2-y^2}}I_{[0,\,|c_\xi|)}(y)\,dy.
\end{align}
Since the walker on the half line, which consists of the non-negative positions in this paper, never be observed at the negative positions, these equations are also true for $x<0$ due to the presence of the indicator function $I_{[0,\,|c_\xi|)}(y)$.
Hence, one may arrive at Theorem \ref{th:limit}.
}
\end{proof}
\bigskip

Taking a good look at the limit distributions in Theorem \ref{th:limit}, we realize that they do not depend on the complex numbers $\alpha$ and $\beta$ which produce the localized initial state of the quantum walk on the half line.
Also, while the quantum walk on the half line is operated by both $U_1$ and $U_2$, its limit distributions are determined by either $U_1$ or $U_2$, but not both, because the index $\xi$ is definined by $\cos\theta_\xi=\min\left\{|\cos\theta_1|,\,|\cos\theta_2|\right\}$.
One can, hence, get approximations independent from the complex numbers $\alpha$ and $\beta$,
\begin{align}
 \lim_{t\to\infty}\mathbb{P}\left(X_t^{HL}=x;0\right)\sim &\frac{|s_\xi|t}{\pi(t+x)\sqrt{(c_\xi t)^2-x^2}}I_{[0,|c_\xi|t)}(x),\label{eq:app_state_0}\\
 \lim_{t\to\infty}\mathbb{P}\left(X_t^{HL}=x;1\right)\sim&\frac{|s_\xi|t}{\pi(t-x)\sqrt{(c_\xi t)^2-x^2}}I_{[0,|c_\xi|t)}(x),\label{eq:app_state_1}\\
 \lim_{t\to\infty}\mathbb{P}\left(X_t^{HL}=x\right)\sim&\frac{2|s_\xi|t^2}{\pi(t^2-x^2)\sqrt{(c_\xi t)^2-x^2}}I_{[0,|c_\xi|t)}(x).\label{eq:app_prob}
\end{align}
Since the right hand sides have the parameters $c_\xi=\cos\theta_\xi$ and $s_\xi=\sin\theta_\xi$, it is figured out that the asymptotic behavior of the quantum walk as $t\to\infty$ is featured by only one of the unitary operations $U_1$ and $U_2$.
The approximations indeed catch the features of the probability distributions $\mathbb{P}(X_t^{HL}=x;0), \mathbb{P}(X_t^{HL}=x;1)$, and $\mathbb{P}(X_t^{HL}=x)$ at time $500$, as shown in Figures \ref{fig:160720_04}, \ref{fig:160718_28}, and \ref{fig:160718_31}.
These pictures show up when the parameters of the unitary operations $U_1$ and $U_2$ are set as $\theta_1=\pi/3$ and $\theta_2=\pi/4$ respectively.
Note that all the graphs plotted by Eq.~\eqref{eq:app_state_0} in the pictures (a), represented by circles, are same, and it is also said for the pictures (b) and (c) because of Eqs.~\eqref{eq:app_state_1} and \eqref{eq:app_prob}. 
\begin{figure}[h]
\begin{center}
 \begin{minipage}{35mm}
  \begin{center}
   \includegraphics[scale=0.3]{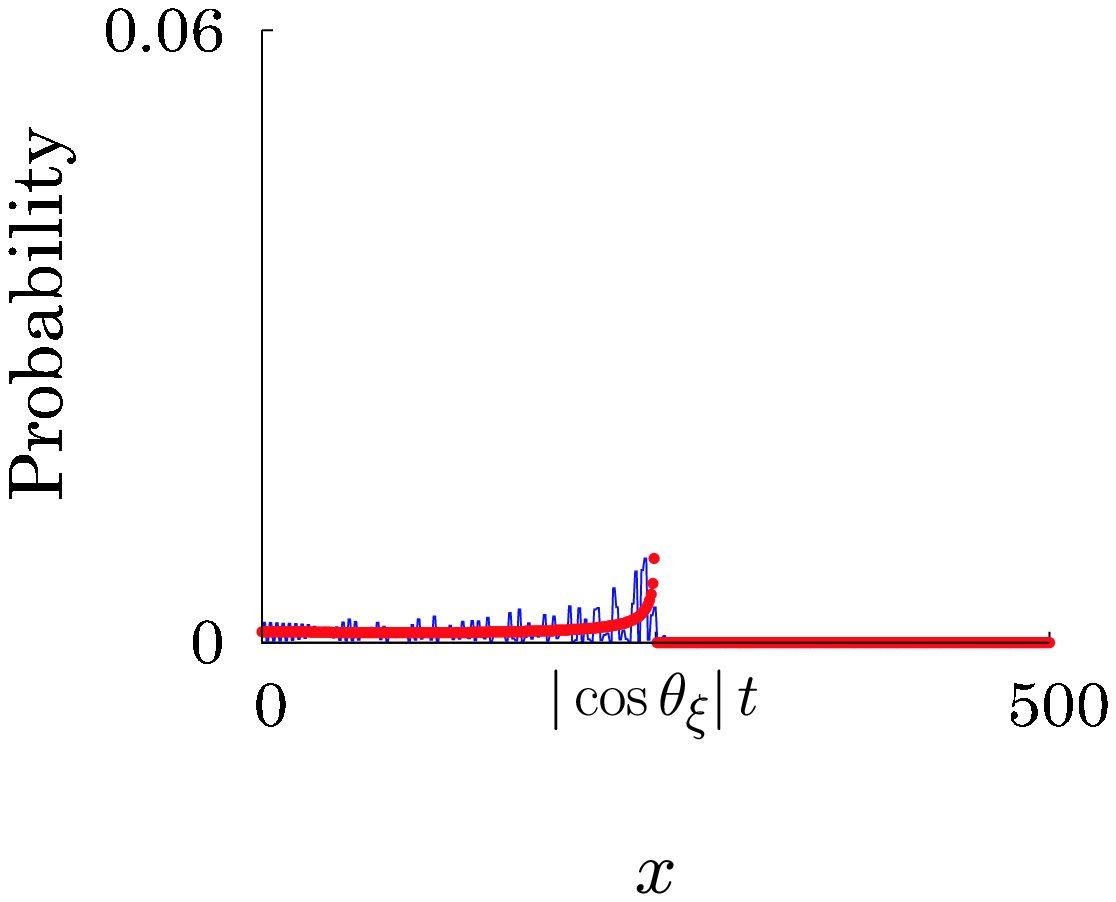}\\[2mm]
  (a) $\mathbb{P}(X_t^{HL}=x;0)$
  \end{center}
 \end{minipage}
 \begin{minipage}{35mm}
  \begin{center}
   \includegraphics[scale=0.3]{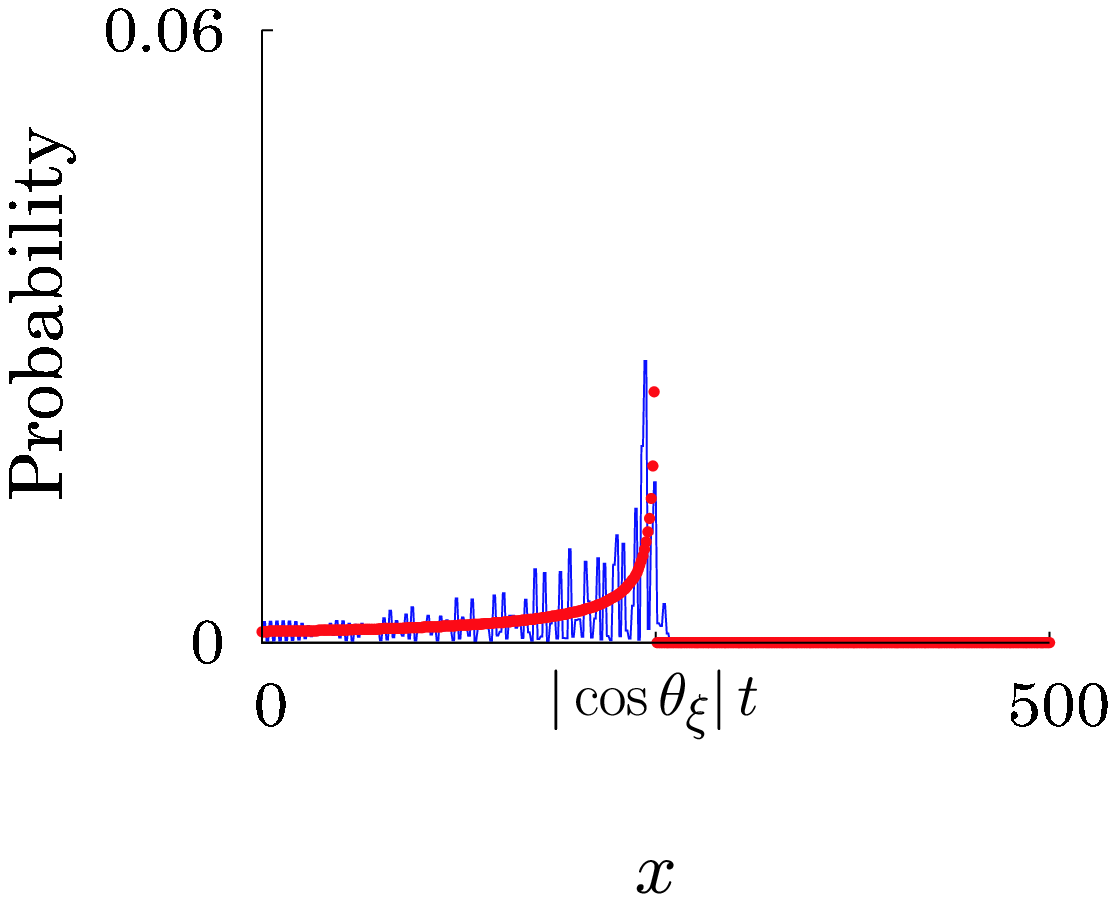}\\[2mm]
  (b) $\mathbb{P}(X_t^{HL}=x;1)$
  \end{center}
 \end{minipage}
 \begin{minipage}{35mm}
  \begin{center}
   \includegraphics[scale=0.3]{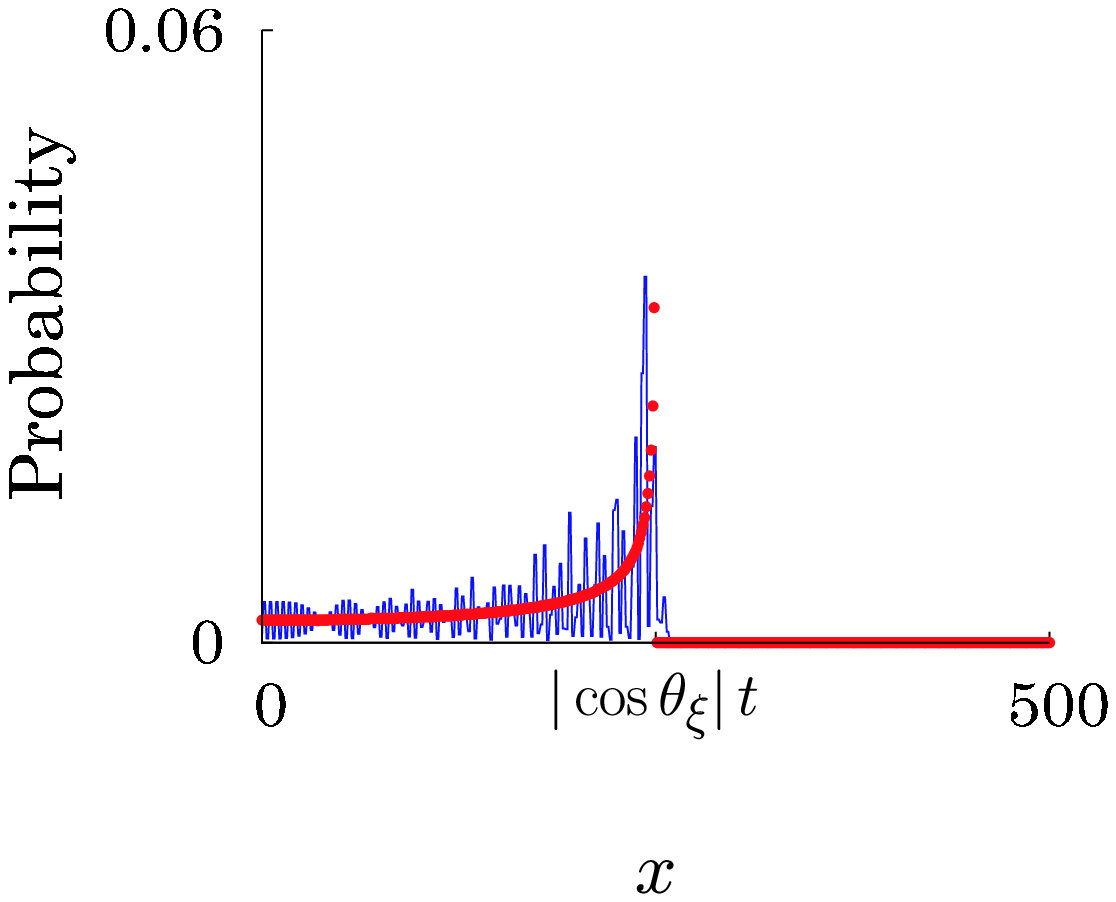}\\[2mm]
  (c) $\mathbb{P}(X_t^{HL}=x)$
  \end{center}
 \end{minipage}
\vspace{5mm}
\caption{$\alpha=1/\sqrt{2},\, \beta=i/\sqrt{2}$ : The blue lines show the probability distributions of the time-dependent quantum walk on the half line at time $500$. The red circles indicate values obtained from the approximations in Eqs.~\eqref{eq:app_state_0}, \eqref{eq:app_state_1}, and \eqref{eq:app_prob} as $t=500$. ($\theta_1=\pi/3,\, \theta_2=\pi/4$)}
\label{fig:160720_04}
\end{center}
\end{figure}

\begin{figure}[h]
\begin{center}
 \begin{minipage}{35mm}
  \begin{center}
   \includegraphics[scale=0.3]{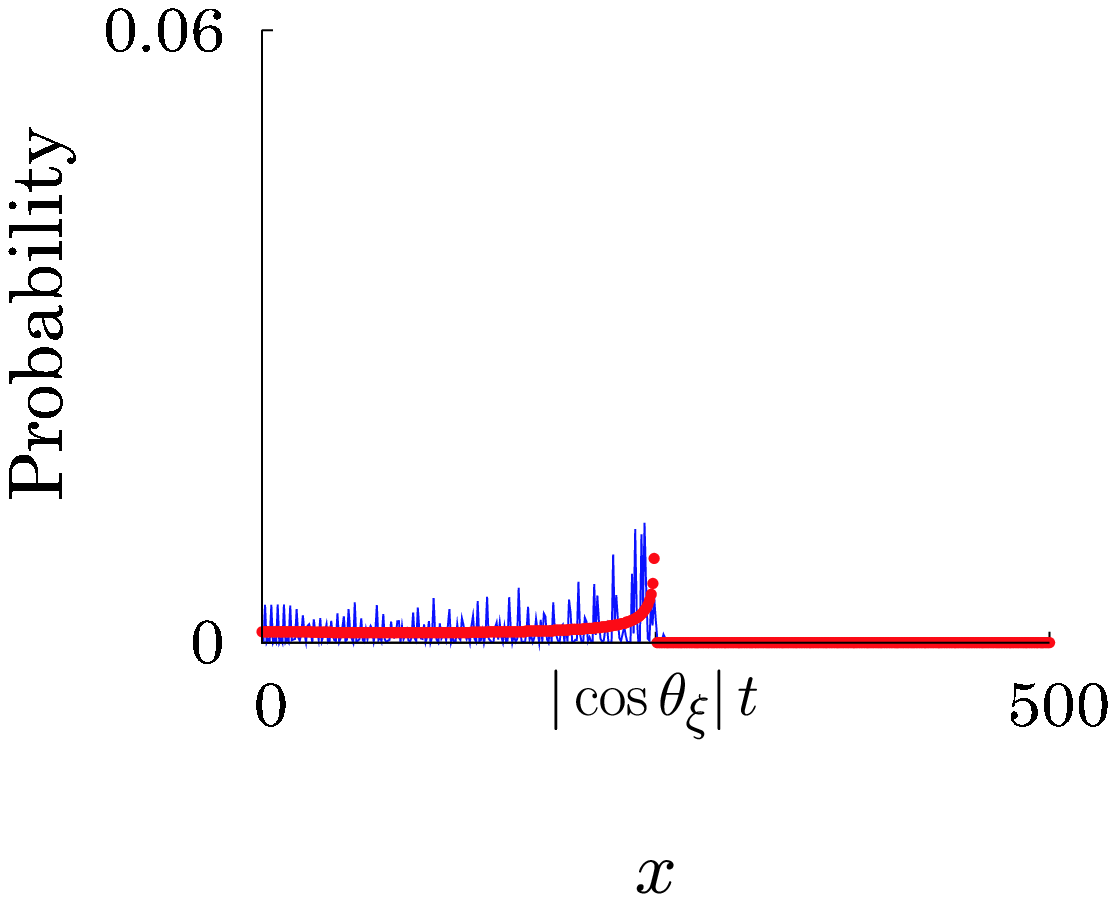}\\[2mm]
  (a) $\mathbb{P}(X_t^{HL}=x;0)$
  \end{center}
 \end{minipage}
 \begin{minipage}{35mm}
  \begin{center}
   \includegraphics[scale=0.3]{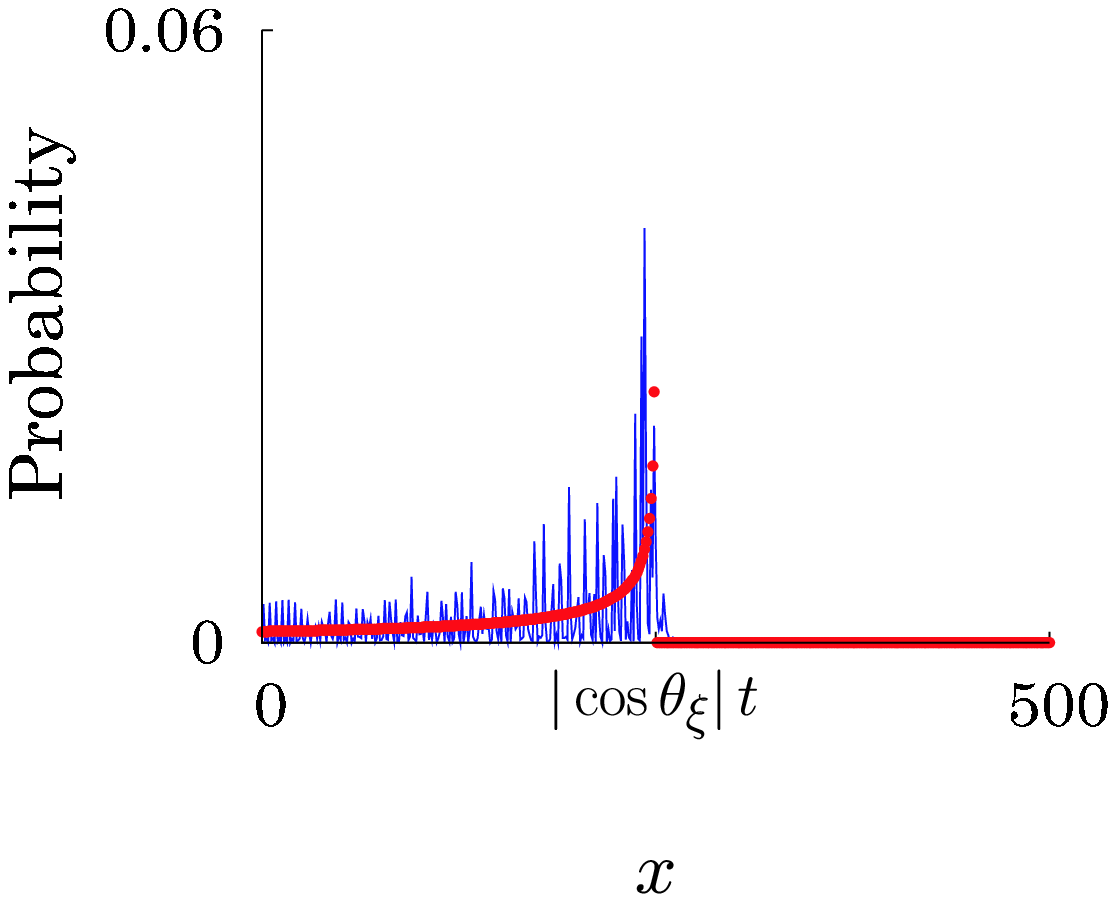}\\[2mm]
  (b) $\mathbb{P}(X_t^{HL}=x;1)$
  \end{center}
 \end{minipage}
 \begin{minipage}{35mm}
  \begin{center}
   \includegraphics[scale=0.3]{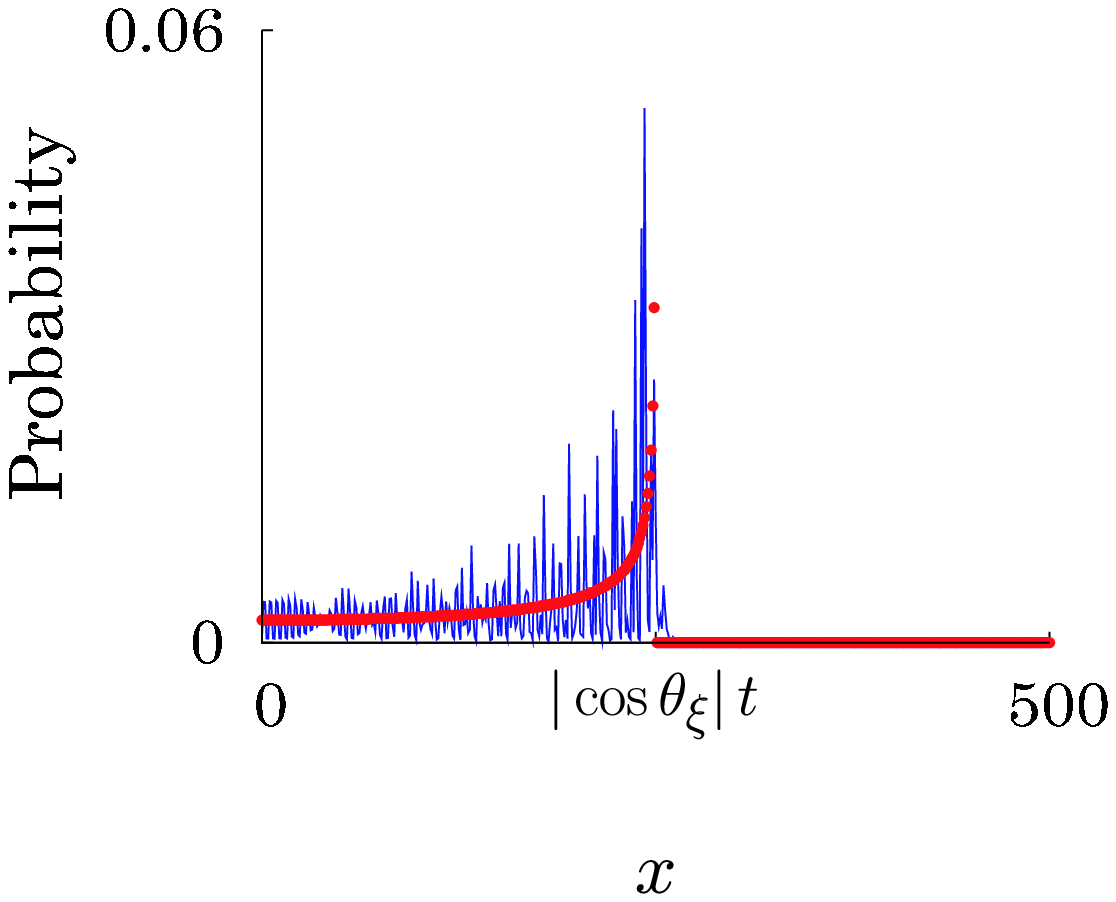}\\[2mm]
  (c) $\mathbb{P}(X_t^{HL}=x)$
  \end{center}
 \end{minipage}
\vspace{5mm}
\caption{$\alpha=1,\, \beta=0$ : The blue lines show the probability distributions of the time-dependent quantum walk on the half line at time $500$. The red circles indicate values obtained from the approximations in Eqs.~\eqref{eq:app_state_0}, \eqref{eq:app_state_1}, and \eqref{eq:app_prob} as $t=500$. ($\theta_1=\pi/3,\, \theta_2=\pi/4$)}
\label{fig:160718_28}
\end{center}
\end{figure}

\begin{figure}[h]
\begin{center}
 \begin{minipage}{35mm}
  \begin{center}
   \includegraphics[scale=0.3]{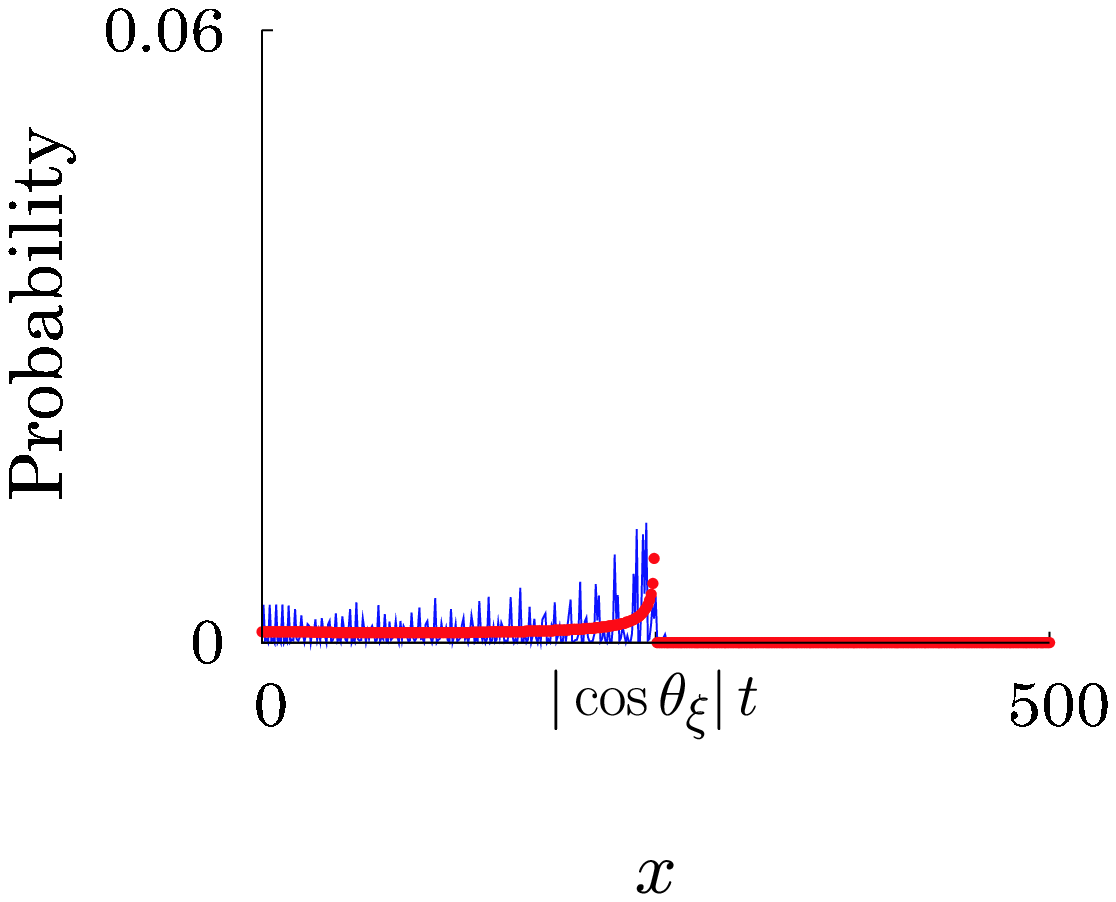}\\[2mm]
  (a) $\mathbb{P}(X_t^{HL}=x;0)$
  \end{center}
 \end{minipage}
 \begin{minipage}{35mm}
  \begin{center}
   \includegraphics[scale=0.3]{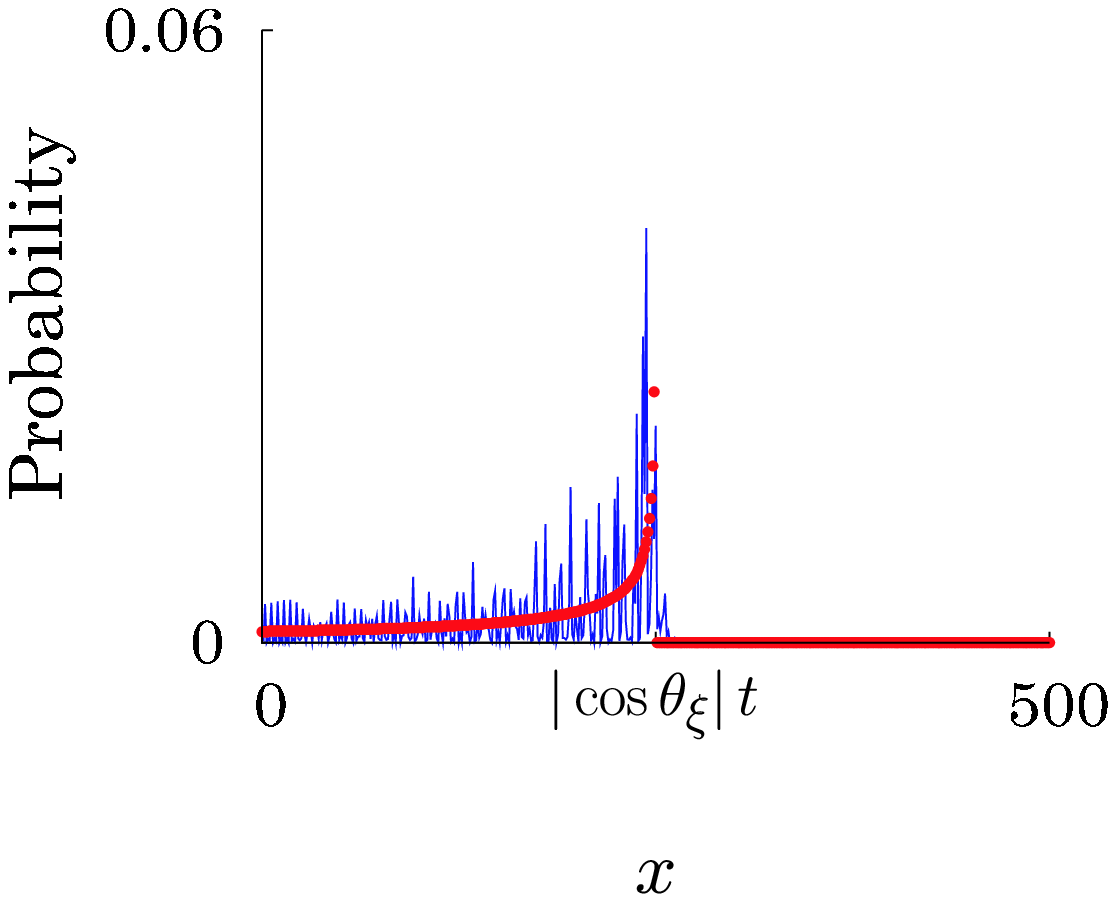}\\[2mm]
  (b) $\mathbb{P}(X_t^{HL}=x;1)$
  \end{center}
 \end{minipage}
 \begin{minipage}{35mm}
  \begin{center}
   \includegraphics[scale=0.3]{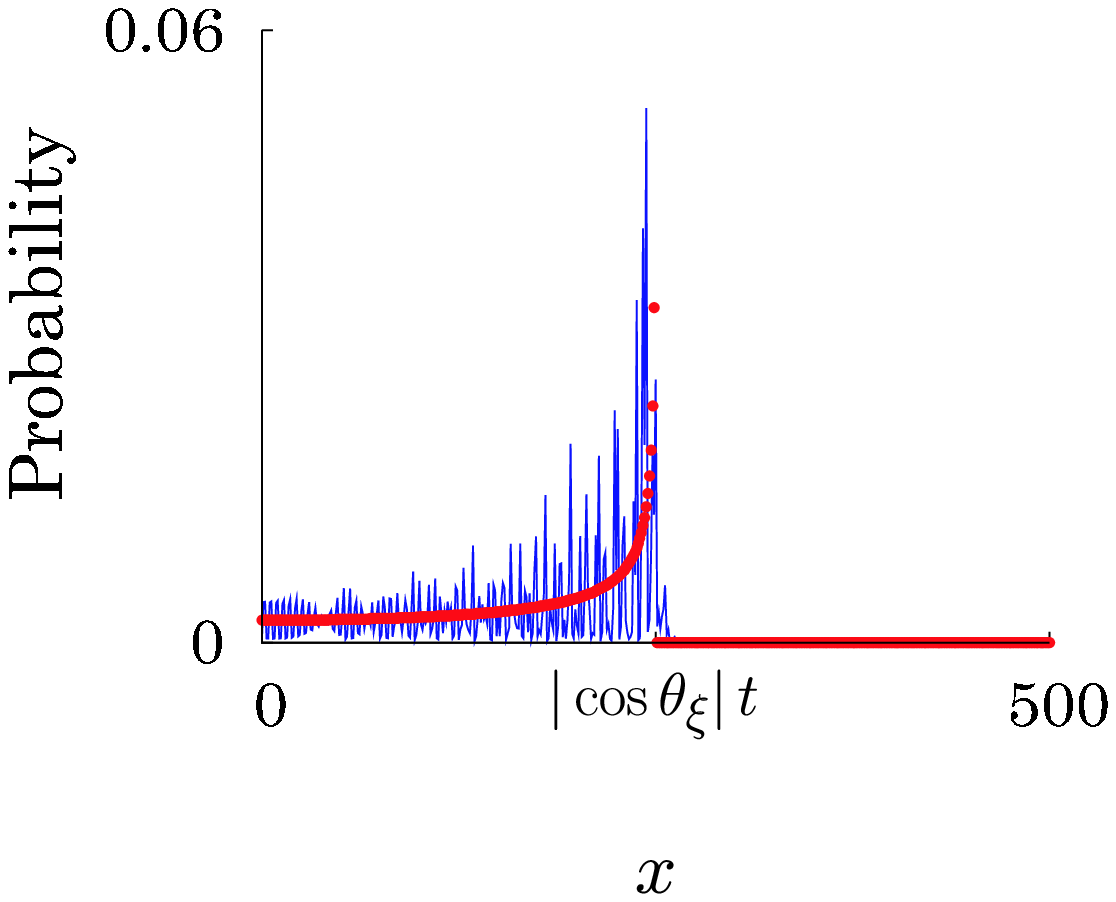}\\[2mm]
  (c) $\mathbb{P}(X_t^{HL}=x)$
  \end{center}
 \end{minipage}
\vspace{5mm}
\caption{$\alpha=0,\, \beta=1$ : The blue lines show the probability distributions of the time-dependent quantum walk on the half line at time $500$. The red circles indicate values obtained from the approximations in Eqs.~\eqref{eq:app_state_0}, \eqref{eq:app_state_1}, and \eqref{eq:app_prob} as $t=500$. ($\theta_1=\pi/3,\, \theta_2=\pi/4$)}
\label{fig:160718_31}
\end{center}
\end{figure}

\clearpage

\section{Time-independent quantum walk on the half line}

We see exact representations for the probability distributions of a time-independent quantum walk on the half line in Machida~\cite{Machida2016}.
They were, however, the result for a special initial state.
On the other hand, since Lemma \ref{lem:160612_03} was also available for the time-independent quantum walk on the half line, one can say the exact representations for the probability distributions of the time-independent quantum walk with a general localized initial state. 
If a value $\theta\in[0,2\pi)$ is substituted to both the parameters $\theta_1$ and $\theta_2$, the quantum walk on the half line becomes a time-independent walk.
Making the most of the result given in Konno~\cite{Konno2002a} under the assumption $\theta\neq 0,\pi/2,\pi,3\pi/2$, we can see representations for the positive values of the probability distribution $\mathbb{P}(X_t^{HL}=x)\,(t=1,2,\ldots)$.
For $m=1,2,\ldots,\left[t/2\right]$, writing $\cos\theta$ and $\sin\theta$ as $c$ and $s$ respectively, we have
\begin{align}
 &\mathbb{P}(X_t^{HL}=t-2m)\nonumber\\
 =&c^{2(t-1)}\sum_{j_1=1}^m\sum_{j_2=1}^m\left(-\frac{s^2}{c^2}\right)^{j_1+j_2}
 {m-1\choose j_1-1}{m-1\choose j_2-1}{t-m-1\choose j_1-1}{t-m-1\choose j_2-1}\nonumber\\
 &\times\left(\frac{1}{j_1 j_2}\right)\biggl[\left\{m^2c^2+(t-m)^2s^2-(j_1+j_2)(t-m)\right\}|\alpha|^2\nonumber\\
 &\qquad\qquad\qquad+\left\{m^2s^2+(t-m)^2c^2-(j_1+j_2)m\right\}|\beta|^2\nonumber\\
 &\qquad\qquad\qquad+\frac{1}{s^2}\Bigl[\left\{(t-2m)(j_1+j_2)+2t(2m-t)s^2\right\}cs\Re(\alpha\overline{\beta})+j_1j_2\Bigr]
 \biggr],\label{eq:160803_01}\\[2mm]
 &\mathbb{P}(X_t^{HL}=t-2m-1)\nonumber\\
 =&c^{2(t-1)}\sum_{j_1=1}^m\sum_{j_2=1}^m\left(-\frac{s^2}{c^2}\right)^{j_1+j_2}
 {m-1\choose j_1-1}{m-1\choose j_2-1}{t-m-1\choose j_1-1}{t-m-1\choose j_2-1}\nonumber\\
 &\times\left(\frac{1}{j_1 j_2}\right)\biggl[\left\{m^2s^2+(t-m)^2c^2-(j_1+j_2)m\right\}|\alpha|^2\nonumber\\
 &\qquad\qquad\qquad+\left\{m^2c^2+(t-m)^2s^2-(j_1+j_2)(t-m)\right\}|\beta|^2\nonumber\\
 &\qquad\qquad\qquad-\frac{1}{s^2}\Bigl[\left\{(t-2m)(j_1+j_2)+2t(2m-t)s^2\right\}cs\Re(\alpha\overline{\beta})-j_1j_2\Bigr]
 \biggr],\label{eq:160803_02}\\[2mm]
 &\mathbb{P}(X_t^{HL}=t)=c^{2(t-1)}|s\alpha-c\beta|^2,\label{eq:160803_03}\\[2mm]
 &\mathbb{P}(X_t^{HL}=-t)=c^{2(t-1)}|c\alpha+s\beta|^2.\label{eq:160803_04}
\end{align}
Note that, for a real number $x$, the floor function $[x]$ outputs the maximum integer less than or equal to the real number $x$. 
The second equation above is good under the condition $m\leq (t-1)/2$, which comes from $t-2m-1\geq 0$.
If we put $\alpha=1/\sqrt{2}$ and $\beta=i/\sqrt{2}$, the representations of the probability distribution totally agree with the ones demonstrated in Machida~\cite{Machida2016}.
The values obtained from Eqs.~\eqref{eq:160803_01}--\eqref{eq:160803_04} completely match numerical experiments, as shown in Fig.~\ref{fig:160803_01}, 
\begin{figure}[h]
\begin{center}
 \begin{minipage}{35mm}
  \begin{center}
   \includegraphics[scale=0.25]{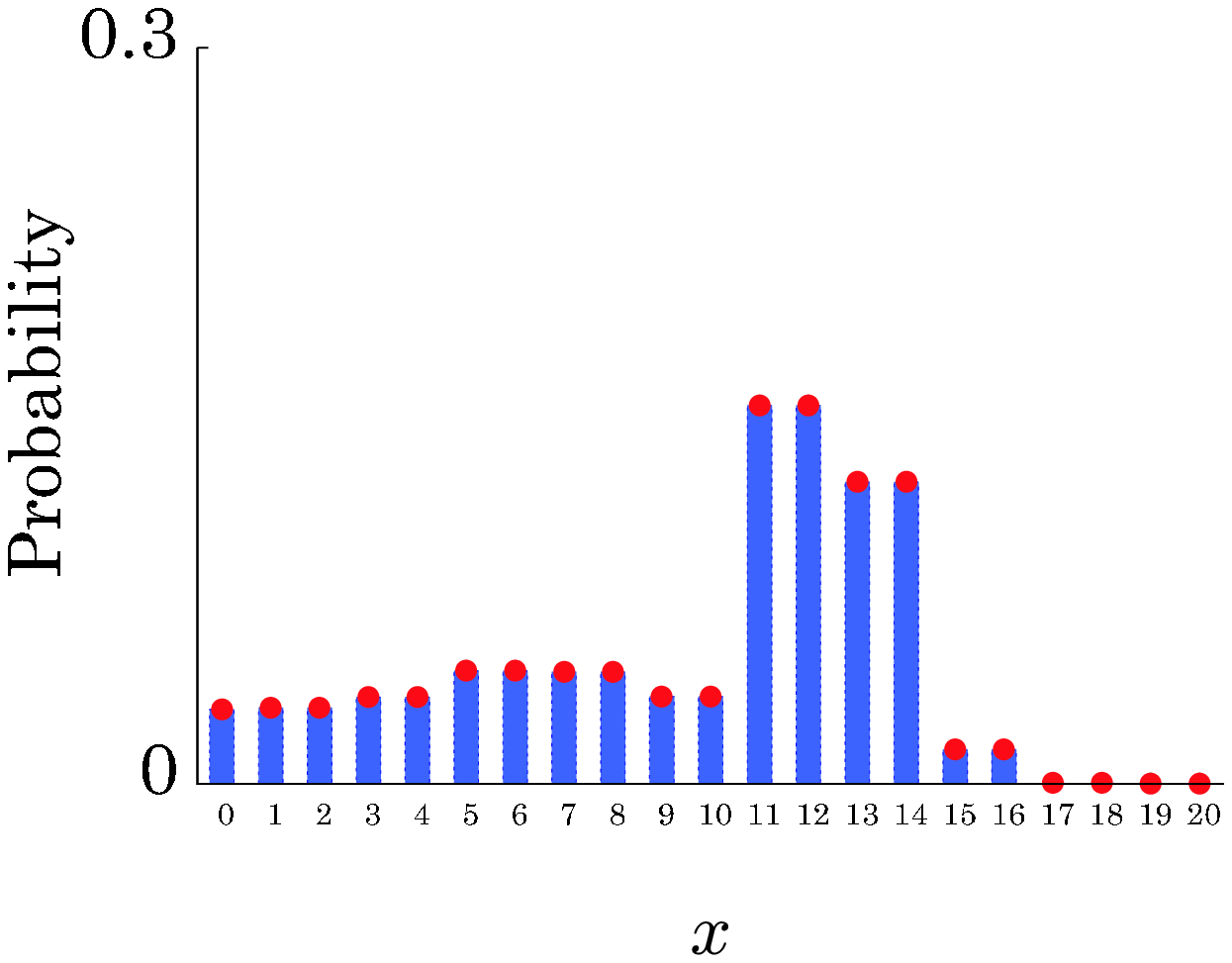}\\[2mm]
  (a) $\alpha=1/\sqrt{2},\,\beta=i/\sqrt{2}$
  \end{center}
 \end{minipage}
 \begin{minipage}{35mm}
  \begin{center}
   \includegraphics[scale=0.25]{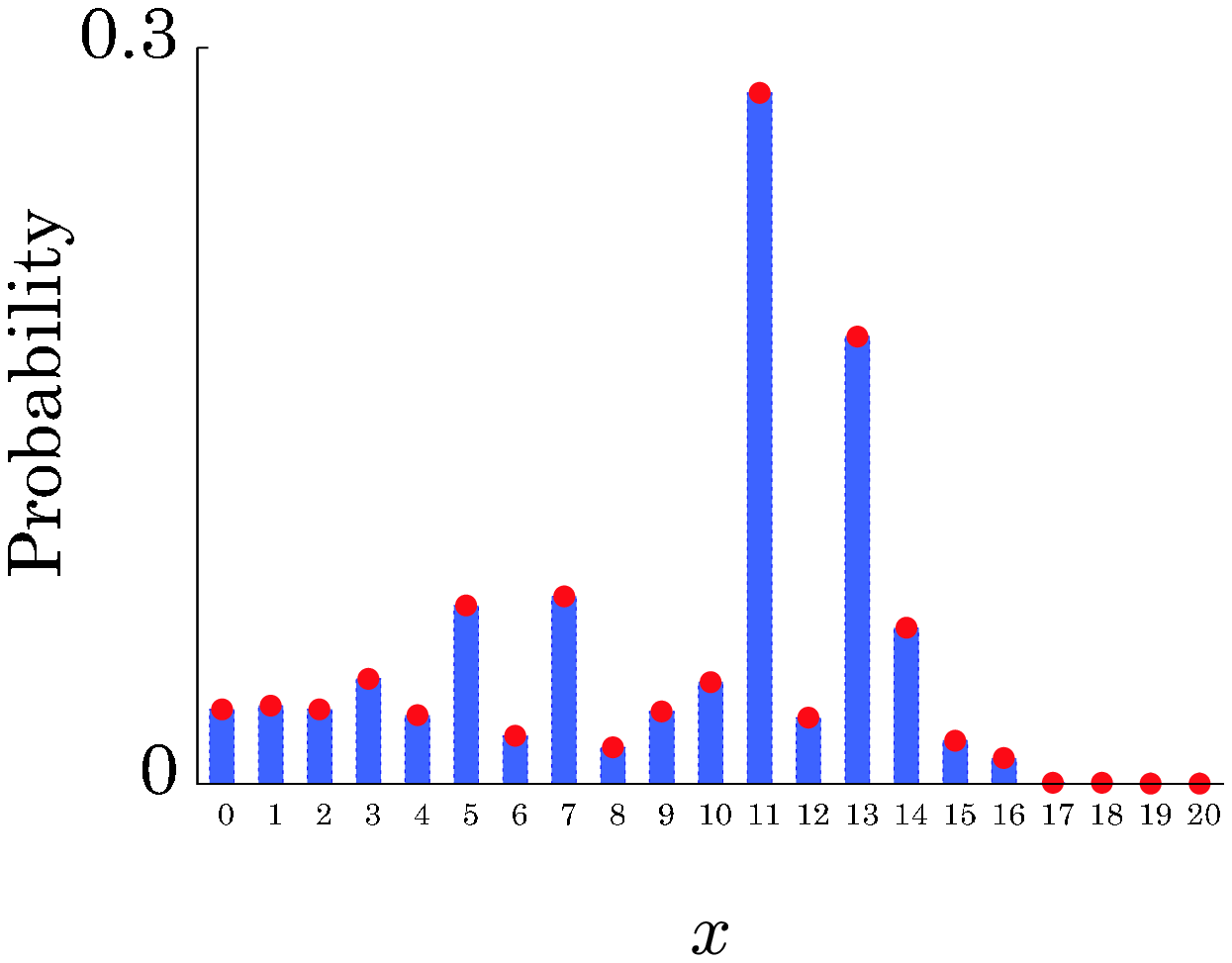}\\[2mm]
  (b) $\alpha=1,\,\beta=0$
  \end{center}
 \end{minipage}
 \begin{minipage}{35mm}
  \begin{center}
   \includegraphics[scale=0.25]{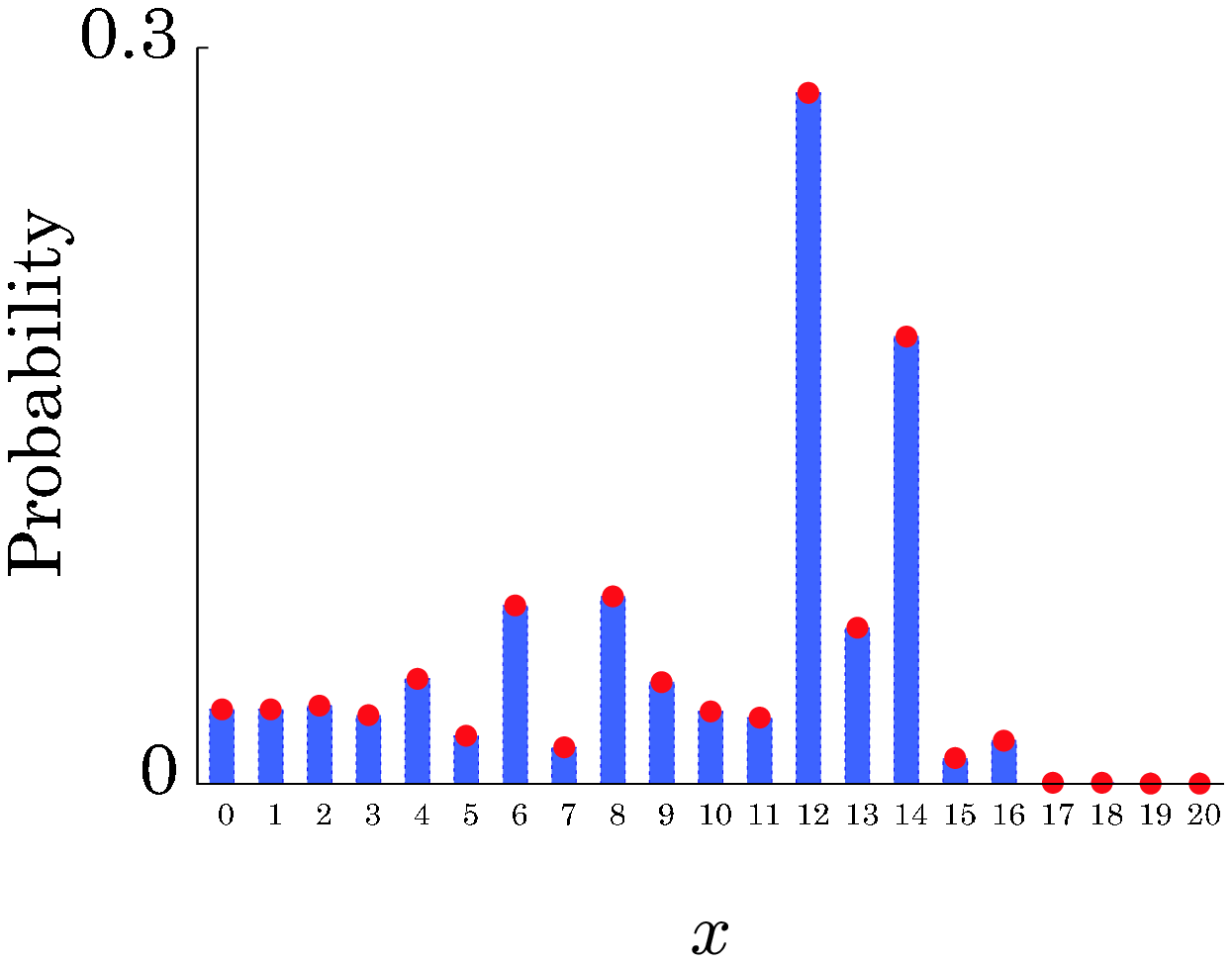}\\[2mm]
  (c) $\alpha=0,\,\beta=1$
  \end{center}
 \end{minipage}
\vspace{5mm}
\caption{$\theta=\pi/4$ : The blue bars represent $\mathbb{P}(X_{20}^{HL}=x)$ of the quantum walk on the half line at time $20$. The values computed from Eqs.~\eqref{eq:160803_01}--\eqref{eq:160803_04} as  $t=20$ are plotted with the red circles.}
\label{fig:160803_01}
\end{center}
\end{figure}

More importantly, the representations are related with the probability distribution of the time-independent quantum walk on the line with the localized initial state $\ket{0}\otimes\left(\alpha\ket{0}+\beta\ket{1}\right)$, which is different from the delocalized initial state in Eq.~\eqref{eq:L_initial_state}.
Once again, the complex numbers $\alpha$ and $\beta$ are supposed to satisfy the condition $|\alpha|^2+|\beta|^2=1$.
Let $\mathbb{P}(Z_t^L=x)\,(x\in\mathbb{Z})$ be the probability distribution of the quantum walk on the line with the localized initial state.
Then, looking at the representations for the probability distribution $\mathbb{P}(Z_t^L=x)$ shown in Konno~\cite{Konno2002a}, one figures out the relation $\mathbb{P}(X_t^{HL}=x)=\mathbb{P}(Z_t^L=-x-1)+\mathbb{P}(Z_t^L=x)\,(x\in\left\{0,1,2,\ldots\right\})$ for which we should note that since the quantum walker launches with the localized initial state, either $\mathbb{P}(Z_t^L=-x-1)$ or $\mathbb{P}(Z_t^L=x)$ is certainly equal to zero.
We, hence, can say that the probability distribution of the time-independent quantum walk on the half line with  the localized initial state $\ket{0}\otimes\left(\alpha\ket{0}+\beta\ket{1}\right)$ is totally described by the probability distribution of the time-independent quantum walk on the line with the same localized initial state $\ket{0}\otimes\left(\alpha\ket{0}+\beta\ket{1}\right)$.
Figure \ref{fig:160803_04} numerically shows that the relation $\mathbb{P}(X_t^{HL}=x)=\mathbb{P}(Z_t^L=-x-1)+\mathbb{P}(Z_t^L=x)$ is true.
\begin{figure}[h]
\begin{center}
 \begin{minipage}{35mm}
  \begin{center}
   \includegraphics[scale=0.25]{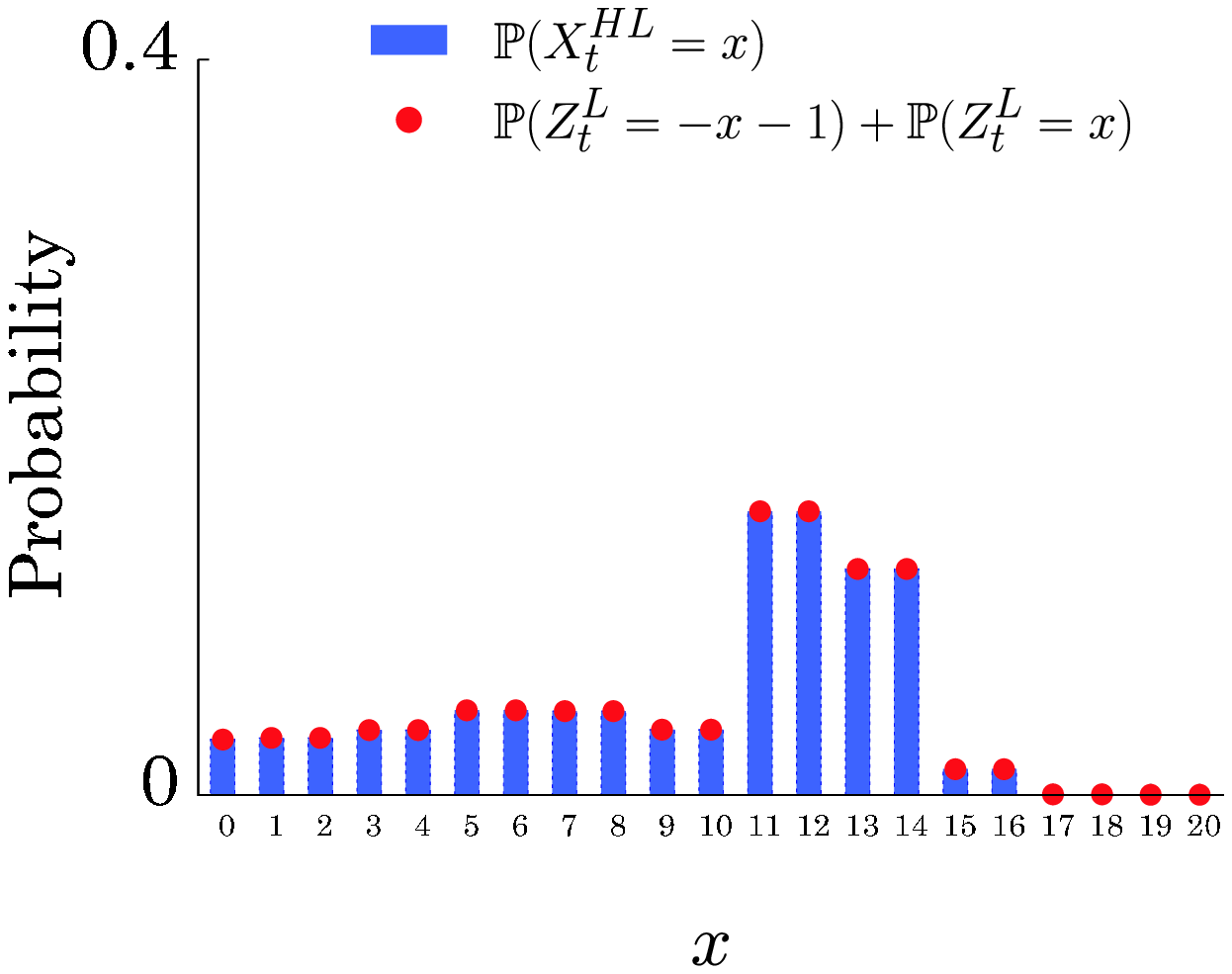}\\[2mm]
  (a) $\alpha=1/\sqrt{2},\,\beta=i/\sqrt{2}$
  \end{center}
 \end{minipage}
 \begin{minipage}{35mm}
  \begin{center}
   \includegraphics[scale=0.25]{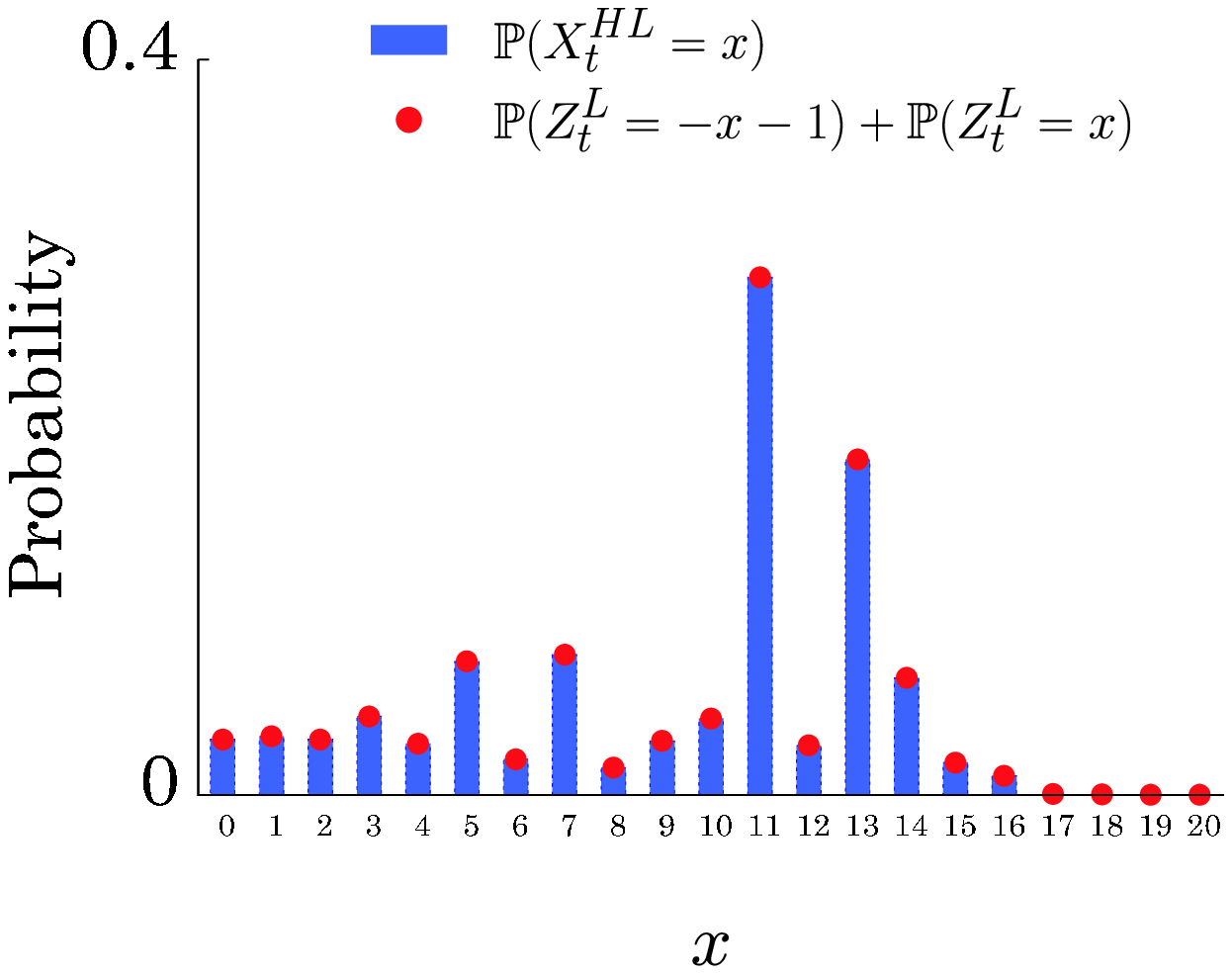}\\[2mm]
  (b) $\alpha=1,\,\beta=0$
  \end{center}
 \end{minipage}
 \begin{minipage}{35mm}
  \begin{center}
   \includegraphics[scale=0.25]{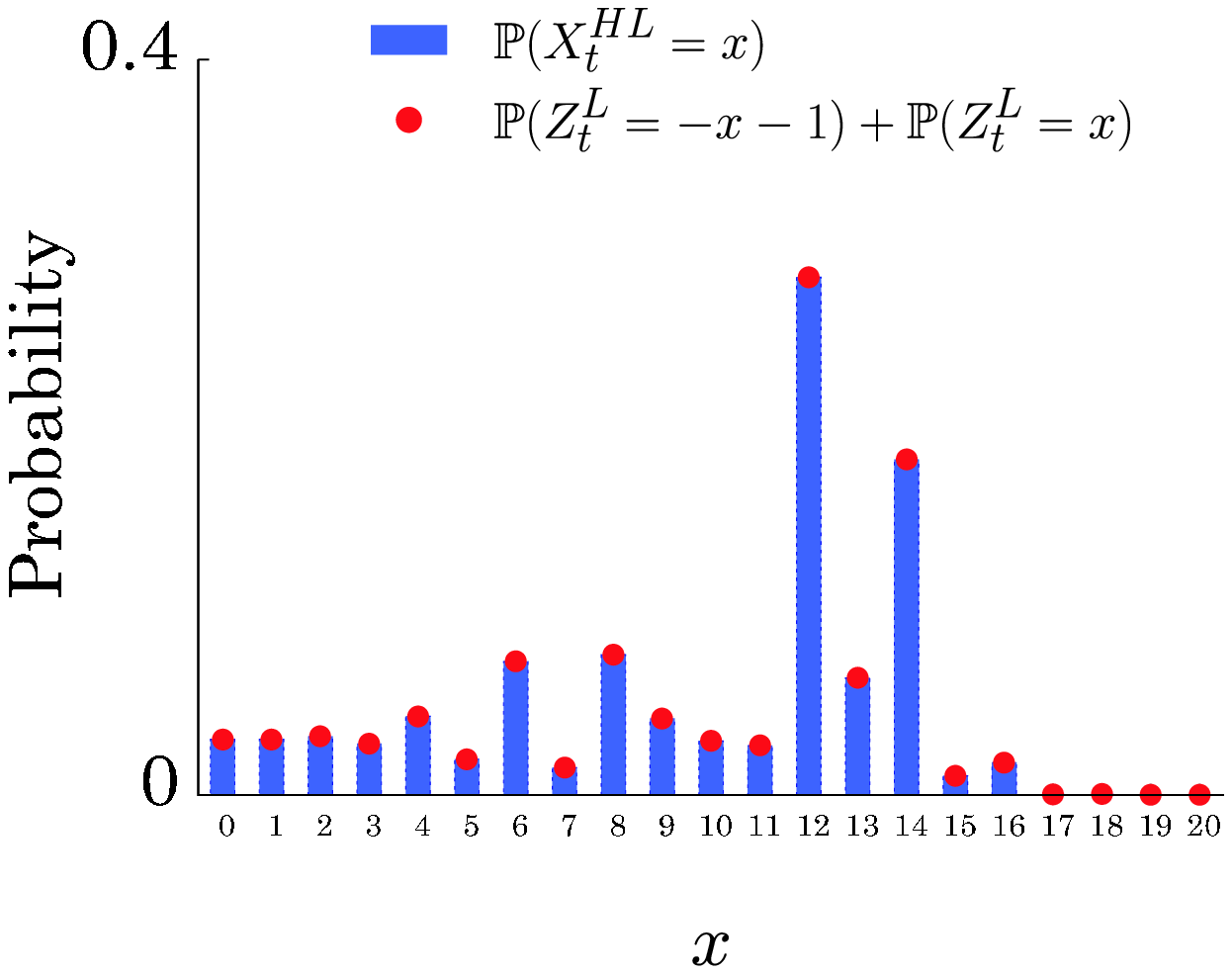}\\[2mm]
  (c) $\alpha=0,\,\beta=1$
  \end{center}
 \end{minipage}
\vspace{5mm}
\caption{$\theta=\pi/4$ : The blue bars represent $\mathbb{P}(X_{20}^{HL}=x)$ of the quantum walk on the half line at time $20$. The red circles represent the sum of two probabilities $\mathbb{P}(Z_{20}^L=-x-1)$ and  $\mathbb{P}(Z_{20}^L=x)$ in which $Z_{20}^L$ means the position of the quantum walker on the line at time $20$. Both quantum walks launch with the localized initial state $\ket{0}\otimes\left(\alpha\ket{0}+\beta\ket{1}\right)$ at time $0$.}
\label{fig:160803_04}
\end{center}
\end{figure}

\section{Summary}
\label{sec:summary}
In this paper we studied a time-dependent quantum walk which started off at the edge of the half line with a localized initial state and repeatedly got updated with two unitary operations cast to the walk alternately.
As Lemma \ref{lem:160612_01} mentioned, the quantum walk had a connection, at the level on amplitude, to a 2-period time-dependent quantum walk on the line with a delocalized initial state, and that fact worked for derivation of Theorem \ref{th:limit}, that is, limit distributions for the quantum walk on the half line. 
To prove the limit distributions, we first found limit distributions, which were shown in Eqs.~\eqref{eq:limit_QWL_0}, \eqref{eq:limit_QWL_1}, and \eqref{eq:limit_QWL_prob}, for the quantum walk on the line.
Then, combining those equations and Lemma \ref{lem:160612_03}, we approached the limit distributions of the quantum walk on the half line.
The limit distributions had a compact support dictated by only one of the two unitary operations.
Most remarkably, they were totally independent from the parameters $\alpha$ and $\beta$ producing the localized initial state of the quantum walk.
Although we took care of a special type of unitary operations in Eqs.~\eqref{eq:coin-flip_operator_1} and \eqref{eq:coin-flip_operator_2}, it would be a future task to struggle with a time-dependent quantum walk on the half line defined by a general form of unitary operations.

\begin{center}
{\bf Acknowledgements}
\end{center}
The author is supported by JSPS Grant-in-Aid for Scientific Research (C) (No. 19K03625).



\begin{thebibliography}{1}
\providecommand{\url}[1]{{#1}}
\providecommand{\urlprefix}{URL }
\expandafter\ifx\csname urlstyle\endcsname\relax
  \providecommand{\doi}[1]{DOI \discretionary{}{}{}#1}\else
  \providecommand{\doi}{DOI \discretionary{}{}{}\begingroup
  \urlstyle{rm}\Url}\fi

\bibitem{Venegas-Andraca2012}
S.E. Venegas-Andraca
\newblock (2012), \textit{Quantum walks: a comprehensive review}, Quantum
  Information Processing, 11(5), pp.  1015--1106.

\bibitem{KonnoSegawa2011}
N. Konno and E. Segawa
\newblock (2011), \textit{Localization of discrete-time quantum walks on a half
  line via the CGMV method}, Quantum Information and Computation, 11(5\&6), pp.
   485--495.

\bibitem{LiuPetulante2013}
C. Liu and N. Petulante
\newblock (2013), \textit{Weak limits for quantum walks on the half-line},
  International Journal of Quantum Information, 11(06),  1350054.

\bibitem{Machida2016}
T. Machida
\newblock (2016), \textit{A quantum walk on the half line with a particular
  initial state}, Quantum Information Processing, 15(8), pp.  3101--3119.

\bibitem{RibeiroMilmanMosseri2004}
P. Ribeiro, P. Milman, R. Mosseri
\newblock (2004),
\textit{Aperiodic quantum random walks},
Phys. Rev. Lett., 93(19), 190503.

\bibitem{BanulsNavarretePerezRoldanSoriano2006}
M.C. Ba{\~n}uls, C. Navarrete, A. P{\'e}rez, E. Rold{\'a}n, J.C. Soriano
\newblock (2006),
\textit{Quantum walk with a time-dependent coin},
Phys. Rev. A, 73(6), 062304.

\bibitem{Romanelli2009}
A. Romanelli
\newblock (2009),
\textit{The Fibonacci quantum walk and its classical trace map},
Physica A: Statistical Mechanics and its Applications, 388(18), pp.  3985--3990.

\bibitem{MachidaKonno2010}
T. Machida and N. Konno
\newblock (2010), \textit{Limit theorem for a time-dependent coined quantum
  walk on the line}, F. Peper et al. (Eds.): IWNC 2009, Proceedings in
  Information and Communications Technology, 2, pp.  226--235.

\bibitem{Machida2011}
T. Machida
\newblock (2011),
\textit{Limit theorems for a localization model of 2-state quantum walks},
International Journal of Quantum Information, 9(3), pp.  863--874.

\bibitem{IdeKonnoMachidaSegawa2011}
Y. Ide, N. Konno, T. Machida, E. Segawa
\newblock (2011),
\textit{Return probability of one-dimensional discrete-time quantum walks with final-time dependence},
Quantum Information and Computation, 11(9\&10), pp.  761--773.

\bibitem{Machida2013b}
T. Machida
\newblock (2013),
\textit{Limit distribution with a combination of density functions for a 2-state quantum walk},
Journal of Computational and Theoretical Nanoscience, 10(7), pp.  1571--1578.

\bibitem{GrunbaumMachida2015}
F.A. Gr{\"u}nbaum and T. Machida
\newblock (2015), \textit{A limit theorem for a 3-period time-dependent quantum
  walk}, Quantum Information and Computation, 15(1\& 2), pp.  50--60.

\bibitem{Konno2002a}
N. Konno
\newblock (2002), \textit{Quantum random walks in one dimension}, Quantum
  Information Processing, 1(5), pp.  345--354.

\end{thebibliography}
\end{document}